\documentclass[12pt,a4paper]{article}
\usepackage[T1]{fontenc}
\usepackage[latin1]{inputenc}
\usepackage{color}

\usepackage{booktabs}

\usepackage{graphicx}

\usepackage{amsfonts}
\usepackage{amsmath}
\usepackage{amssymb,stmaryrd}
\usepackage{theorem}
\usepackage{mathtools}

\theoremstyle{change}
\theorembodyfont{\slshape}
\newtheorem{satz}{Theorem}[section]
\newtheorem{thm}[satz]{Theorem}

\newtheorem{prop}[satz]{Proposition}

\theorembodyfont{\upshape\small}
\newtheorem{bsp}[satz]{Example}
\newtheorem{bem}[satz]{Remark}

\theorembodyfont{\ttfamily}

\newenvironment{proof}{\list{}{\itemindent-\leftmargin}%
                \item\textbf{Proof: }\small}{\hbox{}\hfill\#\newline\endlist\vspace{-2mm}}

\numberwithin{equation}{section}
\newcommand{\ba}{\begin{equation}}
\newcommand{\ea}{\end{equation}}

\newcommand{\0}{\mbox{\boldmath $0$}}

\newcommand{\mG}{\mbox{\textup{\textbf{G}}}}
\newcommand{\mJ}{\mbox{\textup{\textbf{J}}}}

\newcommand{\ftheta}{\mbox{\boldmath $\theta$}}
\newcommand{\fTheta}{\mbox{\boldmath $\Theta$}}
\newcommand{\fthetai}{\mbox{\scriptsize\boldmath $\theta$}}
\newcommand{\fThetai}{\mbox{\scriptsize\boldmath $\Theta$}}

\newcommand{\bin}{\textup{Bin}}
\newcommand{\poi}{\textup{Poi}}

\newcommand{\zip}{\textup{ZIP}}
\newcommand{\nb}{\textup{NB}}
\newcommand{\norm}{\textup{N}}
\newcommand{\pgf}{\textup{pgf}}

\newcommand{\bbn}{\mathbb{N}}

\newcommand{\bbz}{\mathbb{Z}}

\newcommand{\iid}{i.\,i.\,d.}
\newcommand{\ie}{i.\,e., }
\newcommand{\eg}{e.\,g., }

\usepackage{dsfont}
\newcommand{\indfkt}{\mathds{1}}

\usepackage{natbib}

\usepackage{hyperref}

\allowdisplaybreaks
\usepackage{indentfirst}
\begin{document}



\parindent 0.6cm

\title{Conditional-mean Multiplicative Operator Models for Count Time Series}
\author{
Christian H.\ Wei\ss{}\thanks{
Helmut Schmidt University, Department of Mathematics and Statistics, Hamburg, Germany. E-Mail: \href{mailto:weissc@hsu-hh.de}{\nolinkurl{weissc@hsu-hh.de}}. ORCID: \href{https://orcid.org/0000-0001-8739-6631}{\nolinkurl{0000-0001-8739-6631}}.}
\and
Fukang Zhu\thanks{
School of Mathematics, Jilin University, 2699 Qianjin, Changchun 130012, China. E-Mail: \href{mailto:zfk8010@163.com}{\nolinkurl{zfk8010@163.com}}. ORCID: \href{https://orcid.org/0000-0002-8808-8179}{\nolinkurl{0000-0002-8808-8179}}.}\ \thanks{Corresponding author.}
}

\maketitle

\begin{abstract}
\noindent
Multiplicative error models (MEMs) are commonly used for real-valued time series, but they cannot be applied to discrete-valued count time series as the involved multiplication would not preserve the integer nature of the data. Thus, the concept of a multiplicative operator for counts is proposed (as well as several specific instances thereof), which are then used to develop a kind of MEMs for count time series (CMEMs). If equipped with a linear conditional mean, the resulting CMEMs are closely related to the class of so-called integer-valued generalized autoregressive conditional heteroscedasticity (INGARCH) models and might be used as a semi-parametric extension thereof. Important stochastic properties of different types of INGARCH-CMEM as well as relevant estimation approaches are derived, namely types of quasi-maximum likelihood and weighted least squares estimation. The performance and application are demonstrated with simulations as well as with two real-world data examples.

\medskip
\noindent
\textsc{Key words:}
Count time series; INGARCH models; multiplicative error model; multiplicative operator; semi-parametric estimation.
\end{abstract}

\section{Introduction}
\label{Introduction}
Since their introduction by \citet{engle02}, the \emph{multiplicative error models} (MEMs) received considerable research interest in the literature, see \citet{brownlees12}, \citet{cipollini22}, and the references therein. MEMs are defined for a positively real-valued process $(Y_t)_{t\in\bbz=\{\ldots,-1,0,1,\ldots\}}$; let $\mathcal F_{s}$ denote the information being available up to time~$s$. Then, $(Y_t)$ follows a MEM if it satisfies the equation
\ba
\label{MEMdef}
Y_t\ =\ \mu_t\cdot\varepsilon_t,
\ea
where $\mu_t|\mathcal F_{t-1}$ is a deterministic, truly positive number, and where~$\varepsilon_t$ is a positively real-valued random variable that is generated independently of~$\mathcal F_{t-1}$ with (conditional) mean $E[\varepsilon_t\ |\ \mathcal F_{t-1}]=E[\varepsilon_t]=1$ and variance $\sigma_t^2 := V[\varepsilon_t\ |\ \mathcal F_{t-1}]= V[\varepsilon_t]\in (0,\infty)$. The observations $(Y_t)$ of the MEM satisfy
\ba
\label{MEMmom}
E[Y_t\ |\ \mathcal F_{t-1}]\ =\ \mu_t,\qquad
V[Y_t\ |\ \mathcal F_{t-1}]=\sigma_t^2\,\mu_t^2.
\ea
If $(\varepsilon_t)$  are \iid\ with mean~1 and variance~$\sigma^2$, and if the conditional means $(\mu_t)$ satisfy the equation
\ba
\label{ACDdef}
\mu_t\ = a_0+\sum_{i=1}^p a_i\, Y_{t-i}+\sum_{j=1}^q b_j\, \mu_{t-j},
\ea
then the resulting MEM is said to be an autoregressive conditional duration (ACD) model of order $(p,q)$, see \citet{engle98}. Here, the constraints $a_0>0$, $a_1,\ldots,a_{p},b_1,\ldots,b_{q}\geq 0$, and $\sum_{i=1}^p a_i+\sum_{j=1}^q b_j < 1$ have to be satisfied to ensure the existence of a truly positive conditional mean.
The construction \eqref{MEMdef} is also closely related to that of the generalized autoregressive conditional heteroscedasticity (GARCH) models \citep{bollerslev86}.

\medskip
The conditional-mean equation \eqref{ACDdef} of the ACD model gave rise to \citet{heinen03} to define an integer-valued counterpart for a count process $(X_t)$, \ie where the~$X_t$ are (unbounded) count random variables with range $\bbn_0=\{0,1,\ldots\}$. With $M_t=E[X_t\ |\ \mathcal{F}_{t-1}]$ denoting the conditional mean (thus, now, $\mathcal{F}_{t-1}$ is the $\sigma$-field generated by $\{(X_{t-1},M_{t-1}),(X_{t-2},M_{t-2}),\ldots\}$), \citet{heinen03} considered the model equation
\ba
\label{INGARCHmodels}
M_t\ =\ a_0+\sum_{i=1}^p a_i\, X_{t-i}+\sum_{j=1}^q b_j\, M_{t-j},
\ea
where~$a_0>0,a_1,\ldots,a_p\geq0,b_1,\ldots,b_q\geq0$, and
$X_t|\mathcal{F}_{t-1}$ is generated by a conditional count distribution having mean~$M_t$, such as the Poisson (Poi) distribution $\poi(M_t)$. \citet{heinen03} referred to this model as the autoregressive conditional Poisson (ACP) model, while \citet{ferland06} used the term integer-valued GARCH (INGARCH) model. Although there is still a debate about the appropriate terminology for models of the type \eqref{INGARCHmodels}, see Remark~4.1.2 in \citet{weiss18}, it seems that the name INGARCH model is most often used in the literature. Thus, in the sequel, we shall also refer to count models satisfying \eqref{INGARCHmodels} as INGARCH-type models.

\medskip
In this article, our aim is to define MEM-type models for count time series, but not starting at \eqref{ACDdef} as done by \citet{heinen03}, but one step earlier right with the original definition \eqref{MEMdef} of a MEM. But since the multiplication $\alpha\cdot\epsilon$ of a count random variable~$\epsilon$ by a positive real number~$\alpha\in(0,\infty)$ does usually not lead to a count value anymore \citep[``multiplication problem'', see][p.~16]{weiss18}, an integer-valued substitute of the multiplication is required. This multiplication problem is well known from the adaption of the autoregressive moving-average (ARMA) model to the count-data case. The so-called integer-valued ARMA (INARMA) model modifies the ordinary ARMA recursion by substituting its multiplications by so-called ``thinning operators'' (plus several further assumptions, see \citet{weiss18} for details).
During the last decades, a huge number of different types of thinning operators have been proposed in the literature, see \citet{weiss08} and \citet{scotto15} for surveys. Generally, a thinning operator ``$\textup{thinn}$'' is a random operator from~$\bbn_0$ to~$\bbn_0$,
whose conditional mean satisfies $E[\textup{thinn}(\epsilon)\ |\ \epsilon] \leq \epsilon$. Many common thinning operators are defined in an additive way and fall within the family of \emph{generalized thinning} operators proposed by \citet{latour98}:
\ba
\label{genthinning}
\alpha\bullet_{\nu} \epsilon\ :=\ \sum_{j=1}^{\epsilon}\ Z_j,
\ea
where the~$Z_j$ are \iid\ count random variables (\emph{counting series}, generated independently of~$\epsilon$) with mean~$\alpha\in (0,1)$ and variance~$\nu>0$. As a result,
$E[\alpha\bullet_{\nu} \epsilon\ |\ \epsilon] = \alpha\cdot\epsilon\ \leq \epsilon$, which justifies its use as an integer-valued substitute of the multiplication.

\begin{bsp}
\label{examThinnings}
Popular examples of generalized thinning are

\begin{itemize}
	\item[(i)] \emph{binomial thinning} dating back to \citet{steutel79} (and commonly denoted by ``$\alpha\circ$''), where the~$Z_j$ are Bernoulli random variables, $Z_j\sim\bin(1,\alpha)$ (thus $\nu=\nu(\alpha)=\alpha(1-\alpha)$), such that $\alpha\circ \epsilon|\epsilon\ \sim\bin(\epsilon,\alpha)$ with range $\{0,\ldots,\epsilon\}$;
	
	\item[(ii)] \emph{negative-binomial (NB) thinning} proposed by \citet{ristic09} (and commonly denoted by ``$\alpha *$''), where the $Z_j$ are geometrically distributed with parameter $\pi:=1/(1+\alpha)$, $Z_j\sim\nb(1,\pi)$ (hence $E[Z_j]=\alpha$ and $\nu(\alpha)=\alpha(1+\alpha)$), such that $\alpha * \epsilon|\epsilon\ \sim \nb\big(\epsilon, 1/(1+\alpha)\big)$ with range $\bbn_0$;
	
	\item[(iii)] \emph{Poisson thinning} \citep{scotto15} with $Z_j\sim\poi(\alpha)$ (thus $\nu(\alpha)=\alpha$) such that $\alpha\bullet_{\nu} \epsilon|\epsilon\ \sim \poi(\epsilon\cdot\alpha)$ with range $\bbn_0$.
\end{itemize}
\end{bsp}
In this article, we propose multiplicative operators for counts, which arise, for example, by generalizing operator \eqref{genthinning}. Section~\ref{Multiplicative Operators for Counts} introduces new multiplicative operators and discusses their basic properties.
Using such multiplicative operator, we adapt the MEM \eqref{MEMdef} to count time series, leading to a \emph{count MEM} (CMEM). In particular, we propose CMEMs under the framework of INGARCH model, referred to as INGARCH-CMEMs.
Two specific cases are studied in Section~\ref{MEMs for Count Time Series}, namely CMEMs based on the compounding operator and the binomial multiplicative operator, respectively. Stationarity and ergodicity of the new models are established, and some probabilistic properties are also discussed. Note that our novel INGARCH-CMEMs differ from the thinning-based MEM adaption recently proposed by \citet{aknouche22}, see Remark~\ref{remCMEMaknouche} for details.
In Section~\ref{Parameter Estimation}, two kinds of semi-parametric estimation methods are considered: quasi-maximum likelihood estimation (QMLE) and two-stage weighted least squares estimation (2SWLSE).
The finite-sample performance of these methods is investigated by simulations in Section~\ref{Simulations}.
Section~\ref{Real-World Data Examples} then illustrates the application of our novel CMEMs to two real-world data examples, namely to count time series referring to infectious diseases and financial transactions, respectively.
Finally, Section~\ref{Conclusions} concludes the article and outlines topics for future research.

\section{Multiplicative Operators for Counts}
\label{Multiplicative Operators for Counts}
For being able to adapt \eqref{MEMdef} in full generality, the restriction $E[\textup{thinn}(\epsilon)\ |\ \epsilon] \leq \epsilon$ being inherent to any thinning operator is too limiting, also the case ``$>\epsilon$'' (\ie also a ``thickening'') should be permitted.
Thus, any operator ``$\alpha\odot: \bbn_0\to\bbn_0$'' satisfying
\ba
\label{mult_op}
\textstyle
E[\alpha\odot \epsilon\ |\ \epsilon]\ =\ \alpha\cdot\epsilon
\qquad\text{for all } \alpha>0
\ea
is referred to as an integer-valued \emph{multiplicative operator}.
Note that the multiplicative property refers to the conditional mean of~$\alpha\odot \epsilon$, but to keep it simple, we use the terminology ``multiplicative operator'' for ``$\alpha\odot$''.
By \eqref{mult_op}, also $E[\alpha\odot \epsilon] = \alpha\, E[\epsilon]$ is satisfied. To compute the (conditional) variance of $\alpha\odot \epsilon$, further assumptions on ``$\alpha\odot$'' are necessary. In any of the subsequent special cases, also the conditional variance of $\alpha\odot \epsilon | \epsilon$ is proportional to~$\epsilon$, with the constant of proportionality given by some $\nu=\nu(\alpha)>0$. In this special case, we immediately obtain
\ba
\label{mult_op_var}
V[\alpha\odot \epsilon\ |\ \epsilon]\ =\ \nu\cdot\epsilon,
\qquad
V[\alpha\odot \epsilon]
\ =\ \nu\, E[\epsilon] + \alpha^2\,V[\epsilon].
\ea
While moment calculations can be done solely based on \eqref{mult_op} and \eqref{mult_op_var}, likelihood calculations or the simulation of ``$\alpha\odot$'' require to fully specify the conditional distribution of $\alpha\odot \epsilon | \epsilon$. This is done in the subsequent special cases.

\medskip
The requirements imposed by \eqref{mult_op} are easily achieved with an operator of the form \eqref{genthinning}, namely by requiring that~$\alpha$ can take any value in $(0,\infty)$. To not confuse the terminology, we refer to
\ba
\label{comp_op}
\begin{array}{@{}l}
\alpha\bullet_{\nu} \epsilon\ :=\ \sum_{j=1}^{\epsilon}\ Z_j\qquad \text{with \iid\ counting series } (Z_j),\\[1ex]
 \text{which has arbitrary mean~$\alpha>0$ and variance~$\nu>0$},
\end{array}
\ea
as a \emph{compounding operator}.
Important (conditional) moments of ``$\alpha\bullet_{\nu}$'' are
\ba
\label{CompoundingCMom}
\begin{array}{@{}l@{\qquad}l}
E[\alpha\bullet_{\nu} \epsilon\ |\ \epsilon]\ =\ \alpha\cdot\epsilon,
&
V[\alpha\bullet_{\nu} \epsilon\ |\ \epsilon]\ =\ \nu\cdot\epsilon,
\qquad
E\big[(\alpha\bullet_{\nu} \epsilon)^2\ \big|\ \epsilon\big]\ =\ \nu\,\epsilon + \alpha^2\,\epsilon^2,
\\[1ex]
E[\alpha\bullet_{\nu} \epsilon]\ =\ \alpha\, E[\epsilon],
&
V[\alpha\bullet_{\nu} \epsilon]\ =\ \nu\, E[\epsilon] + \alpha^2\,V[\epsilon],
\end{array}
\ea
see \eqref{mult_op_var}. Also the probability generating function (pgf) of \eqref{comp_op} is easily expressed. Denoting the pgf of the counting series by $\pgf_Z(u) := E[u^Z]$, and that of~$\epsilon$ by $\pgf_\epsilon(u)$, it holds that
\ba
\label{CompoundingPgf}
\pgf_{\alpha\bullet_{\nu} \epsilon}(u)\ =\
\pgf_{\epsilon}\big(\pgf_Z(u)\big).
\ea
Thus, the distribution of~$\alpha\bullet_{\nu} \epsilon$ can be expressed as the convolution
\ba
\label{CompoundingPmf}
P(\alpha\bullet_{\nu} \epsilon=k)\ =\ \sum_{l=0}^\infty P(\epsilon=l)\cdot P(\alpha\bullet_{\nu} l=k).
\ea
Obviously, a compounding operator with Poisson or geometric counting series (recall Example~\ref{examThinnings}) is still well-defined, with
$$
\alpha\bullet_{\nu} l\ \sim\ \left\{\begin{array}{ll}
\poi(l\,\alpha) & \text{if } Z_j\sim\poi(\alpha),\\[.5ex]
\nb\big(l, 1/(1+\alpha)\big) & \text{if } Z_j\sim\nb\big(1, 1/(1+\alpha)\big).\\
\end{array}\right.
$$
However, a Bernoulli counting series with $\alpha>1$ does not exist, and the extension of the binomial thinning operator ``$\alpha \circ$'' to a full multiplicative operator is not obvious. In what follows, we propose a generalized definition of the ``$\alpha \circ$'' operator, which contains both thinning (if $\alpha<1$) and thickening (if $\alpha>1$):
\ba
\label{bin_op}
\alpha\otimes \epsilon =\ \lfloor\alpha\rfloor\cdot \epsilon\ +\ \bin\big(\epsilon,\ \alpha-\lfloor\alpha\rfloor\big).
\ea
Here, $\lfloor\alpha\rfloor$ denotes the greatest integer less than or equal to $\alpha$, $\bin\big(\epsilon,\ \alpha-\lfloor\alpha\rfloor\big)$ expresses a count random variable being generated according to this (conditional) distribution, and the two terms in the above addition are conditionally independent.
The operator $\otimes$ in \eqref{bin_op} is referred to as the \emph{binomial multiplicative operator}, and it differs from the (purely deterministic and non-multiplicative) rounding operator of \citet{kachour09}.
Note that if $\alpha>1$, it is guaranteed that $\alpha\otimes \epsilon$ gives a value $\geq \epsilon$, which differs from the compounding operator $\alpha\bullet_{\nu}$ with $\alpha>1$. $\alpha=1$ corresponds to a (non-random) preservation of~$\epsilon$.
Further stochastic properties of the binomial multiplicative operator are as follows.

\begin{prop}
\label{prop_bin_mult}
Operator \eqref{bin_op} has the conditional range $\big\{\lfloor\alpha\rfloor\, \epsilon, \ldots, (\lfloor\alpha\rfloor+1)\, \epsilon\big\}$, the conditional moments
$$
E[\alpha\otimes \epsilon\ |\ \epsilon]\ =\ \alpha\cdot \epsilon,
\qquad
V[\alpha\otimes \epsilon\ |\ \epsilon]\ =\ (\alpha-\lfloor\alpha\rfloor)(1-\alpha+\lfloor\alpha\rfloor)\,\epsilon\ =: \nu(\alpha)\,\epsilon,
$$
the unconditional moments
$$
E[\alpha\otimes \epsilon]
\ =\ \alpha\cdot E[\epsilon],
\qquad
V[\alpha\otimes \epsilon]\ =\
\alpha^2\cdot V[\epsilon]\ +\ (\alpha-\lfloor\alpha\rfloor)(1-\alpha+\lfloor\alpha\rfloor)\cdot E[\epsilon],
$$
and the pgf
$$
\pgf_{\alpha\otimes \epsilon}(u)\ =\
\pgf_{\epsilon}\Big(u^{\lfloor\alpha\rfloor}\, \big(1-(\alpha-\lfloor\alpha\rfloor)+(\alpha-\lfloor\alpha\rfloor)\,u\big)\Big).
$$
\end{prop}
\begin{proof}
The conditional mean of the operator is
\begin{eqnarray*}
E[\alpha\otimes \epsilon\ |\ \epsilon] &=&
E\big[\lfloor\alpha\rfloor\, \epsilon\ |\ \epsilon]\ +\ E\big[\bin\big(\epsilon,\ \alpha-\lfloor\alpha\rfloor\big)\ |\ \epsilon\big]
\\
&=&
\lfloor\alpha\rfloor\, \epsilon\ +\ (\alpha-\lfloor\alpha\rfloor)\,\epsilon
\ =\ \alpha\cdot \epsilon.
\end{eqnarray*}
The conditional variance is
\begin{eqnarray*}
V[\alpha\otimes \epsilon\ |\ \epsilon] &=&
V\big[\underbrace{\lfloor\alpha\rfloor\, \epsilon}_{\text{(*)}}\ +\ \bin\big(\epsilon,\ \alpha-\lfloor\alpha\rfloor\big)\ |\ \epsilon\big]
\\
&=&
V\big[\bin\big(\epsilon,\ \alpha-\lfloor\alpha\rfloor\big)\ |\ \epsilon\big]
\ =\ (\alpha-\lfloor\alpha\rfloor)(1-\alpha+\lfloor\alpha\rfloor)\,\epsilon,
\end{eqnarray*}
where at (*), we used that $\lfloor\alpha\rfloor\, \epsilon \big| \epsilon$ is conditionally constant. So the unconditional moments immediately follow from \eqref{mult_op_var}.
The pgf is
\begin{eqnarray*}
\pgf_{\alpha\otimes \epsilon}(u) &=&
E\big[u^{\alpha\otimes \epsilon}\big]\ =\
E\Big[E\big[u^{\lfloor\alpha\rfloor\cdot \epsilon\ +\ \bin(\epsilon,\ \alpha-\lfloor\alpha\rfloor)}\ \big|\ \epsilon\big]\Big]
\\
&=&
E\Big[u^{\lfloor\alpha\rfloor\cdot \epsilon}\cdot E\big[u^{\bin(\epsilon,\ \alpha-\lfloor\alpha\rfloor)}\ \big|\ \epsilon\big]\Big]
\\
&=&
E\Big[u^{\lfloor\alpha\rfloor\cdot \epsilon}\cdot \big(1-(\alpha-\lfloor\alpha\rfloor)+(\alpha-\lfloor\alpha\rfloor)\,u\big)^\epsilon\Big]
\\
&=&
\pgf_{\epsilon}\Big(u^{\lfloor\alpha\rfloor}\, \big(1-(\alpha-\lfloor\alpha\rfloor)+(\alpha-\lfloor\alpha\rfloor)\,u\big)\Big),
\end{eqnarray*}
which completes the proof of Proposition~\ref{prop_bin_mult}.
\end{proof}
While our main interest is in CMEMs based on the compounding operator \eqref{comp_op} and on the binomial multiplicative operator \eqref{bin_op}, respectively, let us conclude Section~\ref{Multiplicative Operators for Counts} by pointing out that multiplicative operators do not need to be of such an additive structure. Actually, any count distribution with closed-form mean might give rise to a multiplicative operator. For illustration, let us take the zero-inflated Poisson (ZIP) distribution $\zip(\lambda,\omega)$ with $\lambda>0$ and $\omega\in (0,1)$ as the starting point. Its mean~$\mu$ and variance~$\sigma^2$ are given by the relations $\mu=(1-\omega)\,\lambda$ and $\sigma^2/\mu = \lambda-\mu+1$. Thus, a ZIP-multiplicative operator ``$\alpha\varoast_\kappa$'' with $\alpha>0$ and $\kappa>1$ might be defined by requiring that $\alpha\varoast_\kappa \epsilon|\epsilon$ is conditionally ZIP-distributed with mean $E[\alpha\varoast_\kappa \epsilon\,|\,\epsilon] = \alpha\,\epsilon$ and variance $V[\alpha\varoast_\kappa \epsilon\,|\,\epsilon] = \kappa\cdot E[\alpha\varoast_\kappa \epsilon\,|\,\epsilon]$, \ie with conditional ZIP-parameters $\lambda=\alpha\,\epsilon+\kappa-1$ and $\omega=(\kappa-1)/(\alpha\,\epsilon+\kappa-1)$. Note that \eqref{mult_op_var} also applies to ``$\alpha\varoast_\kappa$''.

\section{MEMs for Count Time Series}
\label{MEMs for Count Time Series}
Inspired by Equation \eqref{MEMdef}, a CMEM for a process $(X_t)$ consisting of unbounded counts is defined by
\ba
\label{CountMEMdef}
X_t\ =\ M_t\odot\epsilon_t,
\ea
where $M_t|\mathcal F_{t-1}$ is a deterministic truly positive number. $(\epsilon_t)$ is \iid, and $\epsilon_t$ is a count random variable that is generated independently of~$\mathcal F_{t-1}$ with  mean $E[\epsilon_t]=1$ and variance $V[\epsilon_t]=\sigma^2$. The multiplicative operator ``$\odot$'' defined in \eqref{mult_op} is executed independently of~$\mathcal F_{t-1}$ and~$\epsilon_t$.

\smallskip
With these assumptions, the conditional mean of \eqref{CountMEMdef} satisfies
\ba
\label{CountMEMcmean}
\begin{array}{@{}rl}
E[X_t\ |\ \mathcal F_{t-1}]
\ =&
E\big[E[M_t\odot\epsilon_t\ |\ \epsilon_t, \mathcal F_{t-1}]\ \big|\ \mathcal F_{t-1}\big]
\\[1ex]
=&
E[M_t\cdot\epsilon_t\ |\ \mathcal F_{t-1}]\ =\ M_t\cdot 1\ =\ M_t.
\end{array}
\ea
Thus, if the conditional mean~$M_t$ follows a recursive scheme of the form
\ba
\label{INGARCHgeneral}
M_t\ =\ f(X_{t-1},\ldots, X_{t-p},\ M_{t-1},\ldots, M_{t-q})
\quad\text{with } p\geq 1,\ q\geq 0,
\ea
then the CMEM becomes an INGARCH-type model of order $(p,q)$, denoted as \emph{INGARCH$(p,q)$-CMEM}.
The most well-known special case of \eqref{INGARCHgeneral} is the linear specification \eqref{INGARCHmodels}.
But also non-linear mean specifications are possible for~$M_t$, \eg by adapting the log-linear structure discussed by \citet{fokianos11} or the approximately linear softplus structure proposed by \citet{weissetal22}, both allowing for negative model parameters (also see Section~\ref{Simulations}).

\begin{bem}
\label{bemStochEq}
One of the formal advantages of the proposed model is that it is defined through a stochastic equation and not by a conditional distribution. In particular, the standard Poi-INGARCH model of \citet{ferland06} is nested in this model, see Section~\ref{Count-MEMs using a Compounding Operator} for details, and is thus defined by an equation instead of a conditional distribution.
\end{bem}

\begin{bem}
\label{bemMixtures}
We have already mentioned some specific operators before. An important class of operators defined by (2.1) is given by a hierarchical mixture of INGARCHs. We already indicated that the negative binomial INGARCH is a specific case of our model. Indeed, the mixed Poisson autoregression $\{Y_t,t\in\mathbb{Z}\}$ given by
$$
Y_t=N_t(\varepsilon_t M_t), ~~M_t=h(Y_{t-1},M_{t-1}),~~t\in\mathbb{Z},
$$
is a specific case of \eqref{CountMEMdef}, where $(N_t)$ is an independent sequence of homogeneous Poisson processes with unit intensity. More generally, any hierarchical mixture process given by
\begin{align*}
  Y_t|\mathcal{F}_{t-1},\varepsilon_t&\sim F_{M_t,\varepsilon_t},\\
  \varepsilon_t&\sim F_{\varepsilon_t},
\end{align*}
($F_{M_t,\varepsilon_t}$ and $F_{\varepsilon_t}$ being given distributions)  is a particular instance of model (2.1). So
$$
f_{Y_t|\mathcal{F}_{t-1}}(y)=\sum_{\varepsilon}f_{M_t,\varepsilon_t}(y,\varepsilon)f_{\varepsilon_t}(\varepsilon)
$$
is a finite or infinite mixture of distributions over $\varepsilon_t$ depending on whether the range of $\varepsilon_t$ is finite or infinite. In this sense, the proposed model \eqref{CountMEMdef} is deeply related to a mixture of INGARCH representations, see \citet{zhuetal10}, \citet{mao20}, \citet{doukhan21}, among others.
\end{bem}

\smallskip
In the special case, where the multiplicative operator ``$\odot$'' also satisfies $V[\alpha\odot \epsilon\, |\, \epsilon] = \nu(\alpha)\cdot\epsilon$, we conclude from \eqref{mult_op_var} that
\ba
\label{CountMEMcvar}
\begin{array}{@{}rl}
V[X_t\ |\ \mathcal F_{t-1}]
\ =&
E\big[V[M_t\odot\epsilon_t\ |\ \epsilon_t, \mathcal F_{t-1}]\ \big|\ \mathcal F_{t-1}\big]
 +
V\big[E[M_t\odot\epsilon_t\ |\ \epsilon_t, \mathcal F_{t-1}]\ \big|\ \mathcal F_{t-1}\big]
\\[1ex]
=&
E[\nu(M_t)\cdot\epsilon_t\ |\ \mathcal F_{t-1}] + V[M_t\cdot\epsilon_t\ |\ \mathcal F_{t-1}]\ =\  \nu(M_t) + \sigma^2 M_t^2.
\end{array}
\ea
So the conditional variance is caused by two sources: the variance of the operator ``$\odot$'' and the one of the innovation term~$\epsilon_t$.

\smallskip
To establish the existence of an INGARCH$(p,q)$-CMEM, one can benefit from the existing INGARCH literature. To apply Theorem~2.1 of \citet{doukhan19}, for example, one has to ensure that on the one hand, assumptions (A$1'$) and (A2) of \citet{doukhan19} hold for the given INGARCH specification \eqref{INGARCHgeneral}, and on the other hand, certain distributional requirements for the multiplicative operator ``$\odot$'' have to be satisfied. More precisely, if also the conditional distribution of $X_t|\mathcal{F}_{t-1} = M_t\odot\epsilon_t|M_t =: Q(M_t)$ satisfies (A3) of \citet{doukhan19}, then their Theorem~2.1 can be used to conclude on the existence of a stationary solution to \eqref{CountMEMdef} and \eqref{INGARCHgeneral}, as well as on mixing properties.
This, certainly, requires to first fix a specific choice of the multiplicative operator ``$\odot$''.
In what follows, we discuss two kinds of CMEMs \eqref{CountMEMdef} with the linear specification \eqref{INGARCHmodels}, namely based on the compounding operator \eqref{comp_op} and on the binomial multiplicative operator \eqref{bin_op}, respectively.

\numberwithin{satz}{subsection}

\subsection{CMEMs Using a Compounding Operator}
\label{Count-MEMs using a Compounding Operator}
In this section, let us consider the special case where the CMEM \eqref{CountMEMdef} uses the compounding operator \eqref{comp_op} $\alpha\bullet_{\nu}$ with arbitrary $\alpha,\nu>0$ as the multiplicative operator, \ie
\ba
\label{CompCountMEMdef}
X_t\ =\ M_t\bullet_{\nu} \epsilon_t,
\ea
where $M_t$ is given in \eqref{INGARCHmodels} and $\epsilon_t$ is the same one defined in \eqref{CountMEMdef}.
Note that $X_t$ can be written as $\sum_{j=1}^{\epsilon_t}Z_{tj}$, where $(Z_{tj})$ are count random variables with $E[Z_{tj}\ |\ \mathcal F_{t-1}]=M_t$, and $(\epsilon_t)$ are independent of $(Z_{tj})$. Given $\mathcal F_{t-1}, Z_{sj}$ are conditionally independent for $s\geq t$ and different $j$. Furthermore, if $t$ is fixed, then the $Z_{tj}$ are conditionally \iid\ with some $Z_t$ for different $j$, and more discussions about $Z_t$ will be given later in the context of Theorem~\ref{thm_stat_compCMEM}.

\medskip
The formulae \eqref{CompoundingCMom}--\eqref{CompoundingPgf} are now used to derive properties of the compounding CMEM \eqref{CompCountMEMdef}. While still $E[X_t\ |\ \mathcal F_{t-1}] = M_t$ holds, we get
$$
\begin{array}{@{}rl}
E[X_t^2\ |\ \mathcal F_{t-1}]
\ =&
E\big[E[(M_t\bullet_{\nu}\epsilon_t)^2\ |\ \epsilon_t, \mathcal F_{t-1}]\ \big|\ \mathcal F_{t-1}\big]
\\[1ex]
\overset{\eqref{CompoundingCMom}}{=}&
E\big[\nu(M_t)\,\epsilon_t + M_t^2\,\epsilon_t^2\ \big|\ \mathcal F_{t-1}\big]
\ =\ \nu(M_t) +(\sigma^2+1) M_t^2,
\end{array}
$$
where $\nu=\nu(M_t)$. Thus,
\ba
\label{CompCountMEMcvar}
V[X_t\ |\ \mathcal F_{t-1}]\ =\ \nu(M_t) + \sigma^2 M_t^2,
\ea
also see \eqref{CountMEMcvar}.
In the special case of a Poi-counting series, we have $\nu(M_t) = M_t$, whereas a NB-counting series leads to $\nu(M_t) = M_t(1+M_t)$, recall Example~\ref{examThinnings}.
Compared with an ordinary INGARCH model with conditional mean $M_t$, the main difference in conditional mean and variance between the CMEM \eqref{CompCountMEMdef} and the INGARCH model is as follows: while having the same conditional mean $M_t$, there is an additional term $\sigma^2 M_t^2$ in the conditional variance of the CMEM, which makes the CMEM more flexible for applications in practice. Only in the boundary case where $\epsilon_t=1$ deterministically (thus $\sigma^2=0$), the CMEM specification reduces to the ordinary INGARCH approach.
In particular, the special case $\epsilon_t\equiv 1$ and Poi-counting series just agrees with the ordinary Poi-INGARCH model of \citet{ferland06}.

\begin{bem}
\label{remCMEMaknouche}
\citet{aknouche22} proposed a multiplicative thinning-based INGARCH model defined by
\begin{equation}\label{mul-thin}
X_t=\lambda_t\cdot\epsilon_t,~~~\lambda_t=1+\omega\circ m+\sum_{i=1}^q\alpha_i\circ Y_{t-i}+\sum_{j=1}^p\beta_j\circ\lambda_{t-j},
\end{equation}
where $\epsilon_t$ is given like in \eqref{CountMEMdef}, and where $\circ$ is the binomial thinning operator discussed in Example~\ref{examThinnings}\,(i). Model \eqref{mul-thin} has a similar mean structure as model \eqref{CompCountMEMdef}, and the inclusion of ``$\omega\circ$'' makes model \eqref{mul-thin} analogous to \eqref{CompCountMEMdef}.
But there are also substantial differences between both approaches.
First, the positive integer~$m$ in \eqref{mul-thin} is not treated as a parameter (for practice, \citet{aknouche22} suggest to choose~$m$ as the integer part of the sample mean).
Second, the conditional variance of model \eqref{mul-thin} is a function of two recursive items, see equation (2.11) in \citet{aknouche22}, while the conditional variance of our model \eqref{CountMEMdef} is just a function of $M_t$, see \eqref{CompCountMEMcvar}. The complexity of the conditional variance also affects the estimation method discussed in Section~\ref{Parameter Estimation}.
\end{bem}

The conditional pgf of the compounding CMEM \eqref{CompCountMEMdef} satisfies
\ba
\label{CompCountMEMcpgf}
\pgf_{X_t|\mathcal F_{t-1}}(u)
\ =\
E[u^{M_t\bullet_{\nu}\epsilon_t}\ |\ \mathcal F_{t-1}]
\ \overset{\eqref{CompoundingPgf}}{=}\ \pgf_{\epsilon}\big(\pgf_{Z|\mathcal F_{t-1}}(u)\big),
\ea
where $Z|\mathcal F_{t-1}$ refers to the counting series of the operator ``$M_t\bullet_{\nu}$'' given~$\mathcal F_{t-1}$, recall \eqref{CompCountMEMdef}.

\begin{bsp}
\label{examSpecialCases}
In the special case of Poisson-distributed innovations, $\epsilon_t\sim\poi(1)$, \eqref{CompCountMEMcpgf} implies a conditional compound-Poisson (CP) distribution, namely $\pgf_{X_t|\mathcal F_{t-1}}(u) = \exp\big(\pgf_{Z|\mathcal F_{t-1}}(u)-1\big)$. Thus, if~$M_t$ also satisfies \eqref{INGARCHmodels}, we get a type of CP-INGARCH model as discussed by \citet{goncalves15}.
If, in addition, we have a Poisson-counting series $Z|\mathcal F_{t-1}\sim\poi(M_t)$ with $\pgf_{X_t|\mathcal F_{t-1}}(u) = \exp\big(M_t\,(u-1)\big)$, \eqref{CompCountMEMcpgf} implies $\pgf_{X_t|\mathcal F_{t-1}}(u) = \exp\big(\exp\big(M_t\,(u-1)\big)-1\big)$, \ie we get a Neyman type-A (NTA) INGARCH model as discussed by \citet{goncalves15b}.

\smallskip
Another special case is $\epsilon_t\sim\zip(\frac{1}{1-\omega},\omega)$ with $\omega\in (0,1)$, \ie a ZIP distribution with mean~1 and $\pgf_{\epsilon}(u)=\omega+(1-\omega)\,\exp(\frac{u-1}{1-\omega})$, which leads to a zero-inflated CP-INGARCH model as discussed by \citet{goncalves16}. Generally, there are two possible reasons for the compounding operator $\alpha\bullet_{\nu} \epsilon = \sum_{j=1}^{\epsilon} Z_j$ from \eqref{comp_op} to become zero: either $\epsilon=0$, or $\epsilon>0$ but all $Z_1=\ldots=Z_{\epsilon}=0$, \ie
$$
\textstyle
P(\alpha\bullet_{\nu} \epsilon = 0)\ =\ P(\epsilon=0)\ +\ \sum_{i=1}^\infty P(\epsilon=i)\,P(Z=0)^i.
$$
Thus, if $\epsilon_t\sim\zip(\frac{1}{1-\omega},\omega)$, we get
$$
\begin{array}{@{}rl}
P(X_t=0\ |\ \mathcal F_{t-1})
\ =& \omega\ +\ (1-\omega)\,{\rm e}^{-1/(1-\omega)}\,\sum\limits_{i=0}^\infty \frac{(1-\omega)^{-i}}{i!}\,\pgf_{Z|\mathcal F_{t-1}}(0)^i
\\[1ex]
\ =& \omega\ +\ (1-\omega)\,\exp\Big(\frac{\pgf_{Z|\mathcal F_{t-1}}(0)-1}{1-\omega}\Big),
\end{array}
$$
in accordance with the expression for $\pgf_{X_t|\mathcal F_{t-1}}(0)$.
\end{bsp}
The following proposition gives a necessary and sufficient condition for the first-order stationarity of the
model introduced in \eqref{CompCountMEMdef}, which is analogous to the corresponding property of \citet{ferland06}. The proof of Proposition~\ref{first-order-stationarity} is done with the same arguments used in Proposition~1 of \citet{ferland06}, so we omit the details.

\begin{prop}\label{first-order-stationarity}
The process $(X_t)$ defined in \eqref{CompCountMEMdef} is first-order stationary if and only if $\sum_{i=1}^pa_i+\sum_{j=1}^qb_j<1$.
\end{prop}

As the next step, we discuss the second-order stationarity. We make the following assumption
\begin{equation}\label{var-mean}
\frac{V[X_t\ |\ \mathcal F_{t-1}]}{E[X_t\ |\ \mathcal F_{t-1}]}=v_0+v_1M_t,
\end{equation}
with $v_0\geq0$ and $v_1\geq0$ not being simultaneously zero.  It is obvious that the Poisson- and NB-counting series satisfy the above assumption, recall \eqref{CompCountMEMcvar}, leading to $v_1=\sigma^2$ and $v_1=1+\sigma^2$, respectively.
Then, the results of Section~3.2 in \citet{goncalves15} hold for the model \eqref{CompCountMEMdef}. In particular, we have the following proposition.

\begin{prop}\label{second-order-stationarity}
Consider a first-order stationary INGARCH$(1,1)$-CMEM defined according to \eqref{CompCountMEMdef}, which also satisfies \eqref{var-mean}. A necessary and
sufficient condition for second-order stationarity is $(a_1+b_1)^2+v_1a_1^2<1$.
\end{prop}

For second-order stationarity of the INGARCH$(p,q)$-CMEM model with general order~$(p,q)$, we refer to Theorem~4 in \citet{goncalves15}, which will be not repeated here.

If $\sum_{i=1}^pa_i+\sum_{j=1}^qb_j<1$, we obtain the unconditional mean of $X_t$ as
\ba
\label{INGARCH_mean}
\mu:=E[X_t]=\frac{a_0}{1-\sum_{i=1}^pa_i-\sum_{j=1}^qb_j}.
\ea
Next, we derive the unconditional variance of $X_t$. From Theorem~1 in \citet{weiss09}, we know that
the autocovariances $\gamma_X(k):= {\rm Cov}[X_t,X_{t-k}]$ and $\gamma_M(k):= {\rm Cov}[M_t, M_{t-k} ]$
satisfy the linear equations
\ba
\label{INGARCH_yw}
\begin{array}{@{}l}
\displaystyle
  \gamma_X(k) = \sum_{i=1}^pa_i \gamma_X(|k-i|)+\sum_{j=1}^{\min(k-1,q)}b_j\gamma_X(k-j)+\sum_{j=k}^qb_j\gamma_M(j-k),\\
	\\[-2ex]
\displaystyle
    \gamma_M(k) = \sum_{i=1}^{\min(k,p)}a_i \gamma_M(|k-i|)+\sum_{i=k+1}^pa_i\gamma_X(i-k)+\sum_{j=1}^qb_j\gamma_M(|k-j|),
\end{array}
\ea
for $k\geq1$ and $k\geq0$, respectively. Then, using \eqref{CompCountMEMcvar}, the variance can be obtained from
\ba\label{var-compCMEM}
V[X_t]=E[\nu(M_t) + \sigma^2M_t^2]+V[M_t]
\ =\ E[\nu(M_t)] + \mu^2 \sigma^2+(\sigma^2+1)\,V[M_t],
\ea
where $V[M_t]=\sum_{i=1}^pa_i\gamma_X(i)+\sum_{j=1}^qb_j\gamma_M(j)$.
For the Poisson-counting case, we have $\nu(M_t) = M_t$, thus
\ba\label{var-poi}
V[X_t]=\mu+\mu^2\sigma^2+(\sigma^2+1)\,V[M_t],
\ea
while for the NB-counting case with $\nu(M_t) = M_t(1+M_t)$, we have
\ba\label{var-nb}
V[X_t]=E[M_t+M_t^2] + \mu^2 \sigma^2+(\sigma^2+1)\,V[M_t]=\mu+\mu^2(\sigma^2+1)+(\sigma^2+2)\,V[M_t].
\ea
Let us illustrate the required computations for the case $p=q=1$, which is most common in applications.

\begin{bsp}
\label{examINGARCH11}
If $p=q=1$, we get from \eqref{INGARCH_yw}, after few calculations, that
$$
\big(1-(a_1+b_1)^2 + a_1^2\big)\,V[M_t] = a_1^2\,V[X_t].
$$
This can be used to obtain $V[X_t]$ after substituting $V[M_t]$ into \eqref{var-poi} or \eqref{var-nb}, respectively. The autocorrelation function (ACF) $\rho(k)=Corr[X_t,X_{t-k}]$ takes the usual form \citep[see][]{weiss09}
$$
\rho(k) = (a_1+b_1)^{k-1}\,\frac{a_1\, (1-b_1(a_1+b_1))}{1-(a_1+b_1)^2+a_1^2}.
$$
Together with \eqref{INGARCH_mean}, it can be utilized for deriving method-of-moment estimators.
\end{bsp}

In practice, equations \eqref{INGARCH_yw} can be used together with the sample ACF for the identification of the model order $(p,q)$ as well as for getting initial parameter estimates (method of moments). This is discussed in some more details in Sections~\ref{Approaches for Model Diagnostics} and~\ref{Real-World Data Examples}.

\smallskip
Before discussing the stationary properties of $(X_t)$, we first recall the definition of the CP-INGARCH model given in \citet{goncalves15}.
The count process $(Z_t)$ is said to satisfy a CP-INGARCH$(p,q)$ model if the conditional characteristic function of $Z_t\ |\ \mathcal F_{t-1}$ is $\Phi_{Z_t | \mathcal F_{t-1}}(u)=\exp\{{\rm i}M_t^*[\varphi_t(u)-1]/\varphi'_t(0)\}$ with $u\in\mathbb{R}$ and $t\in\mathbb{Z}$, and if $M_t^*$ is defined like in \eqref{INGARCHmodels}, \ie $M_t^* = a_0+\sum_{i=1}^p a_i Z_{t-i}+\sum_{j=1}^q b_j M_{t-j}^*$.
Here, $(\varphi_t)$ is a family of characteristic functions on $\mathbb{R}$, and ``i'' denotes the imaginary unit.

The following theorem establishes the stationarity and ergodicity of the compounding CMEM $(X_t)$ if the counting series arises from the CP-family. Recall from \eqref{CompCountMEMdef} that $X_t=\sum_{j=1}^{\epsilon_t}Z_{tj}$, where $(Z_{tj})$ is supposed to be identically distributed with some $(Z_t)$ that satisfies $E[Z_t\ |\ \mathcal F_{t-1}] =E[X_t\ |\ \mathcal F_{t-1}] = M_t$.

\begin{thm}\label{thm_stat_compCMEM}
Assume that $(Z_t)$ satisfies a CP-INGARCH process.
If $\varphi_t$ is deterministic and independent of $t$, and if $\sum_{i=1}^pa_i+\sum_{j=1}^qb_j<1$ holds,
then the process $(X_t)$ defined according to the compounding CMEM \eqref{CompCountMEMdef} is stationary and ergodic.
\end{thm}

\begin{proof}
Using arguments in \citet{goncalves15}, we know that $M_t$ can be written as $M_t=a_0/\sum_{j=1}^qb_j+\sum_{k=1}^\infty \psi_kX_{t-k}$, where the coefficient $\psi_k$ is given by the power expansion of a rational function in the neighbourhood of 0.
Based on the construction method for proving stationarity and ergodicity in \citet{ferland06} and \citet{goncalves15},
let $M_t^{(n)}=a_0/\sum_{j=1}^qb_j+\sum_{k=1}^n \psi_kX_{t-k}$, $M_t^{\ast(n)}=a_0/\sum_{j=1}^qb_j+\sum_{k=1}^n \psi_kZ_{t-k}$, $r^{(n)}=E\big|(M_t^{(n)}-M_t^{\ast(n)})-(M_t^{(n-1)}-M_t^{\ast(n-1)})\big|$.
Note that $0\leq r^{(n)}=|\psi_n|\,E|X_{t-n}-Z_{t-n}|=|\psi_n|\,\big(E[(X_{t-n}-Z_{t-n})\,\indfkt(A_n)]+E[(Z_{t-n}-X_{t-n})\,\indfkt(A_n^c)]\big) \leq |\psi_n|\,\big(E(X_{t-n}-Z_{t-n})+E(Z_{t-n}-X_{t-n})\big)=0$,
where $A_n=\{X_{t-n}-Z_{t-n}>0\}$ and $\indfkt(\cdot)$ denotes the indicator function.
Thus, $r^{(n)}=0$. Therefore, $(M_t^{(n)}-M_t^{\ast(n)})$ is a Cauchy sequence and thus $(M_t^{(n)}-M_t^{\ast(n)})$ converges in $L^1$ and a.s., \ie $\lim_{n\rightarrow\infty}(M_t^{(n)}-M_t^{\ast(n)})=0$ a.s.
Using Theorem~5 in \citet{goncalves15}, we know that $(Z_t)$ is stationary and ergodic.

Notice that $X_t=\sum_{j=1}^{\epsilon_t}Z_{tj}$, and recall that $(Z_{tj})$ is supposed to be identically distributed with $(Z_t)$. Then, there exists a real function $\xi$, which does not depend on $t$, such that
$X_t=\xi\big((Z_{tj},\epsilon_t),t\in\mathbb{Z},j\in\mathbb{N}\big)$.
Note that $\epsilon_t$ is \iid\ and independent of $\{Z_{tj}\}$, so the result follows by applying Theorem~36.4 in \citet{bill95}.
\end{proof}

We conclude Section~\ref{Count-MEMs using a Compounding Operator} by pointing out that if the distribution of~$\epsilon_t$ is not further specified,
then the resulting INGARCH-CMEM can be understood as a kind of semi-parametric INGARCH model, in analogy to the semi-parametric INAR models proposed by \citet{drost09} and \citet{liu21}. Then, the conditional distribution of $X_t | \mathcal F_{t-1}$ is a convolution between the distribution of counting series  and the one of~$\epsilon_t$, recall \eqref{CompoundingPmf}. Numerically exact computations are possible if the range of~$\epsilon_t$ is assumed to be bounded, because then, \eqref{CompoundingPmf} becomes a finite sum. Generally, denoting $p_{\epsilon,l}=P(\epsilon_t=l)$, the distribution of~$\epsilon_t$ is completely determined by $p_{\epsilon,2}, p_{\epsilon,3}, \ldots$ with the constraint $\sum_{l\geq 2} l\cdot p_{\epsilon,l} <1$, where $p_{\epsilon,0} = \sum_{l\geq 2} (l-1)\cdot p_{\epsilon,l}$ and $p_{\epsilon,1} = 1-p_{\epsilon,0} - \sum_{l\geq 2} p_{\epsilon,l}$ because of the requirement $E[\epsilon_t]=1$.

\subsection{CMEMs Using a Binomial Multiplicative Operator}
\label{Count-MEMs using a Binomial Multiplicative Operator}
In this section, we consider the special case where the CMEM \eqref{CountMEMdef} uses the binomial multiplicative operator $\otimes$ defined in \eqref{bin_op} as the multiplicative operator, \ie
\ba
\label{CompCountMEMdef2}
X_t\ =\ M_t\otimes \epsilon_t,
\ea
where $M_t$ is given in \eqref{INGARCHmodels}, and where $\epsilon_t$ is defined like in \eqref{CountMEMdef}.
Note that $X_t$ can be expressed as $ \lfloor M_t\rfloor\cdot \epsilon_t\ +\ \bin\big(\epsilon_t,\ M_t-\lfloor M_t\rfloor\big)$. As pointed out in Section \ref{Multiplicative Operators for Counts}, the two terms in the above addition are independent.

Let us derive properties of the CMEM \eqref{CompCountMEMdef2}.
First, we note that the first-order stationarity given in Proposition \ref{first-order-stationarity} still holds here.
Note that $E[X_t\ |\ \mathcal F_{t-1}] = M_t$, and that the variance function $\nu(\alpha) = (\alpha-\lfloor\alpha\rfloor)(1-\alpha+\lfloor\alpha\rfloor)$ according to Proposition~\ref{prop_bin_mult}.
Hence, we have
\begin{align}\label{CompCountMEMcvar2}
E[X_t^2\ |\ \mathcal F_{t-1}]
\ =&\
E\Big[E\big[\big(\lfloor M_t\rfloor\cdot \epsilon_t\ +\ \bin(\epsilon_t,\ M_t-\lfloor M_t\rfloor)\big)^2\ |\ \epsilon_t, \mathcal F_{t-1}\big]\ \big|\ \mathcal F_{t-1}\Big]
\nonumber\\
=&\
E\Big[( M_t-\lfloor M_t\rfloor)(1-\ M_t+\lfloor M_t\rfloor) \epsilon_t + M_t^2\,\epsilon_t^2\ \big|\ \mathcal F_{t-1}\Big]\nonumber\\
\ =&\ ( M_t-\lfloor M_t\rfloor)(1-\ M_t+\lfloor M_t\rfloor)+(\sigma^2+1) M_t^2,\nonumber\\[1ex]
V[X_t\ |\ \mathcal F_{t-1}] =& \ ( M_t-\lfloor M_t\rfloor)(1-\ M_t+\lfloor M_t\rfloor)+\sigma^2 M_t^2,
\end{align}
also see \eqref{CountMEMcvar}.
It is easy to see that $\nu(\alpha)\in [0,0.25]$, so $\sigma^2 M_t^2\leq V[X_t\ |\ \mathcal F_{t-1}] \leq0.25+\sigma^2 M_t^2$ holds.
Similar to arguments in Section \ref{Count-MEMs using a Compounding Operator},
if $\sum_{i=1}^pa_i+\sum_{j=1}^qb_j<1$, we obtain the unconditional mean of $X_t$ as
$$\mu:=E[X_t]=\frac{a_0}{1-\sum_{i=1}^pa_i-\sum_{j=1}^qb_j}.$$
Then, the variance is given by
\begin{align*}
V[X_t]&=E\Big[( M_t-\lfloor M_t\rfloor)(1-\ M_t+\lfloor M_t\rfloor) + \sigma^2M_t^2\Big]+V[M_t]\\
&=E\big[( M_t-\lfloor M_t\rfloor)(1-\ M_t+\lfloor M_t\rfloor)\big] + \mu^2\sigma^2+(\sigma^2+1)\,V[M_t],
\end{align*}
 where $V[M_t]=\sum_{i=1}^pa_i\gamma_X(i)+\sum_{j=1}^qb_j\gamma_M(j)$.
If $p=q=1$, for example, then like in Section~\ref{Count-MEMs using a Compounding Operator}, we can insert $(1-(a_1+b_1)^2 + a_1^2)\,V[M_t] = a_1^2\,V[X_t]$. However, as there is a floor operation in $V[X_t]$, we cannot obtain a closed-form expression for $V[X_t]$ this time. But we have
$\mu^2\sigma^2+(\sigma^2+1)V[M_t]\leq V[X_t]\leq0.25+\mu^2\sigma^2+(\sigma^2+1)V[M_t]$, so we can approximate $V[X_t]$ by using these bounds.

\smallskip
The following theorem establishes the stationarity and ergodicity of $(X_t)$, which can be proven using the technique similar to that in \citet{agost16}.
Note that because of $E[X_t\ |\ \mathcal F_{t-1}] = M_t$, also $E[X_t] = E[M_t]$ holds.

\begin{thm}
If $\sum_{i=1}^pa_i+\sum_{j=1}^qb_j<1$ holds, there exists a weakly dependent stationary and ergodic solution to \eqref{CompCountMEMdef2}.
\end{thm}
\begin{proof}
We first verify that the process $(X_t)$ satisfies the conditions of Theorem 3.1 in \citet{doukhan08} from which the theorem will follow.  To simplify the notation, assume without loss of generality $p = q$ in the following. With $\beta(z)=1-\sum_{i=1}^qb_iz^i,z\in\mathbb{C}$, the assumption $\sum_{i=1}^pa_i+\sum_{j=1}^qb_j<1$ implies that $\beta(z)^{-1}=\psi(z):=\sum_{i=0}^\infty\psi_iz^i$ is well-defined for $|z|\leq 1 +\eta$ with some $\eta> 0$, where $\psi$ is exponentially decreasing and defined recursively by $\psi_0 = 1$ and $\psi_n=\sum_{i=1}^nb_i\psi_{n-i}\geq 0$ for $n\geq 1$. Next, with $\alpha(z)=\sum_{i=1}^qa_iz^{i-1}$, note that $M_t$ defined in \eqref{INGARCHmodels} can be rewritten in terms of the backshift operator $B$ as $\beta(B)M_t=a_0+\alpha(B)X_{t-1}$. Then,
$$
M_t=a_0/\beta(1)+\psi(B)\alpha(B)X_{t-1}=a_0/\beta(1)+\theta(B)X_{t-1},
$$
 where $\theta(z):=\psi(z)\alpha(z)=\sum_{i=0}^\infty\theta_i z^i$ with $\theta_i\geq 0$. Thus, $X_t$ satisfies $X_t=F(X_{t-1},X_{t-2},\ldots;\eta_t)$, where
\begin{align*}
F(X_{t-1},X_{t-2},\ldots;\eta_t)&=\lfloor a_0/\beta(1)+\theta(B)X_{t-1}\rfloor\cdot \epsilon_t\ \\
&~~~+\ \bin_t\big(\epsilon_t,\ a_0/\beta(1)+\theta(B)X_{t-1}-\lfloor a_0/\beta(1)+\theta(B)X_{t-1}\rfloor\big).
\end{align*}
Here, $\eta_t=(\epsilon_t,\bin_t)'$ is an \iid\ sequence with $\bin_t$ denoting a binomial random variable being executed at time~$t$.

From \citet{klenke10}, we know that the stochastic ordering $\bin(n,p_1)\leq_{\rm st} \bin(n,p_2)$ holds if and only if $p_1\leq p_2$.
Denote $F_m$ as the cumulative distribution function (cdf) of $\lfloor m\rfloor\cdot \epsilon\ +\ \bin\big(\epsilon,\ m-\lfloor m\rfloor\big)$, with $m$ being its expectation.
Let $F_m^{-1}(u):=\inf\{t\geq0,F_m(t)\geq u\}$ for $u\in[0,1]$. Let $U$ be a uniform random variable on $(0, 1)$, and define $Z_i=F_{m_i}^{-1}(U),i=1,2$. Then, using arguments similar to those in Lemma~4 of \citet{chen22}, we have $E|Z_1-Z_2|=|m_1-m_2|$.

For any two deterministic sequences $(x_{t-i})_{i\geq1}$ and $(\tilde x_{t-i})_{i\geq1}$, we can define $m_t$ and $\tilde m_t$ accordingly. Then, we obtain
\begin{align*}
&E\big|F(x_{t-1},x_{t-2},\ldots;\eta_t)-F(\tilde x_{t-1},\tilde x_{t-2},\ldots;\eta_t)\big|\\
&=E\Big|\lfloor m_t\rfloor\cdot \epsilon_t\ +\ \bin(\epsilon_t,\ m_t-\lfloor m_t\rfloor)-\big(\lfloor \tilde m_t\rfloor\cdot \epsilon_t\ +\ \bin(\epsilon_t,\ \tilde m_t-\lfloor \tilde m_t\rfloor)\big)\Big|\\
&=|m_t-\tilde m_t|\leq\sum_{i=0}^\infty\theta_i|x_{t-1-i}-\tilde x_{t-1-i}|=\sum_{i=1}^\infty\theta_{i-1}|x_{t-i}-\tilde x_{t-i}|,
\end{align*}
where we have used that $\theta_i\geq0$. Observe that by the identity $\alpha(z)=\beta(z)\theta(z)$, it follows that $\theta(1)<1$ if and only if $\sum_{i=1}^pa_i+\sum_{j=1}^qb_j<1$.
This verifies the conditions of Theorem~3.1 in \citet{doukhan08}, which, in turn, implies that $(X_t)$ is weakly dependent.
\end{proof}

\numberwithin{satz}{section}

\section{Parameter Estimation and Model Diagnostics}
\label{Parameter Estimation}
In this section, we consider two kinds of estimation methods for the INGARCH-CMEMs discussed in Section \ref{MEMs for Count Time Series}: QMLE and 2SWLSE, where the former has been discussed
by \citet{christou14}, \cite{ahmad16} and \cite{aknouche18}, and the latter by \citet{aknouche22}.
To implement these estimation approaches, the INGARCH structure as well as the multiplicative operator have to be specified, but not the distribution of $\epsilon_t$. Thus, QMLE and 2SWLSE approaches can be understood as semi-parametric estimation approaches.

Throughout this section, let $X_1,\ldots, X_n$ be observations from models \eqref{CompCountMEMdef} or \eqref{CompCountMEMdef2}, respectively, with true parameters $\ftheta_0:=(a_{00},a_{01},\ldots,a_{0p},b_{01},\ldots,b_{0q})'$ and $\sigma_0^2$. We assume
that $\ftheta_0\in\fTheta\subset(0,\infty)\times[0,1)^{p+q}$ and $\sigma_0^2\in\Delta\subset(0,\infty)$, where the latter condition ensures the non-degeneracy of $\epsilon_t$.

\subsection{Quasi-Maximum Likelihood Estimation}
\label{Quasi-Maximum Likelihood Estimation}
Because models \eqref{CompCountMEMdef} and \eqref{CompCountMEMdef2} have the same mean, the QMLE approach will yield the same estimate of the regression parameters, although observations come from different models.

\subsubsection{QMLE of Regression Parameters}
\label{QMLE of Regression Parameters}
\citet{christou14} considered Poisson QMLE for negative-binomial INGARCH models, and \citet{ahmad16} studied the asymptotic distribution of the estimator when the parameter belongs to the interior of the parameter space and when it lies at the boundary.

For any generic $\ftheta:=(a_0,a_1,\ldots,a_p,b_1,\ldots,b_q)'\in\fTheta$, define
\ba\label{mean1}
 M_t(\ftheta)\ =\ a_0+\sum_{i=1}^p a_i\, X_{t-i}+\sum_{j=1}^q b_j\, M_{t-j}(\ftheta),~~~t\in\mathbb{Z},
\ea
where $M_t(\ftheta_0)=M_t$.
Let $l_t(\ftheta)=X_t\log M_t(\ftheta)- M_t(\ftheta)$, then the Poisson QMLE (PQ) of $\theta_0$ is defined as
\ba\label{like-fun}
\hat\ftheta_{\textup{PQ}}=\arg\max_{\fthetai\in\fThetai} \sum_{t=1}^n l_t(\ftheta).
\ea
Under some regularity conditions, using techniques similar to those in Theorem~2 of \citet{christou14}, $\hat\ftheta_{\textup{PQ}}$ can be proven to be consistent and asymptotically normal, \ie
$$
\begin{array}{@{}l}
\sqrt n(\hat\ftheta_{\textup{PQ}}-\ftheta_0)\xrightarrow{d}\norm\big(\0,\ \mG^{-1}(\ftheta_0)\mG_1(\ftheta_0,\sigma_0^2)\mG^{-1}(\ftheta_0)\big),
\quad\text{where}
\\[2ex]
\displaystyle
\mG(\ftheta)=-E\left[\frac{\partial^2l_t(\ftheta)}{\partial\ftheta\partial\ftheta'}\right] = E\left[\frac{1}{M_t(\ftheta)}\, \frac{\partial M_t(\ftheta)}{\partial\ftheta}\, \frac{\partial M_t(\ftheta)}{\partial\ftheta'}\right],\\
\\[-2ex]
\displaystyle
\mG_1(\ftheta,\sigma^2)=E\left[V\left[\frac{\partial l_t(\ftheta)}{\partial\ftheta}\bigg| \mathcal F_{t-1}\right]\right] = E\left[\frac{v_t(\ftheta,\sigma^2)}{M_t^2(\ftheta)}\, \frac{\partial M_t(\ftheta)}{\partial\ftheta}\, \frac{\partial M_t(\ftheta)}{\partial\ftheta'}\right],
\end{array}
$$
where $v_t(\ftheta,\sigma^2):=V[X_t\ |\ \mathcal F_{t-1}]$.
In practice, maximization of \eqref{like-fun} is implemented by assuming some (fixed or random) positive starting values $\tilde M_0(\ftheta),\ldots,\tilde M_{1-q}(\ftheta),X_0,\ldots,X_{1-p}$. Let
$\tilde M_t(\ftheta)$ be an observable counterpart of $M_t(\ftheta)$, given by the recursion
\ba\label{mean2}
\tilde M_t(\ftheta) = a_0+\sum_{i=1}^p a_i X_{t-i}+\sum_{j=1}^q b_j\tilde M_{t-j}(\ftheta),~~~t\geq1.
\ea
Then, we have an observable counterpart of $l_t(\ftheta)$ as $\tilde l_t(\ftheta)=X_t\log\tilde M_t(\ftheta)-\tilde M_t(\ftheta)$, and we can approximate $l_t(\ftheta)$ by $\tilde l_t(\ftheta)$.

\medskip
Besides the Poisson QMLE, there are other versions of QMLE as well. \citet{aknouche18} and \cite{aknouche21} discussed the negative-binomial QMLE  and the exponential QMLE, respectively. As noted by \citet{aknouche22}, the asymptotic efficiency of the QMLE depends on the model's variance-to-mean relationship $V[X_t\ |\ \mathcal F_{t-1}]=g(M_t)$, where $g$ is a positive real-valued function.
Thus, each one of the Poisson, negative-binomial, and exponential QMLE is asymptotically optimal whenever $g(x)=c_1x$, $g(x)=c_2x(1+c_3x)$, and $g(x)=c_4x^2$, respectively, for any $c_1,c_2,c_3,c_4>0$.
From \eqref{CompCountMEMcvar} and \eqref{CompCountMEMcvar2}, we may conclude that negative-binomial and exponential QMLEs are asymptotically efficient for models \eqref{CompCountMEMdef} and \eqref{CompCountMEMdef2}, respectively. We should point out that they are not universally (asymptotically) efficient but only within the class of all QMLEs in the linear exponential family (see also \cite{wooldridge} and the references therein).

For model \eqref{CompCountMEMdef}, the negative-binomial QMLE (NQ) is defined as
$$
\hat\ftheta_{\textup{NQ}}=\arg\max_{\fthetai\in\fThetai} \sum_{t=1}^n \Big\{X_t\log\tilde M_t(\ftheta)-(r+X_t)\log(r+\tilde M_t(\ftheta))\Big\}
$$
for some fixed $r>0$ (in their simulations, \citet{aknouche22} used $r=1$ throughout). Note that \citet{aknouche18} also proposed another estimation approach, which consists of four stages estimating both conditional mean and dispersion parameters.

For model \eqref{CompCountMEMdef2}, the exponential QMLE (EQ) is defined as
$$
\hat\ftheta_{\textup{EQ}}=\arg\max_{\fthetai\in\fThetai} \sum_{t=1}^n \Big\{-\log\tilde M_t(\ftheta)-X_t/\tilde M_t(\ftheta)\Big\}.
$$
Large-sample properties of $\hat\ftheta_{\textup{NQ}}$ and $\hat\ftheta_{\textup{EQ}}$ are discussed by \citet{aknouche18} and \citet{aknouche21}, respectively, which will not be repeated here. From a practical point of view, the crucial difference are modified expressions for the matrices $\mG(\ftheta)$ and $\mG_1(\ftheta,\sigma^2)$ involved in the asymptotic normal distribution, namely
$$
\textstyle
\mG(\ftheta) = E\left[\frac{1}{M_t(\fthetai) (r+M_t(\fthetai))}\, \frac{\partial M_t(\fthetai)}{\partial\fthetai}\, \frac{\partial M_t(\fthetai)}{\partial\fthetai'}\right],\
\mG_1(\ftheta,\sigma^2) = E\left[\frac{v_t(\fthetai,\sigma^2)}{M_t^2(\fthetai) (r+M_t(\fthetai))^2}\, \frac{\partial M_t(\fthetai)}{\partial\fthetai}\, \frac{\partial M_t(\fthetai)}{\partial\fthetai'}\right]
\ \text{for } \hat\ftheta_{\textup{NQ}},
$$
and
$$
\textstyle
\mG(\ftheta) = E\left[\frac{1}{M_t^2(\fthetai)}\, \frac{\partial M_t(\fthetai)}{\partial\fthetai}\, \frac{\partial M_t(\fthetai)}{\partial\fthetai'}\right],\
\mG_1(\ftheta,\sigma^2) = E\left[\frac{v_t(\fthetai,\sigma^2)}{M_t^4(\fthetai)}\, \frac{\partial M_t(\fthetai)}{\partial\fthetai}\, \frac{\partial M_t(\fthetai)}{\partial\fthetai'}\right]
\ \text{for } \hat\ftheta_{\textup{EQ}}.
$$
For any of the three QMLE approaches, the estimates are determined by numerical optimization. To initialize the optimization routine, a moment estimate of~$\ftheta_0$ might be used, which is easily obtained from sample mean and ACF by solving the moment properties provided by Section~\ref{MEMs for Count Time Series}.

\subsubsection{Estimation of Parameter $\sigma^2$}
\label{Estimation-of-sigma}
It remains to estimate the parameter value~$\sigma_0^2$, which corresponds to the variance of~$\epsilon_t$. For model \eqref{CompCountMEMdef} with Poisson-counting series, the conditional variance is given by $V[X_t\ |\ \mathcal F_{t-1}] = M_t + \sigma^2 M_t^2$, recall \eqref{CompCountMEMcvar}. Thus, let $e_t=(X_t-M_t)^2-(M_t+\sigma^2\,M_t^2)$. Then, we have $(X_t-M_t)^2=M_t+\sigma^2\,M_t^2+e_t$, or equivalently,
\ba\label{sigma-eq}
\frac{(X_t-M_t(\ftheta))^2-M_t(\ftheta)}{M_t^2(\ftheta)}=\sigma^2+\tau_t,
\ea
where $\tau_t=e_t/M_t^2(\ftheta)$ is a term of  martingale difference. Replacing $M_t(\ftheta)$ by $\hat M_t:=\tilde M_t(\hat\ftheta_{\textup{PQ}})$ or $\tilde M_t(\hat\ftheta_{\textup{NQ}})$ or $\tilde M_t(\hat\ftheta_{\textup{EQ}})$, respectively (one might also use the moment estimate of~$\ftheta_0$), we obtain the least-squares estimate of $\sigma_0^2$ as
$$
\hat\sigma_\poi^2=\frac{1}{n}\sum_{t=1}^n\frac{(X_t-\hat M_t)^2-\hat M_t}{\hat M_t^2}.
$$
Under some regularity conditions, using techniques similar to those in Theorem~3.2 of \citet{aknouche22}, $\hat\sigma_\poi^2$ can be proven to be consistent and asymptotically normal, \ie
$$
\sqrt n(\hat\sigma_\poi^2-\sigma_0^2)\xrightarrow{d} \norm(0,\Lambda),
\quad\text{where}\quad
\Lambda=E\left[\left(\frac{(X_t-M_t(\ftheta_0))^2-(M_t(\ftheta_0)+\sigma^2_0\,M_t^2(\ftheta_0))}{M_t^2(\ftheta_0)}\right)^2\right].
$$
In the case of an NB-counting series, we have $V[X_t\ |\ \mathcal F_{t-1}] = M_t(1+M_t) + \sigma^2 M_t^2$ according to \eqref{CompCountMEMcvar}, so we would define the least-squares estimate of $\sigma_0^2$ as
$$
\hat\sigma_\nb^2=\frac{1}{n}\sum_{t=1}^n\frac{(X_t-\hat M_t)^2-\hat M_t(1+\hat M_t)}{\hat M_t^2}\ =\ \hat\sigma_\poi^2 -1.
$$
But as can be seen by the last equality, we just end up with the estimator of the Poi-case reduced by one. This seems plausible as now a larger part of the variation is explained by the counting series.

\medskip
Finally, for model \eqref{CompCountMEMdef2}, the conditional variance is given by $V[X_t\ |\ \mathcal F_{t-1}] = ( M_t-\lfloor M_t\rfloor)(1-\ M_t+\lfloor M_t\rfloor)+\sigma^2 M_t^2$, recall \eqref{CompCountMEMcvar2}. Thus, let $e_t=(X_t-M_t)^2-\big((M_t-\lfloor M_t\rfloor)(1-\ M_t+\lfloor M_t\rfloor)+\sigma^2\,M_t^2\big)$. Then, we have
$$
\hat\sigma_\bin^2=\frac{1}{n}\sum_{t=1}^n\frac{(X_t-\hat M_t)^2-(\hat M_t-\lfloor \hat M_t\rfloor)(1-\hat M_t+\lfloor \hat M_t\rfloor)}{\hat M_t^2},
$$
and $\sqrt n(\hat\sigma_\bin^2-\sigma_0^2)\xrightarrow{d} \norm(0,\Lambda)$ with
$$
\Lambda=E\left[\left(\frac{(X_t-M_t(\ftheta_0))^2-\big((M_t(\ftheta_0)-\lfloor M_t(\ftheta_0)\rfloor)(1-\ M_t(\ftheta_0)+\lfloor M_t(\ftheta_0)\rfloor)+\sigma^2_0\,M_t^2(\ftheta_0)\big)}{M_t^2(\ftheta_0)}\right)^2\right].
$$

\subsection{Two-Stage Weighted Least-Squares Estimation}
\label{Two-stage Weighted Least Squares Estimation}
As an alternative to the above QMLEs, \citet{aknouche22} proposed a 2SWLSE approach.
The 2SWLSE is not asymptotically efficient in the sense that its asymptotic variance
reaches the Cramer--Rao bound, but is only ``never asymptotically less efficient'' than any QMLE
in the linear exponential family. When the conditional distribution of the model is known and is
a member of the exponential family, then the 2SWLS estimate is asymptotically efficient (in the
sense of the Cramer--Rao bound).  See \citet{aknouche12} and \citet{aknouche22} for more
details and other models.
We shall describe the estimation procedure as follows.
For convenience, we again denote $v_t(\ftheta,\sigma^2):=V[X_t\ |\ \mathcal F_{t-1}]$.
For an arbitrary fixed weighting point, say $(\ftheta'_\ast,\sigma^2_\ast)'\in\fTheta\times\Delta$ (in practice, the aforementioned moment estimates of $(\ftheta_0,\sigma^2_0)$ might be used), a first-stage weighted least-squares estimator (1WLSE, briefly ``1W'') is defined as
\ba
\label{1WLSE}
\hat\ftheta_{\textup{1W}}=\arg\min_{\fthetai\in\fThetai} \sum_{t=1}^n \frac{(X_t-\tilde M_t(\ftheta))^2}{v_t(\ftheta_\ast,\sigma^2_\ast)}.
\ea
We next estimate $\sigma^2$ like in Section~\ref{Estimation-of-sigma}, and we denote the estimator as $\hat\sigma_n^2$. But here, we replace $M_t(\ftheta)$ by $\hat M_t=\tilde M_t(\hat\ftheta_{\textup{1W}})$.
After having consistent estimates of $\ftheta_0$ and $\sigma_0^2$, the optimal weight is estimated by $v_t(\hat\ftheta_{\textup{1W}},\hat\sigma^2_n)$, and a second-stage weighted least-squares estimator (2WLSE, briefly ``2W'') of $\ftheta_0$ has the form
\ba
\label{2WLSE}
\hat\ftheta_{\textup{2W}}=\arg\min_{\fthetai\in\fThetai} \sum_{t=1}^n \frac{(X_t-\tilde M_t(\ftheta))^2}{v_t(\hat\ftheta_{\textup{1W}},\hat\sigma^2_n)}.
\ea
Under some regularity conditions, using techniques similar to those in Theorem~3.3 in \citet{aknouche22}, $\hat\ftheta_{\textup{2W}}$ can be proven to be strongly consistent and asymptotically normal, \ie
$\sqrt n(\hat\ftheta_{\textup{2W}}-\ftheta_0)\xrightarrow{d} \norm(\0,\mJ^{-1}(\ftheta_0,\sigma_0^2)),$
where
$$
\mJ(\ftheta_0,\sigma_0^2)=E\left[\frac{1}{v_t(\ftheta_0,\sigma^2_0)}\, \frac{\partial M_t(\ftheta_0)}{\partial\ftheta}\, \frac{\partial M_t(\ftheta_0)}{\partial\ftheta'}\right].
$$
\citet{aknouche22} stated that $\hat\theta_{\textup{2W}}$ is asymptotically more efficient than $\hat\ftheta_{\textup{1W}}$ and the QMLEs.
Note that if $b_{0j}=0$ for $1\leq j\leq q$, then $\hat\ftheta_{\textup{1W}}$ and $\hat\ftheta_{\textup{2W}}$ have a closed form, which is omitted here.

\subsection{Approaches for Model Diagnostics}
\label{Approaches for Model Diagnostics}
To check the adequacy of the fitted INGARCH-CMEMs (or to compare the performance of different candidate models), several approaches are possible, see Section~2.4 in \citet{weiss18} for a general survey. It should be noted, however, that we solely considered semi-parametric estimation approaches, \ie we did not specify the full distribution of the considered CMEM. Thus, mainly moment-related diagnostic checks are relevant here.

The first option is to compare the fitted model's mean, variance, and ACF (which can be computed from the results derived in Section~\ref{MEMs for Count Time Series}) to the corresponding sample moments.

Second, to check the adequacy of conditional moments, it is common to consider the standardized Pearson residuals $(R_t)$, which are given by
\ba
\label{pres}
R_t(\hat\ftheta, \hat\sigma^2)
\ =\ \frac{X_t-E[X_t\ |\ \mathcal F_{t-1};\ \hat\ftheta, \hat\sigma^2]}{\sqrt{V[X_t\ |\ \mathcal F_{t-1};\ \hat\ftheta, \hat\sigma^2]}}
 \overset{\eqref{CountMEMcvar}}{=} \frac{X_t-M_t(\hat\ftheta)}{\sqrt{\nu\big(M_t(\hat\ftheta)\big) + M_t^2(\hat\ftheta)\cdot \hat\sigma^2}}
\ea
for the models considered here. For an adequate model, the Pearson residuals have mean zero, variance one, and are uncorrelated.
In particular, \citet{aknouche22} used the mean squared Pearson residual, $\textup{MSPR}=\frac{1}{n}\sum_{t=1}^n R_t^2(\hat\ftheta, \hat\sigma^2)$,
for checking the adequacy of the dispersion structure, also see \citet{aleksandrov20} for a more general analysis of this statistic. For an andequate model, the MSPR should be close to one.

Third, for INGARCH-type models, also another type of residuals is commonly used, where the~$X_t$ (or powers thereof) are scaled by~$M_t$, see \citet{weissetal17} for a general discussion.
For INGARCH-type and multiplicative models such as the one proposed by \citet{zhu10} and \citet{aknouche22} and the CMEMs considered here, especially the scaled residuals $S_t(\hat\ftheta) = X_t/M_t(\hat\ftheta)$ are reasonable, and \citet{aknouche22} proposed to compute the sample mean and variance thereof (MSR and VSR, respectively). Note that if $(X_t)$ stems from the CMEM \eqref{CompCountMEMdef} satisfying \eqref{CompCountMEMcvar}, then
$$
\begin{array}{@{}rl}
E[S_t(\ftheta)]\ =& E\big[E[X_t/M_t\ |\ \mathcal F_{t-1}\big]\ =\ 1,
\\[1ex]
V[S_t(\ftheta)]\ =& 0\ +\ E\big[V[X_t/M_t\ |\ \mathcal F_{t-1}\big]\ =\ \sigma^2 + E\big[\nu(M_t)/M_t^2\big].
\end{array}
$$
The last term, $E\big[\nu(M_t)/M_t^2\big]$, depends on the specific type of CMEM. For the compounding CMEM with Poisson counting series, we have $V[S_t(\ftheta)] = \sigma^2 + E[1/M_t]$, while $V[S_t(\ftheta)] = \sigma^2 + 1+E[1/M_t]$ for the NB-counting series. The latter result also shows that an NB-counting series is possible only if the scaled residuals' variance is larger than one. Generally, if~$\nu$ is differentiable, we can approximate $E\big[\nu(M_t)/M_t^2\big]$ by a second-order Taylor expansion around the marginal mean~$\mu=E[X_t]$ as
$$
E\left[\frac{\nu(M_t)}{M_t^2}\right]\ \approx\ \frac{\nu(\mu)}{\mu^2} + \frac{\mu^2\, \nu''(\mu)-4\, \mu\, \nu'(\mu) + 6\,\nu(\mu)}{2\,\mu^4}\,V[X_t],\quad
E\left[\frac{1}{M_t}\right]\ \approx\ \frac{1}{\mu} + \frac{1}{\mu^3}\,V[X_t].
$$
For the binomial CMEM \eqref{CompCountMEMdef2}, the function $\nu(\alpha) = (\alpha-\lfloor\alpha\rfloor)(1-\alpha+\lfloor\alpha\rfloor)$ is not differentiable, but it is bounded by $\nu(\alpha)\in [0,0.25]$. Thus, $V[S_t(\ftheta)] \leq \sigma^2 + E[0.25/M_t^2]$, \ie $V[S_t(\ftheta)]$ often takes a value close to $\sigma^2$.

As a fourth type of diagnostic check, \citet{aknouche22} proposed to compute the mean absolute residual (MAR) defined as $\textup{MAR}=\frac{1}{n}\sum_{t=1}^n |X_t-M_t(\hat\ftheta)|$, where smaller values indicate a better model fit.

\section{Simulation Study}
\label{Simulations}
To analyze the finite-sample performance of the estimation approaches discussed in Section~\ref{Parameter Estimation}, a simulation study was done. For different types of data-generating process (DGP) and sample size $n\in\{300,600,1000\}$, we simulated 10,000 time series and computed estimates according to the methods PQ, NQ, EQ, and 2W (by using moment estimates for initialization, recall Example~\ref{examINGARCH11}). As the DGP, we considered twelve types of INGARCH$(1,1)$-CMEM, namely:
\begin{itemize}
	\item $(a_0, a_1, b_1)=(2.8, 0.4, 0.2)$ and $\sigma^2=1$ or $\sigma^2=0.4$,\\
	similar to the simulation scenario in Section~4 of \citet{aknouche22};
	\item $(a_0, a_1, b_1)=(3, 0.35, 0.5)$ and $\sigma^2=1$ or $\sigma^2=0.1$,\\
	similar to the first data example in Section~\ref{Real-World Data Examples};
	\item $(a_0, a_1, b_1)=(1, 0.25, 0.65)$ and $\sigma^2=1$ or $\sigma^2=0.4$,\\
	similar to the second data example in Section~\ref{Real-World Data Examples};
\end{itemize}
and these parametrizations are combined with either the compounding CMEM \eqref{CompCountMEMdef} with Poisson-counting series, or the binomial CMEM \eqref{CompCountMEMdef2}. The model order $(1,1)$ is relevant for many INGARCH applications in practice, including our subsequent data examples.

The considered values of~$\sigma^2$ are realized by an appropriate choice of the innovations' distribution. Recall that the standard INGARCH model corresponds to a compounding INGARCH-CMEM with deterministic $\epsilon_t=1$ (thus $\sigma^2=0$). In applications, the variance $\sigma^2=V[\epsilon_t]$ of a fitted INGARCH-CMEM is often rather small, see Section~\ref{Real-World Data Examples} for such examples.
Furthermore, an upper bound for~$\sigma^2$ is also implied by Proposition~\ref{second-order-stationarity} if requiring for second-order stationarity. In the case of a Poisson-counting series, for example, $\sigma^2  < (1 - (a_1+b_1)^2)/a_1^2$ has to hold.
If using $\sigma^2=1$ in the above DGP specifications, this value is achieved by $\poi(1)$-distributed $(\epsilon_t)$. Our model for achieving $\sigma^2\in (0,1)$, by contrast, is a three-point distribution with range $\{0,1,2\}$. From the discussion at the end of Section~\ref{Count-MEMs using a Compounding Operator}, we know that we can choose $p_{\epsilon,2}\in (0,0.5)$ arbitrarily. Then, $p_{\epsilon,0} = p_{\epsilon,2}$ and $p_{\epsilon,1} = 1-2\, p_{\epsilon,2}$ follow such that $\sigma^2 = E[\epsilon_t^2]-1 = 2\, p_{\epsilon,2}\in (0,1)$.

\begin{figure}[t]
\centering\small
\begin{tabular}{@{}c@{\qquad}c}
Estimation of $a_0=2.8$ & Estimation of $a_1=0.4$ \\
\includegraphics[viewport=30 45 300 235, clip=, scale=0.75]{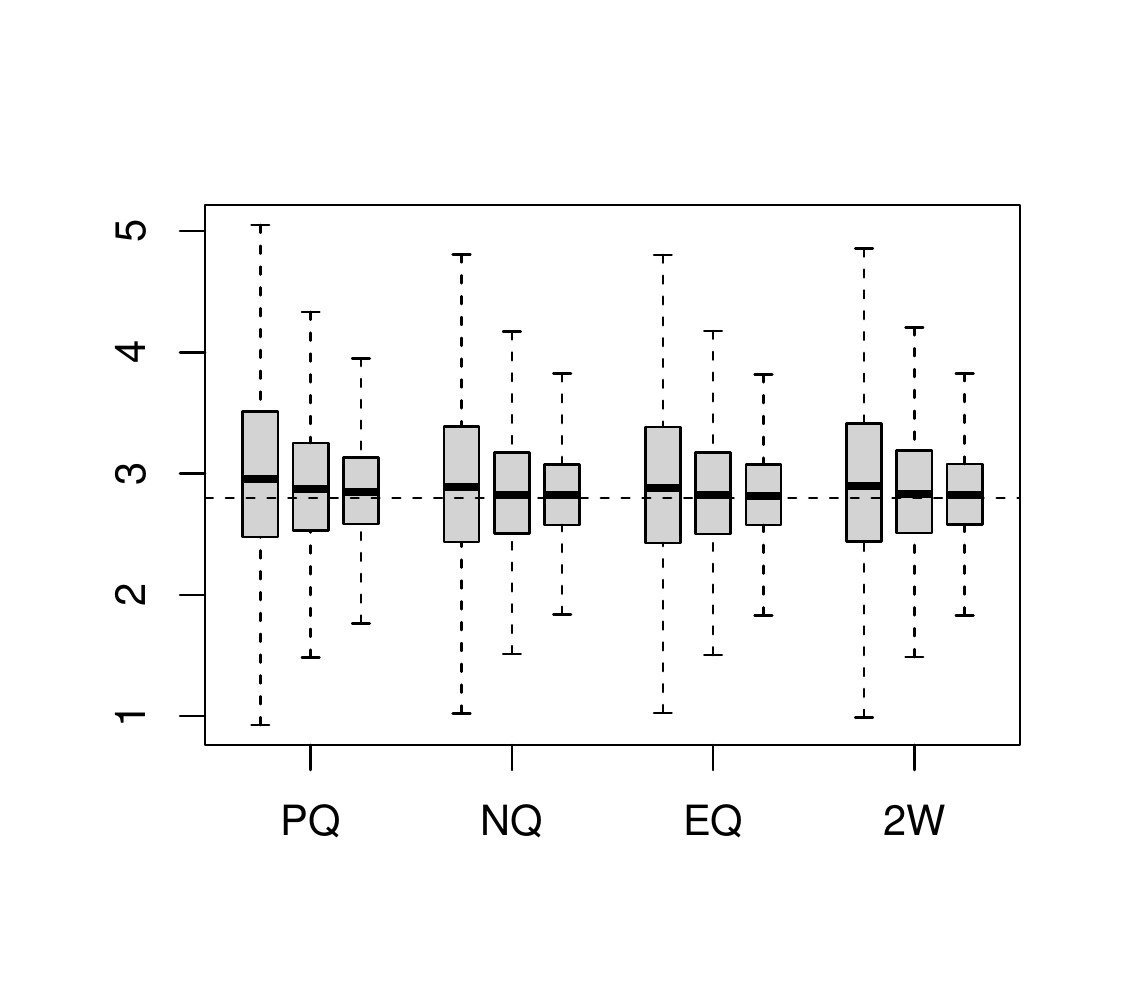}
&
\includegraphics[viewport=30 45 300 235, clip=, scale=0.75]{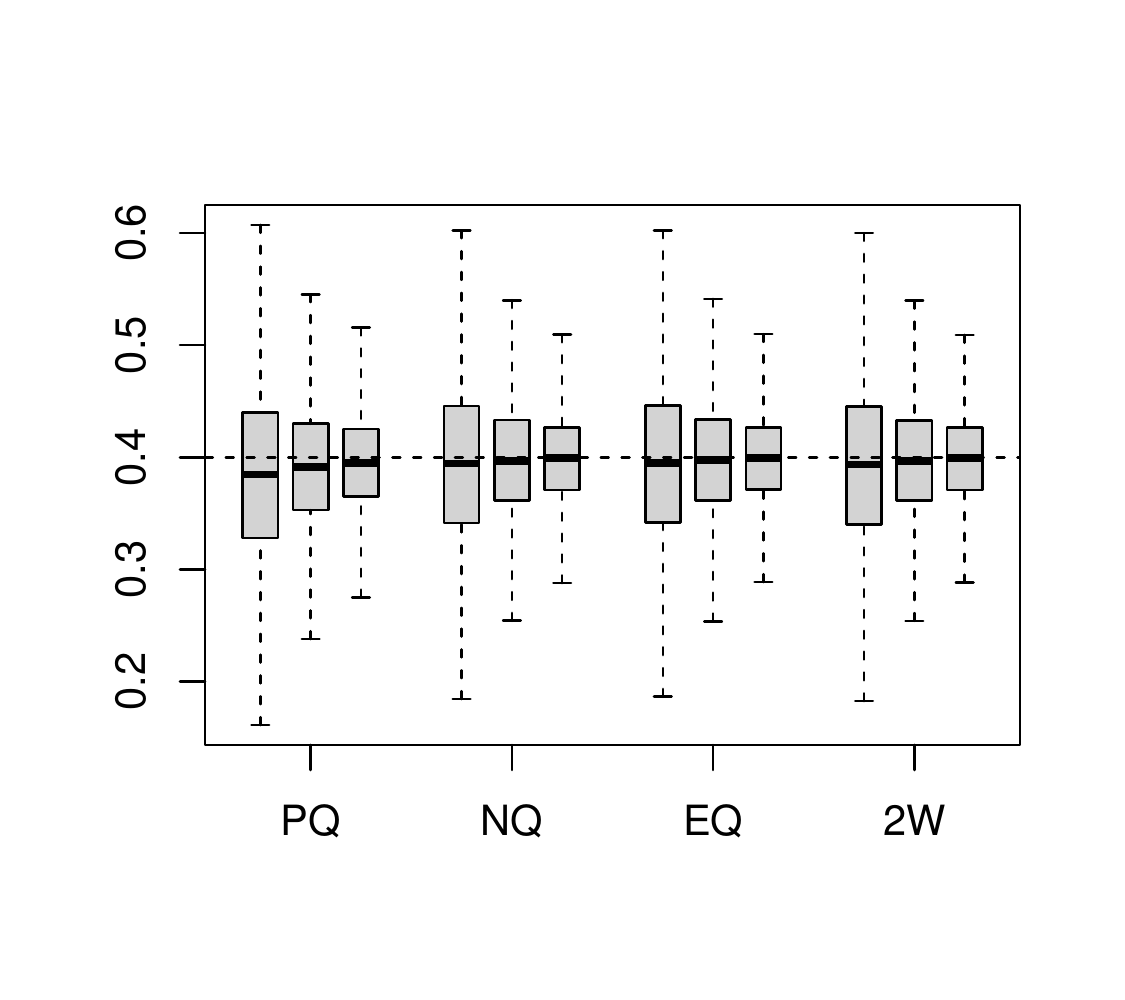}
\\[2ex]
Estimation of $b_1=0.2$ & Estimation of $\sigma^2=1$ \\
\includegraphics[viewport=30 45 300 235, clip=, scale=0.75]{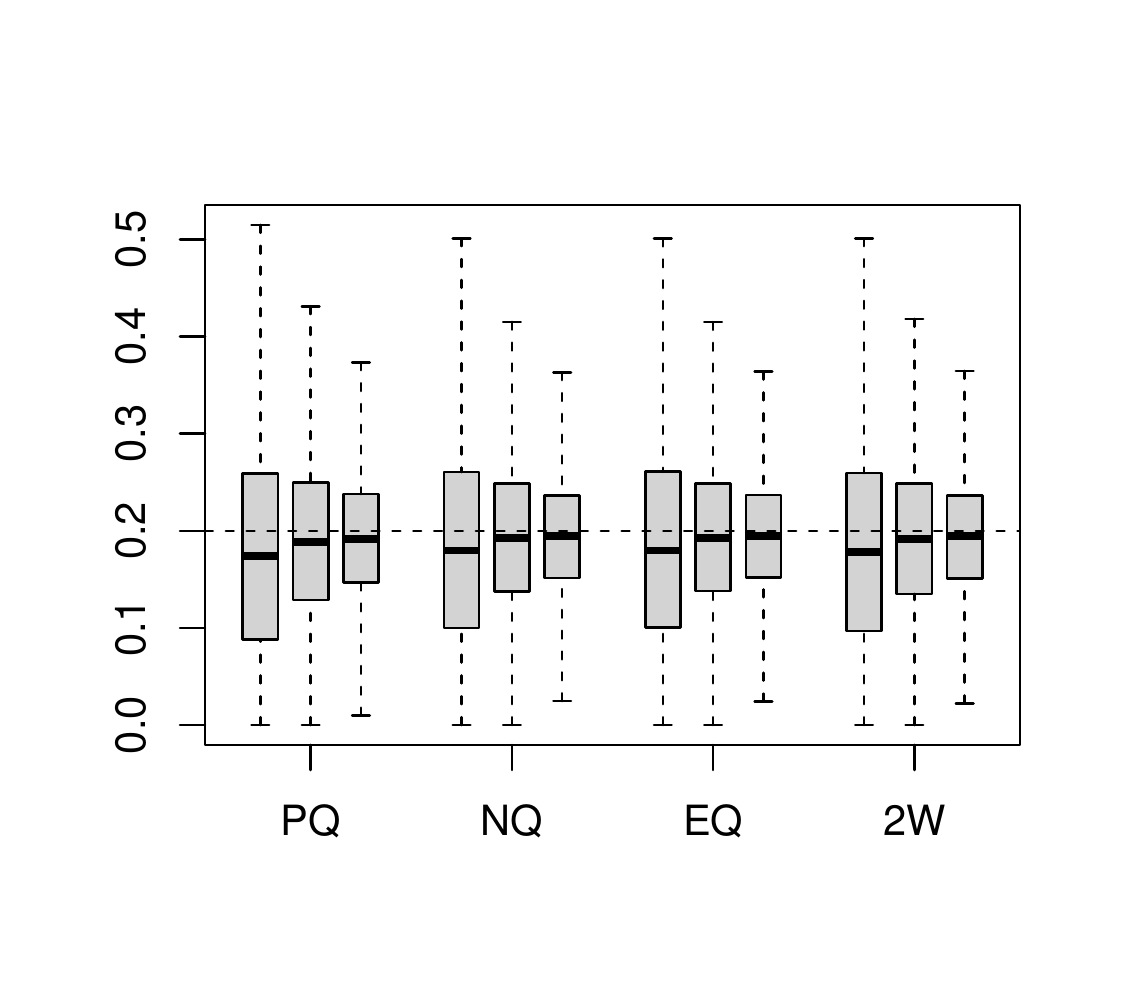}
&
\includegraphics[viewport=30 45 300 235, clip=, scale=0.75]{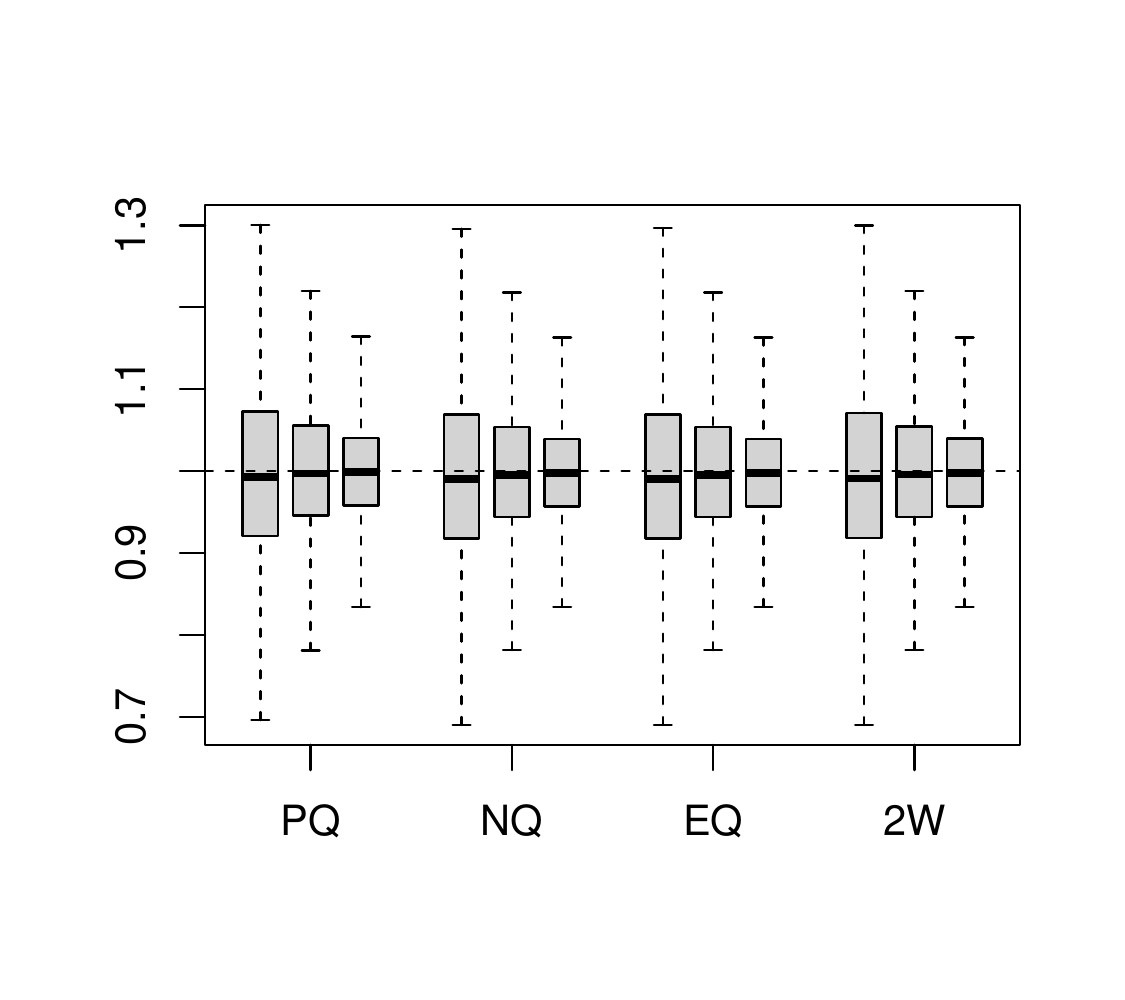}
\end{tabular}
\caption{INGARCH$(1,1)$-CMEM with Poi-counting series: four groups of boxplots for different estimation approaches, resulting from simulated count time series of lengths $n\in\{300,600,1000\}$ per group.}
\label{figSimulations}
\end{figure}

\medskip
Detailed simulation results are summarized in Tables~S.1--S.6 in the Supplement. Furthermore, illustrative boxplots for one example scenario (corresponding to the left part of Table~S.1) are presented in Figure~\ref{figSimulations}. Before discussing these results, a few notes on Tables~S.1--S.6.
The standard errors (SEs) of the considered estimators are approximated based on the asymptotic normal distributions discussed in Section~\ref{Parameter Estimation}. For this purpose, the means $E[\cdot]$ involved in the expressions for $\mG(\ftheta_0)$, $\mG_1(\ftheta_0, \sigma_0^2)$, etc.\ are replaced by sample means, and the respective estimates are plugged-in instead of the true parameter values. The required partial derivatives $\frac{\partial M_t(\fthetai)}{\partial\fthetai}$ can be computed recursively based on \eqref{INGARCHmodels}. For model order $p=q=1$, for example, we get the recursive scheme
$$
\frac{\partial M_t(\ftheta)}{\partial \ftheta} = \big(1, X_{t-1}, M_{t-1}(\ftheta)\big)\ +\ b_1\, \frac{\partial M_{t-1}(\ftheta)}{\partial \ftheta},
$$
also see \citet[pp.~936--937]{ferland06}. In Tables~S.1--S.6, two types of SE are presented: the ``true'' SE, computed as the (trimmed) sample standard deviation among the simulated estimates and abbreviated as ``SSE'', and the mean among the aforementioned approximate SEs, abbreviated as ``ASE''.

The results in Figure~\ref{figSimulations} show that bias and SE decrease with increasing sample size~$n$ (also see Tables~S.1--S.6), confirming the consistency of the estimators. But Figure~\ref{figSimulations} also demonstrates one of the general patterns in Tables~S.1--S.6, namely that the estimates of the AR-parameter~$a_1$ and the variance~$\sigma^2$ are only mildly affected by bias, whereas the estimates of the feedback parameter~$b_1$ and especially of the intercept parameter~$a_0$ might be biased also for larger sample sizes. From Tables~S.1--S.6, it gets clear that this bias is mainly caused by the extent of~$b_1$: the larger~$b_1$, the stronger the bias (especially of~$a_0$). The problems in estimating the intercept term of INGARCH-type models are well known in the literature, see \citet{fokianos09} for example.

Further details can be recognized from Tables~S.1--S.6. If $\sigma^2=1$ (so a large~$\sigma^2$, always see the left part of Tables~S.1--S.6), the PQ-method often causes larger bias and SE than NQ, EQ (especially for~$a_0$), whereas we often see the opposite ranking if $\sigma^2<1$ (right part of Tables~S.1--S.6). This goes along with the discussion of asymptotic optimality in \citet{aknouche22}, recall Section~\ref{QMLE of Regression Parameters}, because for small~$\sigma^2$, the conditional variance of \eqref{CompCountMEMdef} is close to a linear function in~$M_t$. The 2W-method often has bias and SE between the ones of PQ and NQ, EQ. Comparing Tables~S.1 vs.\ S.2, S.3 vs.\ S.4, and S.5 vs.\ S.6, \ie the case of a Poi-counting series to a binomial one, we do not see big differences in performance, \ie the estimation approaches generally perform similar for both classes of DGP. An exception is the bias of~$a_0$, which is somewhat larger for the Bin-counting series if $b_1\geq 0.5$. Finally, the values of SSE and ASE are often close to each other. If there are visible differences between SSE and ASE, these mainly happen for $n=300$, and the strongest absolute discrepancies are observed for the intercept term~$a_0$. The relative deviation, however, is sometimes stronger for~$\sigma^2$. In particular, the 2W-method sometimes leads to visibly larger SSE than ASE (\eg lower left block in Table~S.6), which is caused by a few extreme values of the estimates. The QMLE have less tendency to produce extreme estimates. Overall, the approximate SE seems to give reasonable insight into the true SE of the estimator in most cases.

\bigskip
We conclude this section with some simulation experiments on the effects of a possible model misspecification. Here, we focus on the DGPs with parametrization $(a_0, a_1, b_1)=(2.8, 0.4, 0.2)$ and $\sigma^2=1$ or $\sigma^2=0.4$ (see Tables~S.1 and~S.2). In our first experiment, we fitted an INGARCH$(1,1)$-CMEM to the generated time series by assuming the wrong multiplicative operator:
\begin{itemize}
	\item for the DGP with Poi-counting series from Table~S.1, the misspecified model assumes a Bin-counting series, leading to the simulation results of Table~S.7;
	\item for the DGP with Bin-counting series from Table~S.2, the misspecified model assumes a Poi-counting series, leading to the simulation results of Table~S.8.
\end{itemize}
As the estimation approaches PQ, NQ, and EQ for $(a_0, a_1, b_1)$ do not depend on the specific structure of the multiplicative operator, the corresponding simulated means and SEs in Tables~S.7 and~S.8 are identical to those in Tables~S.1 and~S.2, respectively. The 2W-method, by contrast, depends on the choice of the multiplicative operator, but as can be seen by comparing Table~S.7 to~S.1 and Table~S.8 to~S.2, the effect of the misspecified operator is negligible. So altogether, the estimation of $(a_0, a_1, b_1)$ is robust against such model misspecification. The approximate SEs, in turn, are always affected by the misspecified operator, but the difference between correct and misspecified ASE is again small for $(a_0, a_1, b_1)$.

A clear difference is only observed for the estimation of~$\sigma^2$ (and its corresponding ASE), which is plausible in view of Section~\ref{Estimation-of-sigma}, where tailor-made estimation methods are presented for each type of counting series. However, in view of George Box' famous words ``all models are wrong but some are useful'' (\ie real-world data will never follow exactly one of the considered models, but some of these models might provide a good approximation to the true DGP), the practitioner is probably more interested in the performance of the fitted models. In Table~S.9, the mean MAR values are provided if estimation is done based on the correct vs.\ the misspecified model. For the estimation approaches PQ, NQ, and EQ, the MAR values are identical as the estimates for $(a_0, a_1, b_1)$ do not depend on the specific multiplicative operator. Also for the 2W-method, there is an at most slight difference in the mean MARs, so the MAR performance is very robust against such misspecification. Analogous results but concerning the MSPR performance are shown in Table~S.10. If the Poi-counting series is misspecified as a Bin-counting series (upper part of Table~S.10), there is again an at most tiny difference in the resulting MSPRs. A somewhat stronger effect (although still quite small) is observed in the opposite scenario, where the Bin-counting series is misspecified as a Poi-counting series. So it seems that the CMEM with Bin-counting series can better adapt to the given DGP than the Poi-one.

\smallskip
As a final experiment on the effects of model misspecification, we introduce non-linearity into the DGP while model fitting still assumes the conditional linearity in \eqref{INGARCHmodels}. More precisely, we modify \eqref{INGARCHmodels} by adding the softplus response function $s_c(x)=c\,\ln\big(1+\exp(x/c)\big)$ with adjustment parameter~$c>0$, in analogy to the class of softplus INGARCH models proposed by \citet{weissetal22}:
\ba
\label{spINGARCHmodels}
\textstyle
M_t\ =\ s_c\Big(a_0+\sum_{i=1}^p a_i\, X_{t-i}+\sum_{j=1}^q b_j\, M_{t-j}\Big).
\ea
Increasing~$c$ increases the extent of non-linearity, whereas $\lim_{c\to 0} s_c(x)=\max\{0,x\}$ is piecewise linear. We focus on the illustrative choice $c=2$ (for the default choice $c=1$, we hardly observed any effect on estimates and diagnostics), where the results are summarized in Tables~S.11--S.14. Comparing Table~S.11 to~S.1 and Table~S.12 to~S.2, we note a slight effect of the model misspecification on the estimates for all parameters, and with the strongest effect again on~$\sigma^2$. For practice, however, the main question is the achieved model performance. In Tables~S.13 and~S.14, the MAR and MSPR performance, respectively, is analyzed in comparison to that of the truly linear DGPs. For the MSPR values in Table~S.14, there is an at most tiny effect if DGPs with Poi-counting series are considered, whereas the increase in MSPR is slightly more pronounced in the Bin-case. The main difference, however, is observed in the MAR performance, see Table~S.13, namely clearly increased MAR values under neglected non-linearity. This is plausible as the softplus response function directly affects the conditional mean~$M_t$, which, in turn, is the crucial quantity when computing the MAR, see Section~\ref{Approaches for Model Diagnostics}.

\section{Real-World Data Examples}
\label{Real-World Data Examples}

\begin{figure}[b]
\centering\small
(a)\hspace{-3ex}\includegraphics[viewport=0 45 405 235, clip=, scale=0.7]{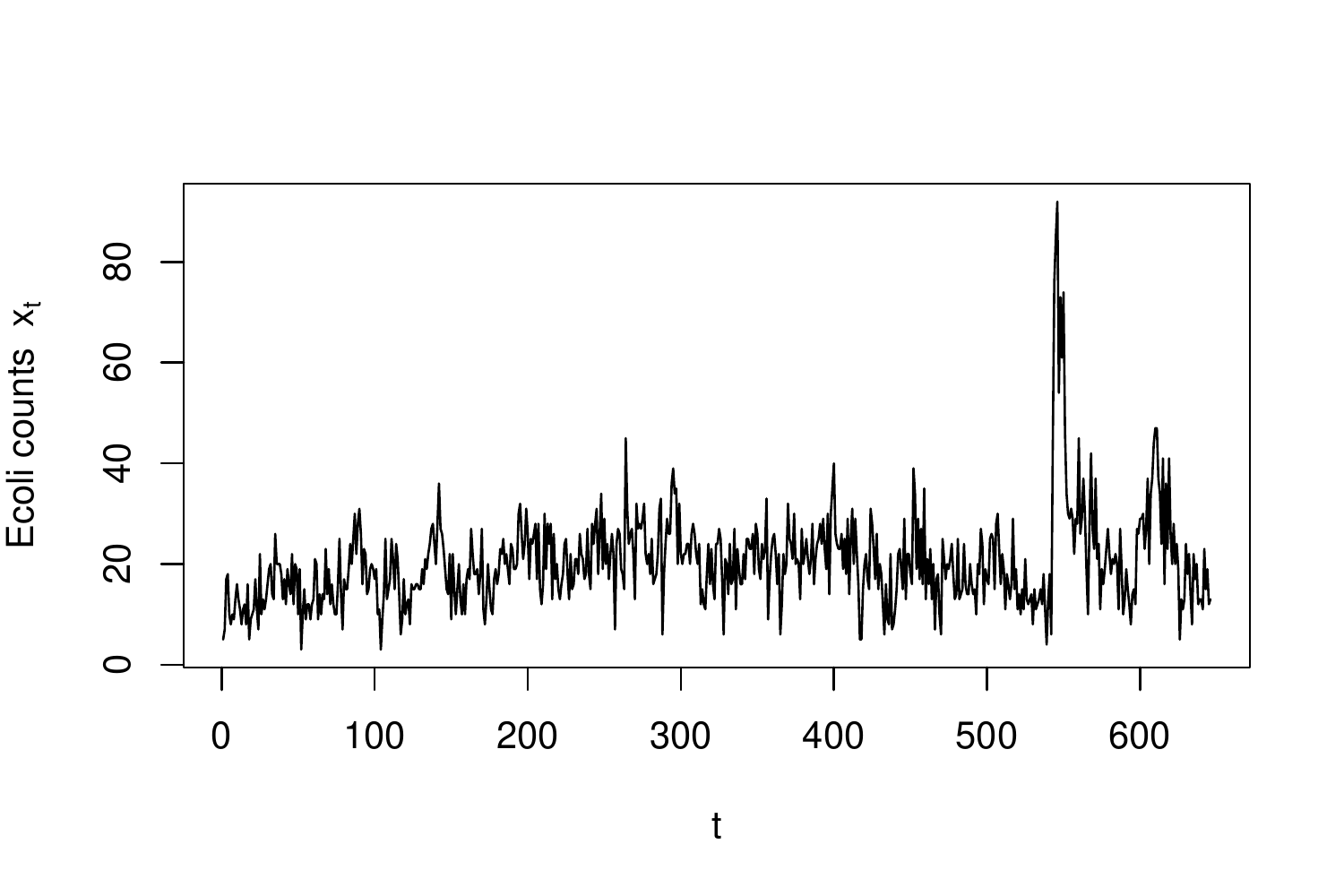}\hspace{-1ex}$t$
\qquad
(b)\hspace{-3ex}\includegraphics[viewport=0 45 190 235, clip=, scale=0.7]{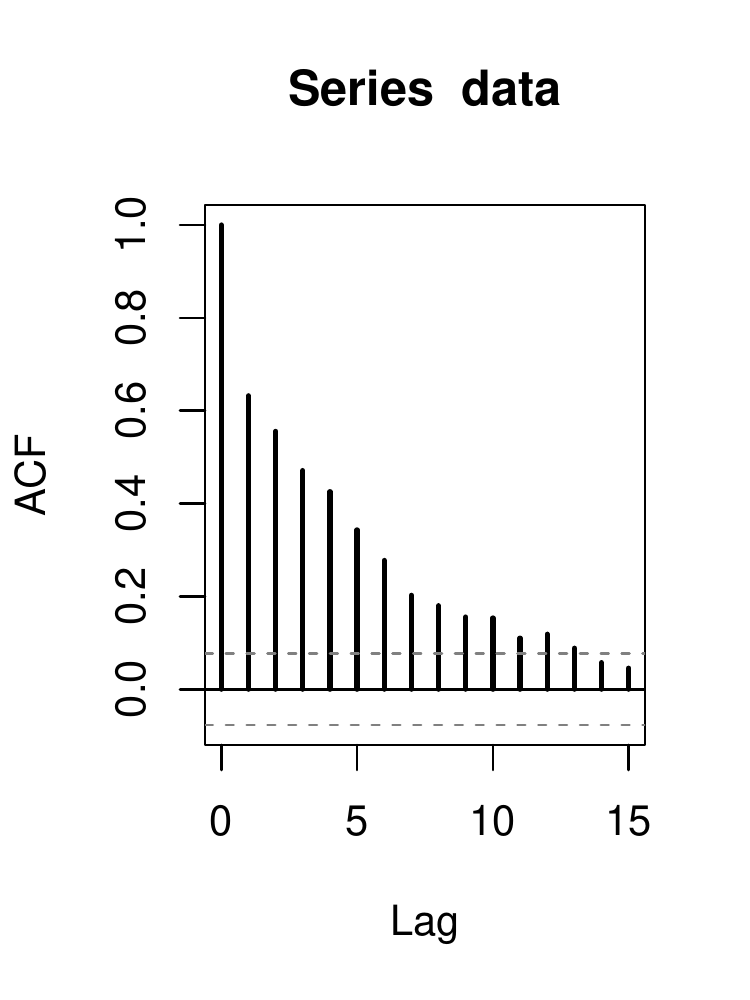}\ $k$
\caption{Ecoli counts from Section~\ref{Disease Counts}: (a) time series plot and (b) sample ACF against lag~$k$.}
\label{figEcoli}
\end{figure}

\subsection{Disease Counts}
\label{Disease Counts}
As our first data example, we use one of the data sets analyzed by \citet{aknouche22}, as these authors proposed another multiplicative model for counts, recall Remark~\ref{remCMEMaknouche}. More precisely, we consider the weekly counts of disease cases caused by Escherichia coli (Ecoli), as reported for North Rhine-Westphalia (Germany) from January 2001 to May 2013. These data are provided through the command \texttt{ecoli} of the R-package \href{https://CRAN.R-project.org/package=tscount}{\nolinkurl{tscount}}, see \citet{liboschik17}, and they were originally taken from SurvStat{@}RKI~2.0 at \url{https://survstat.rki.de/}. Plots of this time series, $x_1,\ldots,x_n$ of length $n=646$, as well as of its sample ACF are shown in Figure~\ref{figEcoli}. We recognize a slow decay of the ACF, and also the partial ACF (not displayed) exhibits a gradual fading. This supports the recommendation of \citet{aknouche22} to use an INGARCH$(1,1)$-type model for these data. Initial model fitting is done by the method of moments (MM), in analogy to Section~\ref{Simulations}. From sample mean and ACF at lags~1--2, recall Example~\ref{examINGARCH11}, we get the estimates $(\hat{a}_{0,\textup{MM}}, \hat{a}_{1,\textup{MM}}, \hat{b}_{1,\textup{MM}}) \approx (2.465, 0.431, 0.448)$. These can be used to compute the scaled residuals $S_t(\hat\ftheta) = X_t/M_t(\hat\ftheta)$, recall Section~\ref{Approaches for Model Diagnostics}, the sample variance of which is given by $\approx 0.121$. As this value is much smaller than one, we conclude that an INGARCH$(1,1)$-CMEM with NB-counting series is not suitable for the Ecoli counts, so we restrict the remaining model fitting to the cases of a Poi- and Bin-counting series. Using the approach of Section~\ref{Estimation-of-sigma}, the corresponding initial estimates of~$\sigma^2$ are~$0.068$ (Poi-case) and~$0.120$ (Bin-case), respectively.

\begin{table}[t]
\centering\small
\caption{Ecoli counts from Section~\ref{Disease Counts}: Estimation results from different methods, approximate SEs in parentheses.}
\label{tabEcoliEst}

\smallskip
\begin{tabular}{r@{\qquad}cccc@{\qquad}cccc}
\toprule
 & \multicolumn{4}{c}{Poi-counting series \eqref{CompCountMEMdef}} & \multicolumn{4}{c}{Bin-counting series \eqref{CompCountMEMdef2}} \\
 & $a_0$ & $a_1$ & $b_1$ & $\sigma^2$ & $a_0$ & $a_1$ & $b_1$ & $\sigma^2$ \\
\midrule
PQ & 2.887 & 0.378 & 0.481 & 0.063 & 2.887 & 0.378 & 0.481 & 0.115 \\
 & (0.620) & (0.040) & (0.055) & (0.012) & (0.649) & (0.043) & (0.057) & (0.012) \\
\midrule
NQ & 3.054 & 0.337 & 0.512 & 0.063 & 3.054 & 0.337 & 0.512 & 0.115 \\
 & (0.616) & (0.038) & (0.055) & (0.012) & (0.577) & (0.037) & (0.052) & (0.012) \\
\midrule
WQ & 3.081 & 0.336 & 0.511 & 0.063 & 3.081 & 0.336 & 0.511 & 0.114 \\
 & (0.626) & (0.038) & (0.055) & (0.012) & (0.580) & (0.037) & (0.053) & (0.012) \\
\midrule
2W & 2.938 & 0.351 & 0.505 & 0.063 & 3.084 & 0.339 & 0.508 & 0.114 \\
 & (0.590) & (0.038) & (0.053) & (0.012) & (0.581) & (0.037) & (0.053) & (0.012) \\
\bottomrule
\end{tabular}
\end{table}

\smallskip
Using these initial estimates, the final estimates according to the methods PQ, NQ, EQ, and 2W are computed, see Table~\ref{tabEcoliEst} for an overview. Note that the methods PQ, NQ, and EQ lead to the same estimates for $a_0, a_1, b_1$ for both models \eqref{CompCountMEMdef} and \eqref{CompCountMEMdef2}, whereas the estimates of~$\sigma^2$ and the approximate SEs differ. The method 2W, by contrast, leads to different estimates throughout. According to Table~\ref{tabEcoliEst}, the NQ- and EQ-estimates are quite similar to each other for a given type of model (and for \eqref{CompCountMEMdef2}, they are also close to the 2W-estimates), but they differ visibly from the PQ-estimates. This goes along with our findings from Section~\ref{Simulations}, where we noted the same (dis)agreements in Tables~S.3 and~S.4 (right blocks, as~$\sigma^2$ is estimated to about~0.1). Also from these tables, recalling that $n=646$ for the Ecoli counts, we expect the estimates of $a_0, a_1, b_1$ to be biased. For $\sigma^2$, by contrast, we do not expect a notable bias. Its estimate is larger for model \eqref{CompCountMEMdef2} than for \eqref{CompCountMEMdef}, which is plausible as the Bin-counting series causes less variation than the Poi-one, so more variation has to be caused by the innovations~$\epsilon_t$.

\begin{table}[t]
\centering\small
\caption{Ecoli counts from Section~\ref{Disease Counts}: Mean, variance, and ACF at lags~1--5: sample properties vs.\ properties of fitted models (with Poi- or Bin-counting series).}
\label{tabEcoliProp}

\smallskip
\begin{tabular}{r@{\qquad}ccccccc}
\toprule
 & Mean & Var & $\rho(1)$ & $\rho(2)$ & $\rho(3)$ & $\rho(4)$ & $\rho(5)$ \\
\midrule
Sample & 20.334 & 88.753 & 0.632 & 0.555 & 0.471 & 0.426 & 0.343 \\
\midrule
Poi, PQ & 20.462 & 75.261 & 0.547 & 0.470 & 0.404 & 0.347 & 0.298 \\
NQ & 20.223 & 66.333 & 0.485 & 0.412 & 0.350 & 0.297 & 0.252 \\
EQ & 20.214 & 65.992 & 0.483 & 0.409 & 0.347 & 0.294 & 0.249 \\
2W & 20.306 & 69.539 & 0.509 & 0.435 & 0.372 & 0.319 & 0.272 \\
\midrule
Bin, PQ & 20.462 & $[79.56,79.97]$ & 0.547 & 0.470 & 0.404 & 0.347 & 0.298 \\
NQ & 20.223 & $[69.13,69.50]$ & 0.485 & 0.412 & 0.350 & 0.297 & 0.252 \\
EQ & 20.214 & $[68.74,69.10]$ & 0.483 & 0.409 & 0.347 & 0.294 & 0.249 \\
2W & 20.223 & $[69.20,69.57]$ & 0.487 & 0.413 & 0.350 & 0.296 & 0.251 \\
\bottomrule
\end{tabular}

\bigskip
%
\caption{Ecoli counts from Section~\ref{Disease Counts}: model  diagnostics according to Section~\ref{Approaches for Model Diagnostics}, applied to CMEMs and some competing models.}
\label{tabEcoliDiag}

\smallskip
\resizebox{\linewidth}{!}{
\begin{tabular}{r@{\qquad}cccc@{\qquad}cccc@{\qquad}cccc}
\toprule
 & \multicolumn{4}{c}{Poi-counting series \eqref{CompCountMEMdef}} & \multicolumn{4}{c}{Bin-counting series \eqref{CompCountMEMdef2}} & \multicolumn{4}{c}{Model \eqref{mul-thin}} \\
 & MAR & MSR & VSR & MSPR & MAR & MSR & VSR & MSPR & MAR & MSR & VSR & MSPR \\
\midrule
PQ & 5.154 & 1.000 & 0.116 & 0.989 & 5.154 & 1.000 & 0.116 & 1.000 & \it 5.166 & \it 1.051 & \it 0.215 & \it 0.999 \\
NQ & 5.143 & 1.000 & 0.115 & 0.995 & 5.143 & 1.000 & 0.115 & 1.000 & \it 5.150 & \it 1.054 & \it 0.218 & \it 1.004 \\
EQ & 5.143 & 1.000 & 0.115 & 0.995 & 5.143 & 1.000 & 0.115 & 1.000 & \it 5.150 & \it 1.053 & \it 0.217 & \it 1.004 \\
2W & 5.145 & 1.000 & 0.115 & 0.992 & 5.144 & 1.000 & 0.115 & 1.000 & \it 5.154 & \it 1.053 & \it 0.216 & \it 1.002 \\
\midrule
 & \multicolumn{4}{c}{Poi-INGARCH$(1,1)$} & \multicolumn{4}{c}{NB-INGARCH$(1,1)$} \\
\cmidrule{1-9}
CML & 5.154 & 1.000 & 0.116 & 2.267 & 5.144 & 1.000 & 0.116 & 1.035 \\
\bottomrule
\multicolumn{13}{l}{\it\footnotesize {\bfseries Note:} Italic values taken from \citet[Tables~III--IV]{aknouche22}.}
\end{tabular}}
\end{table}

\smallskip
To check the adequacy of the fitted models, we consider the approaches discussed in Section~\ref{Approaches for Model Diagnostics}. In Table~\ref{tabEcoliProp}, marginal mean and variance as well as the ACF at lags~1--5 of the fitted models are compared with the respective sample values. Here, it should be recalled that for the Bin-model \eqref{CompCountMEMdef2}, we do not have a closed-form expression, but only a lower and upper bound, see Section~\ref{Count-MEMs using a Binomial Multiplicative Operator}. But the resulting intervals in the lower part of Table~\ref{tabEcoliProp} are quite narrow and, thus, provide reasonable insight into the true variance of these models. While the means agree very well, the fitted models' variances are somewhat lower than the sample variance, where a closer agreement is reached for the PQ-fitted model. The fitted models' ACF is slowly decaying like the sample ACF, but its actual values are too low (again, PQ gives a better fit).
The statistics in Table~\ref{tabEcoliDiag} have also been considered by \citet[Tables~III--IV]{aknouche22} for their type of multiplicative count model \eqref{mul-thin}, recall Remark~\ref{remCMEMaknouche}. Therefore, we provide the corresponding values for MAR, MSR, VSR, and MSPR from \citet[Tables~III--IV]{aknouche22} within Table~\ref{tabEcoliDiag} for comparison. As further competitors, we fit the (fully parametrized) Poi-INGARCH$(1,1)$ model of \citet{ferland06} and NB-INGARCH$(1,1)$ model of \citet{zhu11} by conditional maximum likelihood (CML) estimation, see the lower part of Table~\ref{tabEcoliDiag}.
The MARs in Table~\ref{tabEcoliDiag} are lower than the MARs for the respective fit of model \eqref{mul-thin}, and we get lower MARs for NQ, EQ, 2W than for PQ. The MSPRs for model \eqref{CompCountMEMdef} are again closer to the target value of~1 if using NQ-, EQ- or 2W-estimates, but a nearly perfect agreement is achieved for the Bin-model \eqref{CompCountMEMdef2}. Again, we observe a better MSPR performance than \citet{aknouche22} did for their model fits of \eqref{mul-thin}. Regarding the scaled residuals, the MSR is extremely close to the target value~1 without exception. The VSR values do not differ much between the different model fits and are close to the estimate of~$\sigma^2$ within the Bin-model \eqref{CompCountMEMdef2}. For the fully parametrized INGARCH$(1,1)$ models, the Poi-INGARCH$(1,1)$ model clearly falls short in terms of MSPR, because it is not able to capture the (conditional) variance of the Ecoli counts. The NB-INGARCH$(1,1)$ model is more powerful in this respect, but its MSPR value still shows stronger deviation from the target value of~1. So altogether, the CMEMs show the best overall performance.

\smallskip
To sum up, while the unconditional moments are better met if using PQ-estimates, the conditional properties are better captured by NQ, EQ, and 2W for the Ecoli data. Overall, the model \eqref{CompCountMEMdef2} with Bin-counting series provides a better fit, but with visible deviations regarding marginal variance and ACF. A reason might be the period $t=544,\ldots,550$, where exceptionally large count values~$x_t$ are observed (larger than~50). This segment might affect the obtained estimates (as well as the printed sample properties) beyond the ``natural'' bias that is to be expected for an INGARCH$(1,1)$-CMEM time series of length $n=646$.

\begin{table}[t]
\centering\small
\caption{Ecoli counts from Section~\ref{Disease Counts}: predictive performance for CMEMs and some competing models.}
\label{tabEcoliPred}

\smallskip
\begin{tabular}{r@{\qquad}cccc@{\qquad}cccc}
\toprule
 & \multicolumn{4}{c}{Poi-counting series \eqref{CompCountMEMdef}} & \multicolumn{4}{c}{Bin-counting series \eqref{CompCountMEMdef2}} \\
 & MAR & MSR & VSR & MSPR & MAR & MSR & VSR & MSPR \\
\midrule
PQ & 6.249 & 1.066 & 0.124 & 1.136 & 6.249 & 1.066 & 0.124 & 1.092 \\
NQ & 6.323 & 1.075 & 0.124 & 1.167 & 6.323 & 1.075 & 0.124 & 1.115 \\
EQ & 6.329 & 1.075 & 0.124 & 1.169 & 6.329 & 1.075 & 0.124 & 1.117 \\
2W & 6.278 & 1.071 & 0.124 & 1.152 & 6.323 & 1.075 & 0.124 & 1.116 \\
\midrule
 & \multicolumn{4}{c}{Poi-INGARCH$(1,1)$} & \multicolumn{4}{c}{NB-INGARCH$(1,1)$} \\
\midrule
CML & 6.249 & 1.066 & 0.124 & 2.812 & 6.272 & 1.070 & 0.123 & 1.201 \\
\bottomrule
\end{tabular}
\end{table}

\bigskip
Finally, let us analyze the forecast performance of the fitted CMEMs in comparison to the competing Poi- and NB-INGARCH$(1,1)$ model. In Table~\ref{tabEcoliDiag}, we applied the MAR, MSR, VSR, and MSPR to the same data that was also used for model fitting, \ie they provide insight into the goodness of fit. Next, we apply these criteria to 84 new Ecoli counts, taken from SurvStat{@}RKI~2.0 at \url{https://survstat.rki.de/} for the remainder of 2013 as well as for the year 2014, to get insight into the predictive abilities of the fitted models. Results are summarized in Table~\ref{tabEcoliPred}. While all models perform similar in terms of MSR and VSR, the best MAR performance is achieved for the PQ-fit of both types of CMEM as well as for the Poi-INGARCH$(1,1)$ model. The latter, however, again falls short in terms of the MSPR. The NB-INGARCH$(1,1)$ model has a MSPR value much closer to~1, but worse than for any of the CMEMs. Among the CMEMs, the models with Bin-counting series have a better MSPR performance than their Poi-counterparts, and the MSPR closest to~1 is achieved again for the PQ-fit. Hence, the PQ-fitted CMEM \eqref{CompCountMEMdef2} with Bin-counting series shows the best predictive performance regarding the 84~Ecoli counts from June~2013 until December~2014.

\begin{figure}[t]
\centering\small
(a)\hspace{-3ex}\includegraphics[viewport=0 45 405 235, clip=, scale=0.7]{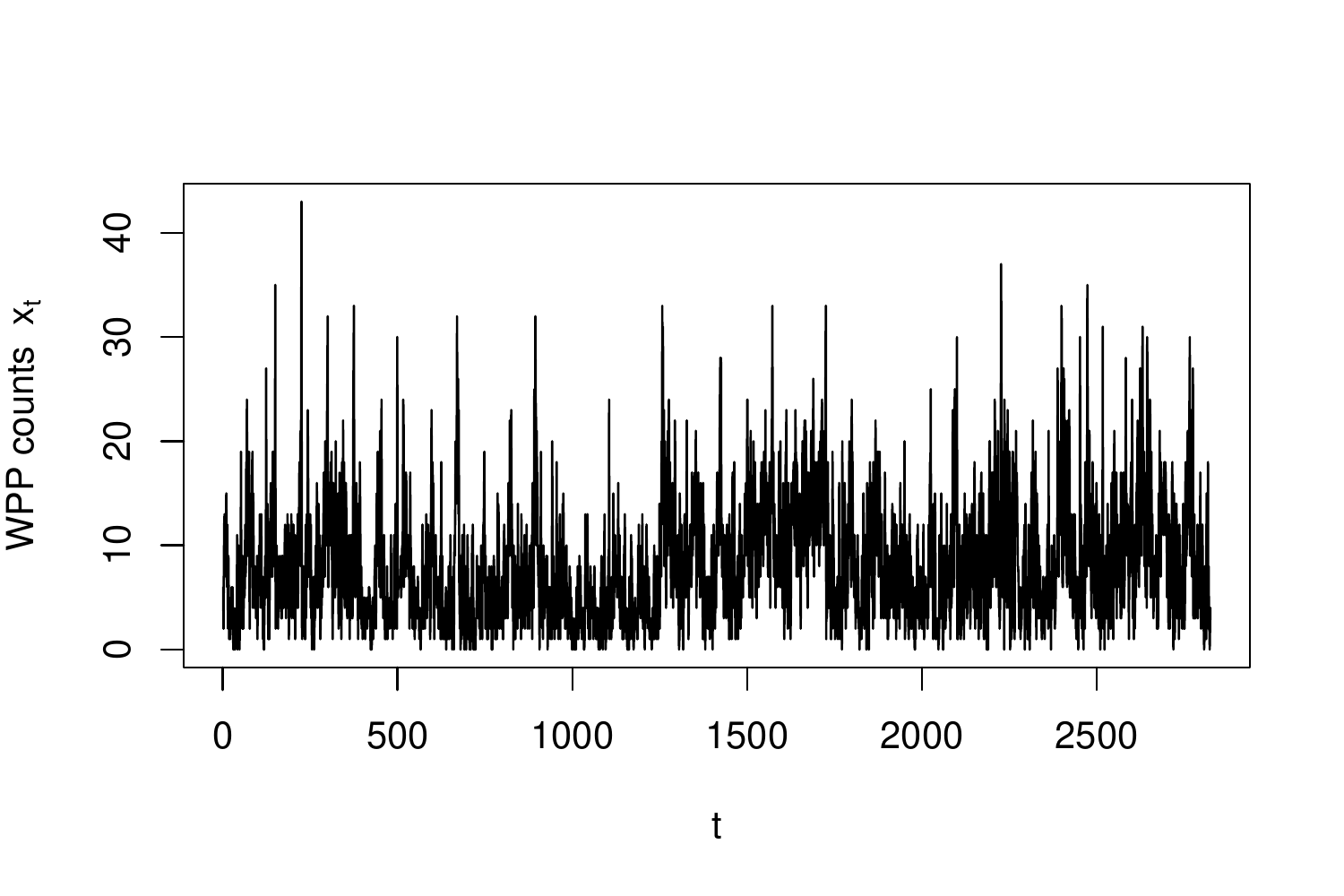}\hspace{-1ex}$t$
\qquad
(b)\hspace{-3ex}\includegraphics[viewport=0 45 190 235, clip=, scale=0.7]{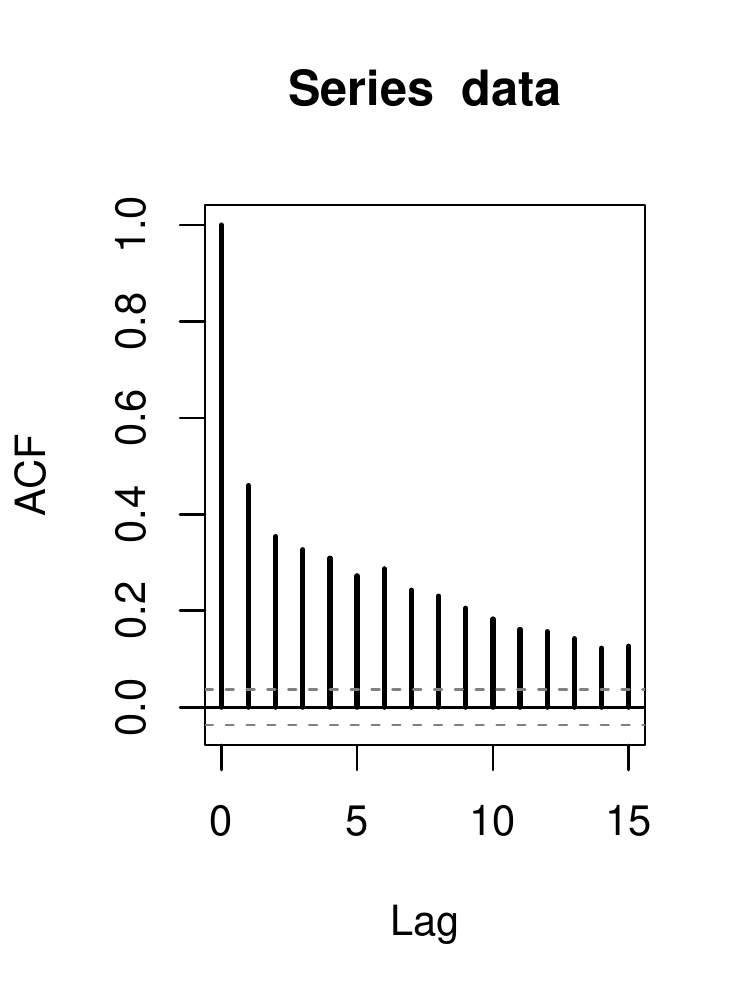}\ $k$
\caption{WPP counts from Section~\ref{Transaction Counts}: (a) time series plot and (b) sample ACF against lag~$k$.}
\label{figWPP}
\end{figure}

\subsection{Transaction Counts}
\label{Transaction Counts}
Besides the surveillance of infectious diseases, another common application area for count time series is financial transactions. Here, we consider one of the time series provided by \citet{aknoucheetal22}, namely the number of stock transactions concerning the Wausau Paper Corporation (WPP), measured in 5-min intervals between
9:45\,AM and 4:00\,PM for the period from January~3 to February~18 in 2005.
More precisely, in analogy to \citet[Section~6.2]{aknoucheetal22}, we leave out the last 100 observations for a later forecast experiment.
The resulting count time series $x_1,\ldots,x_n$ of length $n=2825$ is plotted in Figure~\ref{figWPP} together with its sample ACF. The ACF is slowly decaying, which indicates that an INGARCH$(1,1)$-type model with a large value for the feedback parameter~$b_1$ could be appropriate for these data, also see \citet{aknoucheetal22}.
In analogy to Section~\ref{Disease Counts}, we start with an initial analysis of the WPP counts. The moment estimates regarding the conditional mean \eqref{INGARCHmodels} are $(\hat{a}_{0,\textup{MM}}, \hat{a}_{1,\textup{MM}}, \hat{b}_{1,\textup{MM}}) \approx (1.874, 0.361, 0.409)$, which are used to compute the scaled residuals and the sample variance thereof. The obtained value $\approx 0.437$ is again much smaller than one, so an NB-counting series is not suitable and we restrict the following model fitting to INGARCH$(1,1)$-CMEMs with Poi- and Bin-counting series. The corresponding initial estimates of~$\sigma^2$ are~$0.299$ (Poi-case) and~$0.433$ (Bin-case), respectively.

\begin{table}[t]
\centering\small
\caption{WPP counts from Section~\ref{Transaction Counts}: Estimation results from different methods, approximate SEs in parentheses.}
\label{tabWPPEst}

\smallskip
\begin{tabular}{r@{\qquad}cccc@{\qquad}cccc}
\toprule
 & \multicolumn{4}{c}{Poi-counting series \eqref{CompCountMEMdef}} & \multicolumn{4}{c}{Bin-counting series \eqref{CompCountMEMdef2}} \\
 & $a_0$ & $a_1$ & $b_1$ & $\sigma^2$ & $a_0$ & $a_1$ & $b_1$ & $\sigma^2$ \\
\midrule
PQ & 0.792 & 0.268 & 0.634 & 0.301 & 0.792 & 0.268 & 0.634 & 0.439 \\
 & (0.108) & (0.018) & (0.026) & (0.022) & (0.110) & (0.019) & (0.027) & (0.022) \\
\midrule
NQ & 0.768 & 0.266 & 0.640 & 0.302 & 0.768 & 0.266 & 0.640 & 0.439 \\
 & (0.103) & (0.017) & (0.025) & (0.022) & (0.100) & (0.017) & (0.025) & (0.022) \\
\midrule
WQ & 0.769 & 0.265 & 0.640 & 0.301 & 0.769 & 0.265 & 0.640 & 0.439 \\
 & (0.104) & (0.017) & (0.025) & (0.022) & (0.100) & (0.017) & (0.025) & (0.022) \\
\midrule
2W & 0.769 & 0.266 & 0.640 & 0.302 & 0.769 & 0.265 & 0.640 & 0.439 \\
 & (0.103) & (0.017) & (0.025) & (0.022) & (0.100) & (0.017) & (0.025) & (0.022) \\
\bottomrule
\end{tabular}
\end{table}

\smallskip
The final estimates together with their approximate SEs are summarized in Table~\ref{tabWPPEst}. Because of the very large sample size, namely $n=2825$, we do not expect any bias problems in estimation, recall our simulation study in Section~\ref{Simulations}. In fact, all estimation methods lead to very similar results. The only visible deviations are observed for the intercept parameter~$a_0$ on the one hand, where~PQ leads to a slightly larger estimate than the other methods, and for $\sigma^2$ on the other hand, which has larger estimates for the Bin- than for the Poi-counting series (for the same reason as explained in Section~\ref{Disease Counts}). It is worth mentioning that our final estimates are close to the ones obtained by \citet[Supplement~1.2]{aknoucheetal22}, although these authors used fully parametric INGARCH$(1,1)$ models. In particular, we note a very close agreement to their ``NB2-INGARCH model'' in Table~3, which is the NB-INGARCH$(1,1)$ model of \citet{zhu11} and constitutes their recommended model choice.

\begin{table}[t]
\centering\small
\caption{WPP counts from Section~\ref{Transaction Counts}: Mean, variance, and ACF at lags~1--5: sample properties vs.\ properties of fitted models (with Poi- or Bin-counting series).}
\label{tabWPPProp}

\smallskip
\begin{tabular}{r@{\qquad}ccccccc}
\toprule
 & Mean & Var & $\rho(1)$ & $\rho(2)$ & $\rho(3)$ & $\rho(4)$ & $\rho(5)$ \\
\midrule
Sample & 8.144 & 35.976 & 0.460 & 0.354 & 0.327 & 0.309 & 0.273 \\
\midrule
Poi, PQ & 8.150 & 44.375 & 0.446 & 0.403 & 0.364 & 0.329 & 0.297 \\
NQ & 8.164 & 44.697 & 0.446 & 0.404 & 0.366 & 0.332 & 0.301 \\
EQ & 8.163 & 44.613 & 0.446 & 0.404 & 0.366 & 0.331 & 0.300 \\
2W & 8.166 & 44.786 & 0.447 & 0.405 & 0.367 & 0.332 & 0.301 \\
\midrule
Bin, PQ & 8.150 & $[48.92,49.34]$ & 0.446 & 0.403 & 0.364 & 0.329 & 0.297 \\
NQ & 8.164 & $[49.34,49.76]$ & 0.446 & 0.404 & 0.366 & 0.332 & 0.301 \\
EQ & 8.163 & $[49.23,49.65]$ & 0.446 & 0.404 & 0.366 & 0.331 & 0.300 \\
2W & 8.161 & $[49.24,49.66]$ & 0.446 & 0.404 & 0.366 & 0.331 & 0.300 \\
\bottomrule
\end{tabular}

\bigskip
%
\caption{WPP counts from Section~\ref{Transaction Counts}: model  diagnostics according to Section~\ref{Approaches for Model Diagnostics}, applied to CMEMs and some competing models.}
\label{tabWPPDiag}

\smallskip
\begin{tabular}{r@{\qquad}cccc@{\qquad}cccc}
\toprule
 & \multicolumn{4}{c}{Poi-counting series \eqref{CompCountMEMdef}} & \multicolumn{4}{c}{Bin-counting series \eqref{CompCountMEMdef2}} \\
 & MAR & MSR & VSR & MSPR & MAR & MSR & VSR & MSPR \\
\midrule
PQ & 3.916 & 1.000 & 0.443 & 0.975 & 3.916 & 1.000 & 0.443 & 0.998 \\
NQ & 3.916 & 1.000 & 0.443 & 0.975 & 3.916 & 1.000 & 0.443 & 0.998 \\
EQ & 3.916 & 1.000 & 0.443 & 0.975 & 3.916 & 1.000 & 0.443 & 0.998 \\
2W & 3.916 & 1.000 & 0.443 & 0.975 & 3.916 & 1.000 & 0.443 & 0.998 \\
\midrule
 & \multicolumn{4}{c}{Poi-INGARCH$(1,1)$} & \multicolumn{4}{c}{NB-INGARCH$(1,1)$} \\
\midrule
CML & 3.916 & 1.000 & 0.443 & 3.181 & 3.916 & 1.000 & 0.444 & 1.036 \\
\bottomrule
\end{tabular}
\end{table}

\smallskip
Thus, in analogy to Section~\ref{Disease Counts}, we also fit the (fully parametrized) Poi-INGARCH$(1,1)$ model of \citet{ferland06} and NB-INGARCH$(1,1)$ model of \citet{zhu11} as further competitors.
Looking at the adequacy checks in Tables~\ref{tabWPPProp}--\ref{tabWPPDiag}, we generally recognize a very good performance of the CMEMs with only minor deviations between the different estimation methods. Table~\ref{tabWPPProp} shows that the fitted CMEMs' mean and ACF closely agree with the corresponding sample properties. Only for the marginal variance, we observe an exceedance for the fitted models, which is less pronounced in case of the Poi-counting series \eqref{CompCountMEMdef}. The opposite ranking is implied by the MSPR values in Table~\ref{tabWPPDiag}, \ie the conditional variance is better captured by the Bin-counting series \eqref{CompCountMEMdef2}. But in any case, the differences between the different models are rather small. Compared with the CML-fitted Poi- and NB-INGARCH$(1,1)$ model, we note that the Poi-INGARCH$(1,1)$ model fails again in terms of the MSPR, which is much larger than~1. The NB-INGARCH$(1,1)$ model does clearly better, but still shows a stronger deviation of the MSPR from~1 than any of the CMEMs.
Altogether, any of the fitted CMEMs performs reasonably well, where the Poi-counting series \eqref{CompCountMEMdef} is advantageous regarding the marginal variance, and the Bin-counting series \eqref{CompCountMEMdef2} regarding the conditional variance.

\begin{table}[t]
\centering\small
\caption{WPP counts from Section~\ref{Transaction Counts}: predictive performance for CMEMs and some competing models.}
\label{tabWPPPred}

\smallskip
\begin{tabular}{r@{\qquad}cccc@{\qquad}cccc}
\toprule
 & \multicolumn{4}{c}{Poi-counting series \eqref{CompCountMEMdef}} & \multicolumn{4}{c}{Bin-counting series \eqref{CompCountMEMdef2}} \\
 & MAR & MSR & VSR & MSPR & MAR & MSR & VSR & MSPR \\
\midrule
PQ & 3.613 & 1.026 & 0.552 & 1.164 & 3.613 & 1.026 & 0.552 & 1.232 \\
NQ & 3.616 & 1.026 & 0.553 & 1.165 & 3.616 & 1.026 & 0.553 & 1.233 \\
EQ & 3.616 & 1.026 & 0.553 & 1.165 & 3.616 & 1.026 & 0.553 & 1.233 \\
2W & 3.616 & 1.026 & 0.553 & 1.166 & 3.616 & 1.026 & 0.553 & 1.233 \\
\midrule
 & \multicolumn{4}{c}{Poi-INGARCH$(1,1)$} & \multicolumn{4}{c}{NB-INGARCH$(1,1)$} \\
\midrule
CML & 3.610 & 1.025 & 0.551 & 3.400 & 3.618 & 1.027 & 0.554 & 1.236 \\
\bottomrule
\end{tabular}
\end{table}

\bigskip
Finally, like in Section~\ref{Disease Counts}, we do a forecast experiment based on the WPP's last 100 observations that were left out during model fitting, see Table~\ref{tabWPPPred} for the results. The CMEMs have rather similar MAR, MSR, and VSR values, but the MSPR of the CMEMs with Poi-counting series is closer to~1, with the best predictive performance for the PQ-fitted CMEM. Regarding the competitors, the Poi-INGARCH$(1,1)$ model fails again in terms of the MSPR, while the NB-INGARCH$(1,1)$ model performs slightly worse than the preferred CMEM with respect to all four criteria. Hence, the PQ-fitted CMEM \eqref{CompCountMEMdef} with Poi-counting series shows the best predictive performance regarding the last 100 WPP counts.

\section{Conclusions and Future Research}
\label{Conclusions}
In this article, we proposed a novel framework for transferring the well-established MEMs for real-valued time series to the case of count time series. The proposed CMEMs make use of an appropriate multiplicative operator for counts. After discussing such operators and the corresponding CMEMs in general, we focused on two special cases, namely a compounding operator using, \eg a Poi-counting series, and a binomial multiplicative operator. We established the existence of stationary and ergodic solutions to such INGARCH-CMEMs and derived relevant moment properties. Afterwards, we developed two approaches for semi-parametric model estimation, namely QMLE and 2WLSE. The finite-sample performance of these estimation approaches was analyzed in a simulation study, and the applicability of our novel CMEMs was demonstrated with two real-world data examples.

\medskip
There are several directions for future research. First, CMEMs using another multiplicative operator (such as the ZIP-multiplicative operator mentioned in the end of
Section~\ref{Multiplicative Operators for Counts}) or a non-linear mean specification might be analyzed. Second, the semi-parametric INGARCH model sketched in the end of
Section~\ref{Count-MEMs using a Compounding Operator} deserves further investigation, \eg with respect to the question if a semi-parametric maximum likelihood (ML) approach, analogous to \citet{drost09} and \cite{liu21},
is possible for INGARCH models. Such an ML-fitted CMEM would be attractive for applications as it would allow, for example, for density forecasts like in \citet{aknoucheetal22}.
Third, related to the previous research direction, it could also be studied if a combination of semi-parametric ML estimation with the smoothing method discussed
by \citet{faymonville22} could be beneficial for practice.
Finally, $\mathbb{Z}$-valued GARCH models also get attention, see \citet{cui21} and \cite{xu22} among others, so how to generalize the CMEM framework to the $\mathbb{Z}$-valued case is interesting.

\subsubsection*{Acknowledgements}
The authors thank the associate editor and the two referees for their useful and constructive comments on an earlier draft of this article.
Zhu's work is supported by National Natural Science Foundation of China (No.\ 12271206) and Natural Science Foundation of
Jilin Province (No.\ 20210101143JC), and the Science and Technology Research Planning Project of Jilin Provincial
Department of Education (No.\ JJKH20231122KJ).

\newpage


\newcommand{\genstirlingII}[3]{%
	\genfrac{\{}{\}}{0pt}{#1}{#2}{#3}%
}
\newcommand{\stirlingII}[2]{\genstirlingII{}{#1}{#2}}
\newcommand{\dstirlingII}[2]{\genstirlingII{0}{#1}{#2}}

\renewcommand{\contentsname}{S\quad Supplement}




\parindent 0cm

\part*{\centering\vspace{-7ex}Supplementary Material:}
\section*{\centering Conditional-mean Multiplicative Operator Models\\ for Count Time Series}

\maketitle

\appendix
\setcounter{section}{18}
\numberwithin{table}{section}

\section{Detailed Simulation Results}
\label{Supplement}
In what follows, we present tables with detailed simulation results, which correspond to the simulation study presented in Section~5 of the main manuscript. The results rely on 10,000 replications per scenario. To achieve robustness with respect to possible outliers, we used the trimmed versions of sample mean and standard deviation, where the smallest and largest 0.1\,\% of simulated values are omitted during calculation. We observed such outliers only in very few simulation runs, but these would have a strongly distorting influence on ordinary sample mean and standard deviation.

%
%
%


\bigskip
\bigskip

\subsubsection*{Acronyms}
Bin = binomial

NB = negative binomial

Poi = Poisson

\medskip
DGP = data-generating process

CMEM = count multiplicative error model

INGARCH = integer-valued generalized autoregressive conditional heteroscedasticity

\medskip
QMLE = quasi-maximum likelihood estimation

PQ / NQ / EQ = Poi- / NB- / exponential QMLE

2W, 2SWLSE = two-stage weighted least squares estimation

\newpage


%

%

\begin{table}[th!]
\centering
\caption{Different estimation approaches applied to simulated count time series from INGARCH$(1,1)$-CMEM with Poi-counting series: Mean, simulated standard error (SSE), and mean approximated standard error (ASE) for different sample sizes~$n$.}
\label{tabDGP1Poi}

\smallskip
\resizebox{\linewidth}{!}{
\begin{tabular}{lrcrrrrcrrrr}
\toprule
$n$ &  & Param. & PQ & NQ & EQ & 2W & Param. & PQ & NQ & EQ & 2W \\
\midrule
300 & Mean & $a_0=2.8$ & 2.997 & 2.922 & 2.915 & 2.935 & $a_0=2.8$ & 2.965 & 2.936 & 2.932 & 2.942 \\
 & SSE &  & 0.722 & 0.686 & 0.688 & 0.701 &  & 0.681 & 0.668 & 0.671 & 0.666 \\
 & ASE &  & 0.780 & 0.694 & 0.692 & 0.699 &  & 0.720 & 0.686 & 0.687 & 0.687 \\[1ex]
600 & Mean &  & 2.904 & 2.858 & 2.854 & 2.866 &  & 2.889 & 2.875 & 2.874 & 2.877 \\
 & SSE &  & 0.529 & 0.492 & 0.493 & 0.500 &  & 0.488 & 0.474 & 0.476 & 0.473 \\
 & ASE &  & 0.536 & 0.480 & 0.479 & 0.481 &  & 0.498 & 0.477 & 0.478 & 0.477 \\[1ex]
1000 & Mean &  & 2.871 & 2.839 & 2.837 & 2.843 &  & 2.855 & 2.844 & 2.843 & 2.845 \\
 & SSE &  & 0.406 & 0.375 & 0.376 & 0.378 &  & 0.380 & 0.366 & 0.367 & 0.365 \\
 & ASE &  & 0.412 & 0.369 & 0.369 & 0.370 &  & 0.382 & 0.366 & 0.367 & 0.366 \\
\midrule
300 & Mean & $a_1=0.4$ & 0.384 & 0.393 & 0.394 & 0.392 & $a_1=0.4$ & 0.392 & 0.395 & 0.396 & 0.395 \\
 & SSE &  & 0.081 & 0.077 & 0.077 & 0.079 &  & 0.062 & 0.060 & 0.060 & 0.060 \\
 & ASE &  & 0.082 & 0.074 & 0.074 & 0.074 &  & 0.061 & 0.059 & 0.059 & 0.059 \\[1ex]
600 & Mean &  & 0.392 & 0.397 & 0.397 & 0.397 &  & 0.397 & 0.398 & 0.399 & 0.398 \\
 & SSE &  & 0.057 & 0.053 & 0.053 & 0.053 &  & 0.044 & 0.042 & 0.043 & 0.042 \\
 & ASE &  & 0.059 & 0.052 & 0.052 & 0.052 &  & 0.043 & 0.042 & 0.042 & 0.042 \\[1ex]
1000 & Mean &  & 0.395 & 0.399 & 0.399 & 0.398 &  & 0.398 & 0.399 & 0.399 & 0.399 \\
 & SSE &  & 0.045 & 0.041 & 0.041 & 0.041 &  & 0.034 & 0.033 & 0.033 & 0.033 \\
 & ASE &  & 0.046 & 0.041 & 0.041 & 0.041 &  & 0.034 & 0.032 & 0.032 & 0.032 \\
\midrule
300 & Mean & $b_1=0.2$ & 0.180 & 0.184 & 0.185 & 0.183 & $b_1=0.2$ & 0.182 & 0.183 & 0.184 & 0.183 \\
 & SSE &  & 0.121 & 0.116 & 0.117 & 0.120 &  & 0.112 & 0.110 & 0.111 & 0.110 \\
 & ASE &  & 0.128 & 0.116 & 0.115 & 0.116 &  & 0.119 & 0.114 & 0.114 & 0.114 \\[1ex]
600 & Mean &  & 0.189 & 0.192 & 0.193 & 0.191 &  & 0.190 & 0.190 & 0.190 & 0.190 \\
 & SSE &  & 0.089 & 0.083 & 0.084 & 0.085 &  & 0.080 & 0.079 & 0.079 & 0.079 \\
 & ASE &  & 0.088 & 0.080 & 0.080 & 0.081 &  & 0.082 & 0.079 & 0.080 & 0.079 \\[1ex]
1000 & Mean &  & 0.192 & 0.194 & 0.195 & 0.194 &  & 0.194 & 0.194 & 0.194 & 0.194 \\
 & SSE &  & 0.068 & 0.064 & 0.064 & 0.064 &  & 0.063 & 0.061 & 0.062 & 0.061 \\
 & ASE &  & 0.068 & 0.062 & 0.062 & 0.062 &  & 0.063 & 0.061 & 0.061 & 0.061 \\
\midrule
300 & Mean & $\sigma^2=1$ & 1.001 & 0.998 & 0.998 & 1.002 & $\sigma^2=0.4$ & 0.398 & 0.397 & 0.397 & 0.397 \\
 & SSE &  & 0.114 & 0.113 & 0.113 & 0.130 &  & 0.050 & 0.050 & 0.050 & 0.050 \\
 & ASE &  & 0.128 & 0.127 & 0.127 & 0.128 &  & 0.046 & 0.046 & 0.046 & 0.046 \\[1ex]
600 & Mean &  & 1.002 & 0.999 & 0.999 & 1.000 &  & 0.400 & 0.399 & 0.399 & 0.399 \\
 & SSE &  & 0.080 & 0.080 & 0.080 & 0.081 &  & 0.035 & 0.035 & 0.035 & 0.035 \\
 & ASE &  & 0.092 & 0.091 & 0.091 & 0.091 &  & 0.033 & 0.033 & 0.033 & 0.033 \\[1ex]
1000 & Mean &  & 1.000 & 0.999 & 0.999 & 0.999 &  & 0.399 & 0.399 & 0.399 & 0.399 \\
 & SSE &  & 0.061 & 0.061 & 0.061 & 0.062 &  & 0.027 & 0.027 & 0.027 & 0.027 \\
 & ASE &  & 0.071 & 0.071 & 0.071 & 0.071 &  & 0.025 & 0.025 & 0.025 & 0.025 \\
\bottomrule
\end{tabular}}
\end{table}

\newpage

\begin{table}[th!]
\centering
\caption{Different estimation approaches applied to simulated count time series from INGARCH$(1,1)$-CMEM with Bin-counting series: Mean, simulated standard error (SSE), and mean approximated standard error (ASE) for different sample sizes~$n$.}
\label{tabDGP1Bin}

\smallskip
\resizebox{\linewidth}{!}{
\begin{tabular}{lrcrrrrcrrrr}
\toprule
$n$ &  & Param. & PQ & NQ & EQ & 2W & Param. & PQ & NQ & EQ & 2W \\
\midrule
300 & Mean & $a_0=2.8$ & 2.995 & 2.926 & 2.920 & 2.938 & $a_0=2.8$ & 2.967 & 2.942 & 2.939 & 2.942 \\
 & SSE &  & 0.713 & 0.667 & 0.667 & 0.676 &  & 0.655 & 0.629 & 0.629 & 0.629 \\
 & ASE &  & 0.768 & 0.673 & 0.670 & 0.673 &  & 0.693 & 0.650 & 0.648 & 0.649 \\[1ex]
600 & Mean &  & 2.913 & 2.869 & 2.865 & 2.875 &  & 2.892 & 2.879 & 2.878 & 2.878 \\
 & SSE &  & 0.520 & 0.477 & 0.477 & 0.483 &  & 0.479 & 0.454 & 0.453 & 0.453 \\
 & ASE &  & 0.531 & 0.467 & 0.466 & 0.466 &  & 0.480 & 0.451 & 0.451 & 0.451 \\[1ex]
1000 & Mean &  & 2.875 & 2.844 & 2.842 & 2.847 &  & 2.865 & 2.856 & 2.855 & 2.856 \\
 & SSE &  & 0.406 & 0.369 & 0.368 & 0.371 &  & 0.366 & 0.347 & 0.347 & 0.346 \\
 & ASE &  & 0.407 & 0.359 & 0.358 & 0.358 &  & 0.369 & 0.347 & 0.347 & 0.347 \\
\midrule
300 & Mean & $a_1=0.4$ & 0.384 & 0.393 & 0.394 & 0.393 & $a_1=0.4$ & 0.392 & 0.394 & 0.395 & 0.395 \\
 & SSE &  & 0.077 & 0.072 & 0.072 & 0.072 &  & 0.056 & 0.054 & 0.054 & 0.054 \\
 & ASE &  & 0.079 & 0.070 & 0.070 & 0.070 &  & 0.056 & 0.053 & 0.053 & 0.053 \\[1ex]
600 & Mean &  & 0.391 & 0.396 & 0.396 & 0.396 &  & 0.396 & 0.398 & 0.398 & 0.398 \\
 & SSE &  & 0.056 & 0.050 & 0.050 & 0.050 &  & 0.040 & 0.038 & 0.038 & 0.038 \\
 & ASE &  & 0.057 & 0.050 & 0.050 & 0.050 &  & 0.040 & 0.037 & 0.037 & 0.037 \\[1ex]
1000 & Mean &  & 0.395 & 0.398 & 0.398 & 0.398 &  & 0.398 & 0.399 & 0.399 & 0.399 \\
 & SSE &  & 0.043 & 0.038 & 0.038 & 0.038 &  & 0.030 & 0.029 & 0.029 & 0.029 \\
 & ASE &  & 0.044 & 0.039 & 0.039 & 0.039 &  & 0.031 & 0.029 & 0.029 & 0.029 \\
\midrule
300 & Mean & $b_1=0.2$ & 0.181 & 0.185 & 0.185 & 0.183 & $b_1=0.2$ & 0.181 & 0.183 & 0.183 & 0.182 \\
 & SSE &  & 0.118 & 0.112 & 0.112 & 0.114 &  & 0.106 & 0.103 & 0.103 & 0.103 \\
 & ASE &  & 0.124 & 0.112 & 0.111 & 0.111 &  & 0.112 & 0.106 & 0.106 & 0.106 \\[1ex]
600 & Mean &  & 0.189 & 0.192 & 0.192 & 0.191 &  & 0.189 & 0.190 & 0.190 & 0.190 \\
 & SSE &  & 0.086 & 0.080 & 0.080 & 0.081 &  & 0.078 & 0.075 & 0.075 & 0.075 \\
 & ASE &  & 0.086 & 0.078 & 0.078 & 0.078 &  & 0.078 & 0.074 & 0.074 & 0.074 \\[1ex]
1000 & Mean &  & 0.192 & 0.194 & 0.194 & 0.193 &  & 0.192 & 0.193 & 0.193 & 0.193 \\
 & SSE &  & 0.067 & 0.062 & 0.062 & 0.062 &  & 0.059 & 0.057 & 0.057 & 0.057 \\
 & ASE &  & 0.066 & 0.060 & 0.060 & 0.060 &  & 0.060 & 0.057 & 0.057 & 0.057 \\
\midrule
300 & Mean & $\sigma^2=1$ & 1.004 & 1.001 & 1.001 & 1.002 & $\sigma^2=0.4$ & 0.402 & 0.401 & 0.401 & 0.401 \\
 & SSE &  & 0.101 & 0.100 & 0.100 & 0.101 &  & 0.042 & 0.041 & 0.041 & 0.041 \\
 & ASE &  & 0.100 & 0.099 & 0.099 & 0.100 &  & 0.029 & 0.029 & 0.029 & 0.029 \\[1ex]
600 & Mean &  & 1.003 & 1.001 & 1.001 & 1.002 &  & 0.401 & 0.400 & 0.400 & 0.400 \\
 & SSE &  & 0.071 & 0.071 & 0.071 & 0.071 &  & 0.029 & 0.029 & 0.029 & 0.029 \\
 & ASE &  & 0.071 & 0.071 & 0.071 & 0.071 &  & 0.020 & 0.020 & 0.020 & 0.020 \\[1ex]
1000 & Mean &  & 1.003 & 1.001 & 1.001 & 1.001 &  & 0.400 & 0.400 & 0.400 & 0.400 \\
 & SSE &  & 0.054 & 0.053 & 0.053 & 0.054 &  & 0.022 & 0.022 & 0.022 & 0.022 \\
 & ASE &  & 0.055 & 0.055 & 0.055 & 0.055 &  & 0.016 & 0.016 & 0.016 & 0.016 \\
\bottomrule
\end{tabular}}
\end{table}

\newpage

\begin{table}[th!]
\centering
\caption{Different estimation approaches applied to simulated count time series from INGARCH$(1,1)$-CMEM with Poi-counting series: Mean, simulated standard error (SSE), and mean approximated standard error (ASE) for different sample sizes~$n$.}
\label{tabDGP2Poi}

\smallskip
\resizebox{\linewidth}{!}{
\begin{tabular}{lrcrrrrcrrrr}
\toprule
$n$ &  & Param. & PQ & NQ & EQ & 2W & Param. & PQ & NQ & EQ & 2W \\
\midrule
300 & Mean & $a_0=3$ & 3.723 & 3.551 & 3.553 & 3.579 & $a_0=3$ & 3.981 & 4.092 & 4.103 & 3.997 \\
 & SSE &  & 1.398 & 1.243 & 1.247 & 1.356 &  & 1.217 & 1.222 & 1.226 & 1.199 \\
 & ASE &  & 1.304 & 1.075 & 1.075 & 1.083 &  & 1.132 & 1.135 & 1.139 & 1.112 \\[1ex]
600 & Mean &  & 3.363 & 3.256 & 3.258 & 3.265 &  & 3.500 & 3.579 & 3.587 & 3.529 \\
 & SSE &  & 0.876 & 0.737 & 0.738 & 0.786 &  & 0.778 & 0.783 & 0.787 & 0.767 \\
 & ASE &  & 0.877 & 0.711 & 0.712 & 0.713 &  & 0.749 & 0.747 & 0.749 & 0.736 \\[1ex]
1000 & Mean &  & 3.238 & 3.166 & 3.168 & 3.171 &  & 3.290 & 3.354 & 3.361 & 3.318 \\
 & SSE &  & 0.647 & 0.553 & 0.555 & 0.567 &  & 0.574 & 0.580 & 0.582 & 0.566 \\
 & ASE &  & 0.671 & 0.539 & 0.539 & 0.539 &  & 0.563 & 0.560 & 0.562 & 0.553 \\
\midrule
300 & Mean & $a_1=0.35$ & 0.336 & 0.349 & 0.349 & 0.347 & $a_1=0.35$ & 0.361 & 0.365 & 0.366 & 0.363 \\
 & SSE &  & 0.073 & 0.066 & 0.066 & 0.068 &  & 0.050 & 0.050 & 0.050 & 0.050 \\
 & ASE &  & 0.079 & 0.066 & 0.066 & 0.066 &  & 0.051 & 0.050 & 0.051 & 0.050 \\[1ex]
600 & Mean &  & 0.342 & 0.349 & 0.350 & 0.349 &  & 0.358 & 0.361 & 0.361 & 0.360 \\
 & SSE &  & 0.054 & 0.046 & 0.046 & 0.047 &  & 0.035 & 0.034 & 0.034 & 0.034 \\
 & ASE &  & 0.058 & 0.046 & 0.046 & 0.046 &  & 0.036 & 0.036 & 0.036 & 0.035 \\[1ex]
1000 & Mean &  & 0.345 & 0.350 & 0.351 & 0.350 &  & 0.354 & 0.357 & 0.357 & 0.356 \\
 & SSE &  & 0.042 & 0.036 & 0.036 & 0.036 &  & 0.028 & 0.027 & 0.027 & 0.027 \\
 & ASE &  & 0.045 & 0.036 & 0.036 & 0.036 &  & 0.028 & 0.028 & 0.028 & 0.027 \\
\midrule
300 & Mean & $b_1=0.5$ & 0.468 & 0.468 & 0.467 & 0.468 & $b_1=0.5$ & 0.439 & 0.430 & 0.429 & 0.437 \\
 & SSE &  & 0.112 & 0.097 & 0.097 & 0.105 &  & 0.081 & 0.079 & 0.079 & 0.080 \\
 & ASE &  & 0.106 & 0.090 & 0.090 & 0.090 &  & 0.083 & 0.083 & 0.083 & 0.082 \\[1ex]
600 & Mean &  & 0.485 & 0.485 & 0.485 & 0.485 &  & 0.467 & 0.460 & 0.460 & 0.464 \\
 & SSE &  & 0.074 & 0.061 & 0.061 & 0.065 &  & 0.056 & 0.055 & 0.055 & 0.055 \\
 & ASE &  & 0.073 & 0.061 & 0.061 & 0.061 &  & 0.058 & 0.057 & 0.057 & 0.057 \\[1ex]
1000 & Mean &  & 0.490 & 0.490 & 0.490 & 0.490 &  & 0.481 & 0.475 & 0.475 & 0.478 \\
 & SSE &  & 0.055 & 0.046 & 0.046 & 0.048 &  & 0.043 & 0.043 & 0.043 & 0.042 \\
 & ASE &  & 0.056 & 0.047 & 0.047 & 0.047 &  & 0.044 & 0.044 & 0.044 & 0.043 \\
\midrule
300 & Mean & $\sigma^2=1$ & 1.018 & 1.008 & 1.008 & 1.019 & $\sigma^2=0.1$ & 0.101 & 0.100 & 0.100 & 0.101 \\
 & SSE &  & 0.121 & 0.113 & 0.112 & 0.159 &  & 0.020 & 0.020 & 0.020 & 0.020 \\
 & ASE &  & 0.118 & 0.114 & 0.114 & 0.120 &  & 0.020 & 0.020 & 0.020 & 0.020 \\[1ex]
600 & Mean &  & 1.011 & 1.005 & 1.005 & 1.009 &  & 0.101 & 0.100 & 0.100 & 0.100 \\
 & SSE &  & 0.082 & 0.078 & 0.078 & 0.094 &  & 0.014 & 0.014 & 0.014 & 0.014 \\
 & ASE &  & 0.083 & 0.081 & 0.081 & 0.083 &  & 0.014 & 0.014 & 0.014 & 0.014 \\[1ex]
1000 & Mean &  & 1.007 & 1.003 & 1.003 & 1.004 &  & 0.100 & 0.100 & 0.100 & 0.100 \\
 & SSE &  & 0.063 & 0.060 & 0.060 & 0.065 &  & 0.011 & 0.011 & 0.011 & 0.011 \\
 & ASE &  & 0.064 & 0.063 & 0.063 & 0.064 &  & 0.011 & 0.011 & 0.011 & 0.011 \\
\bottomrule
\end{tabular}}
\end{table}

\newpage

\begin{table}[th!]
\centering
\caption{Different estimation approaches applied to simulated count time series from INGARCH$(1,1)$-CMEM with Bin-counting series: Mean, simulated standard error (SSE), and mean approximated standard error (ASE) for different sample sizes~$n$.}
\label{tabDGP2Bin}

\smallskip
\resizebox{\linewidth}{!}{
\begin{tabular}{lrcrrrrcrrrr}
\toprule
$n$ &  & Param. & PQ & NQ & EQ & 2W & Param. & PQ & NQ & EQ & 2W \\
\midrule
300 & Mean & $a_0=3$ & 3.750 & 3.574 & 3.575 & 3.601 & $a_0=3$ & 4.131 & 4.256 & 4.267 & 4.169 \\
 & SSE &  & 1.412 & 1.234 & 1.235 & 1.369 &  & 1.217 & 1.212 & 1.213 & 1.181 \\
 & ASE &  & 1.310 & 1.070 & 1.069 & 1.077 &  & 1.100 & 1.080 & 1.082 & 1.064 \\[1ex]
600 & Mean &  & 3.376 & 3.272 & 3.275 & 3.279 &  & 3.587 & 3.687 & 3.696 & 3.654 \\
 & SSE &  & 0.887 & 0.748 & 0.749 & 0.789 &  & 0.767 & 0.768 & 0.770 & 0.746 \\
 & ASE &  & 0.880 & 0.705 & 0.704 & 0.706 &  & 0.733 & 0.715 & 0.716 & 0.710 \\[1ex]
1000 & Mean &  & 3.234 & 3.165 & 3.167 & 3.165 &  & 3.359 & 3.440 & 3.448 & 3.427 \\
 & SSE &  & 0.654 & 0.542 & 0.543 & 0.553 &  & 0.562 & 0.571 & 0.573 & 0.549 \\
 & ASE &  & 0.671 & 0.532 & 0.532 & 0.532 &  & 0.554 & 0.538 & 0.539 & 0.536 \\
\midrule
300 & Mean & $a_1=0.35$ & 0.335 & 0.348 & 0.349 & 0.346 & $a_1=0.35$ & 0.366 & 0.369 & 0.369 & 0.367 \\
 & SSE &  & 0.072 & 0.065 & 0.065 & 0.067 &  & 0.047 & 0.045 & 0.045 & 0.046 \\
 & ASE &  & 0.079 & 0.065 & 0.065 & 0.065 &  & 0.048 & 0.045 & 0.045 & 0.045 \\[1ex]
600 & Mean &  & 0.342 & 0.350 & 0.350 & 0.349 &  & 0.361 & 0.363 & 0.363 & 0.362 \\
 & SSE &  & 0.054 & 0.045 & 0.045 & 0.046 &  & 0.033 & 0.031 & 0.031 & 0.031 \\
 & ASE &  & 0.057 & 0.046 & 0.046 & 0.046 &  & 0.034 & 0.032 & 0.032 & 0.032 \\[1ex]
1000 & Mean &  & 0.345 & 0.350 & 0.350 & 0.350 &  & 0.357 & 0.359 & 0.359 & 0.359 \\
 & SSE &  & 0.042 & 0.035 & 0.035 & 0.035 &  & 0.026 & 0.024 & 0.024 & 0.024 \\
 & ASE &  & 0.045 & 0.035 & 0.035 & 0.035 &  & 0.026 & 0.025 & 0.025 & 0.025 \\
\midrule
300 & Mean & $b_1=0.5$ & 0.467 & 0.467 & 0.466 & 0.467 & $b_1=0.5$ & 0.427 & 0.419 & 0.418 & 0.425 \\
 & SSE &  & 0.110 & 0.095 & 0.095 & 0.105 &  & 0.077 & 0.074 & 0.074 & 0.075 \\
 & ASE &  & 0.106 & 0.089 & 0.089 & 0.089 &  & 0.078 & 0.076 & 0.076 & 0.075 \\[1ex]
600 & Mean &  & 0.484 & 0.484 & 0.484 & 0.484 &  & 0.460 & 0.453 & 0.452 & 0.455 \\
 & SSE &  & 0.073 & 0.061 & 0.061 & 0.064 &  & 0.053 & 0.051 & 0.051 & 0.050 \\
 & ASE &  & 0.073 & 0.060 & 0.060 & 0.060 &  & 0.054 & 0.053 & 0.053 & 0.052 \\[1ex]
1000 & Mean &  & 0.490 & 0.490 & 0.490 & 0.490 &  & 0.475 & 0.469 & 0.469 & 0.470 \\
 & SSE &  & 0.055 & 0.046 & 0.046 & 0.047 &  & 0.041 & 0.040 & 0.040 & 0.039 \\
 & ASE &  & 0.056 & 0.046 & 0.046 & 0.046 &  & 0.042 & 0.040 & 0.040 & 0.040 \\
\midrule
300 & Mean & $\sigma^2=1$ & 1.020 & 1.011 & 1.010 & 1.022 & $\sigma^2=0.1$ & 0.101 & 0.101 & 0.101 & 0.101 \\
 & SSE &  & 0.112 & 0.105 & 0.105 & 0.158 &  & 0.018 & 0.018 & 0.018 & 0.018 \\
 & ASE &  & 0.107 & 0.103 & 0.103 & 0.109 &  & 0.017 & 0.017 & 0.017 & 0.017 \\[1ex]
600 & Mean &  & 1.013 & 1.007 & 1.007 & 1.010 &  & 0.101 & 0.101 & 0.101 & 0.101 \\
 & SSE &  & 0.081 & 0.075 & 0.075 & 0.087 &  & 0.013 & 0.013 & 0.013 & 0.013 \\
 & ASE &  & 0.076 & 0.073 & 0.073 & 0.075 &  & 0.012 & 0.012 & 0.012 & 0.012 \\[1ex]
1000 & Mean &  & 1.007 & 1.004 & 1.003 & 1.005 &  & 0.100 & 0.100 & 0.100 & 0.100 \\
 & SSE &  & 0.058 & 0.057 & 0.056 & 0.060 &  & 0.010 & 0.010 & 0.010 & 0.010 \\
 & ASE &  & 0.057 & 0.056 & 0.056 & 0.057 &  & 0.010 & 0.010 & 0.009 & 0.010 \\
\bottomrule
\end{tabular}}
\end{table}

\newpage

\begin{table}[th!]
\centering
\caption{Different estimation approaches applied to simulated count time series from INGARCH$(1,1)$-CMEM with Poi-counting series: Mean, simulated standard error (SSE), and mean approximated standard error (ASE) for different sample sizes~$n$.}
\label{tabDGP3Poi}

\smallskip
\resizebox{\linewidth}{!}{
\begin{tabular}{lrcrrrrcrrrr}
\toprule
$n$ &  & Param. & PQ & NQ & EQ & 2W & Param. & PQ & NQ & EQ & 2W \\
\midrule
300 & Mean & $a_0=1$ & 1.452 & 1.435 & 1.446 & 1.474 & $a_0=1$ & 1.474 & 1.527 & 1.544 & 1.481 \\
 & SSE &  & 0.738 & 0.695 & 0.699 & 0.857 &  & 0.640 & 0.632 & 0.638 & 0.660 \\
 & ASE &  & 0.612 & 0.541 & 0.544 & 0.556 &  & 0.549 & 0.535 & 0.541 & 0.523 \\[1ex]
600 & Mean &  & 1.208 & 1.206 & 1.217 & 1.211 &  & 1.225 & 1.268 & 1.283 & 1.238 \\
 & SSE &  & 0.388 & 0.361 & 0.367 & 0.416 &  & 0.355 & 0.356 & 0.364 & 0.350 \\
 & ASE &  & 0.382 & 0.333 & 0.335 & 0.334 &  & 0.342 & 0.330 & 0.334 & 0.325 \\[1ex]
1000 & Mean &  & 1.127 & 1.126 & 1.135 & 1.123 &  & 1.130 & 1.163 & 1.175 & 1.145 \\
 & SSE &  & 0.279 & 0.257 & 0.263 & 0.264 &  & 0.256 & 0.255 & 0.261 & 0.249 \\
 & ASE &  & 0.283 & 0.244 & 0.246 & 0.244 &  & 0.252 & 0.241 & 0.244 & 0.239 \\
\midrule
300 & Mean & $a_1=0.25$ & 0.246 & 0.257 & 0.258 & 0.255 & $a_1=0.25$ & 0.256 & 0.263 & 0.264 & 0.261 \\
 & SSE &  & 0.061 & 0.058 & 0.058 & 0.061 &  & 0.051 & 0.049 & 0.049 & 0.050 \\
 & ASE &  & 0.067 & 0.059 & 0.059 & 0.059 &  & 0.053 & 0.051 & 0.051 & 0.050 \\[1ex]
600 & Mean &  & 0.248 & 0.255 & 0.256 & 0.255 &  & 0.254 & 0.260 & 0.261 & 0.258 \\
 & SSE &  & 0.044 & 0.040 & 0.040 & 0.041 &  & 0.036 & 0.034 & 0.034 & 0.034 \\
 & ASE &  & 0.048 & 0.041 & 0.041 & 0.041 &  & 0.037 & 0.035 & 0.036 & 0.035 \\[1ex]
1000 & Mean &  & 0.248 & 0.253 & 0.254 & 0.253 &  & 0.252 & 0.256 & 0.257 & 0.255 \\
 & SSE &  & 0.035 & 0.031 & 0.031 & 0.031 &  & 0.027 & 0.026 & 0.026 & 0.026 \\
 & ASE &  & 0.037 & 0.032 & 0.032 & 0.032 &  & 0.029 & 0.027 & 0.027 & 0.027 \\
\midrule
300 & Mean & $b_1=0.65$ & 0.600 & 0.593 & 0.591 & 0.590 & $b_1=0.65$ & 0.594 & 0.582 & 0.580 & 0.589 \\
 & SSE &  & 0.114 & 0.104 & 0.105 & 0.123 &  & 0.090 & 0.085 & 0.086 & 0.091 \\
 & ASE &  & 0.102 & 0.091 & 0.092 & 0.093 &  & 0.087 & 0.084 & 0.085 & 0.083 \\[1ex]
600 & Mean &  & 0.628 & 0.622 & 0.620 & 0.622 &  & 0.622 & 0.613 & 0.611 & 0.618 \\
 & SSE &  & 0.066 & 0.059 & 0.059 & 0.067 &  & 0.055 & 0.053 & 0.054 & 0.053 \\
 & ASE &  & 0.067 & 0.059 & 0.059 & 0.059 &  & 0.057 & 0.055 & 0.055 & 0.054 \\[1ex]
1000 & Mean &  & 0.637 & 0.633 & 0.631 & 0.633 &  & 0.634 & 0.627 & 0.625 & 0.630 \\
 & SSE &  & 0.049 & 0.044 & 0.044 & 0.045 &  & 0.042 & 0.040 & 0.041 & 0.040 \\
 & ASE &  & 0.051 & 0.044 & 0.044 & 0.044 &  & 0.043 & 0.041 & 0.042 & 0.041 \\
\midrule
300 & Mean & $\sigma^2=1$ & 1.039 & 1.020 & 1.018 & 1.051 & $\sigma^2=0.4$ & 0.414 & 0.408 & 0.407 & 0.415 \\
 & SSE &  & 0.155 & 0.125 & 0.123 & 0.258 &  & 0.057 & 0.050 & 0.050 & 0.072 \\
 & ASE &  & 0.143 & 0.131 & 0.130 & 0.153 &  & 0.048 & 0.044 & 0.044 & 0.050 \\[1ex]
600 & Mean &  & 1.026 & 1.015 & 1.013 & 1.025 &  & 0.410 & 0.406 & 0.405 & 0.409 \\
 & SSE &  & 0.111 & 0.090 & 0.088 & 0.147 &  & 0.039 & 0.035 & 0.035 & 0.041 \\
 & ASE &  & 0.100 & 0.094 & 0.093 & 0.101 &  & 0.034 & 0.032 & 0.031 & 0.034 \\[1ex]
1000 & Mean &  & 1.017 & 1.011 & 1.010 & 1.014 &  & 0.407 & 0.405 & 0.404 & 0.406 \\
 & SSE &  & 0.081 & 0.070 & 0.068 & 0.090 &  & 0.029 & 0.027 & 0.027 & 0.029 \\
 & ASE &  & 0.076 & 0.072 & 0.072 & 0.075 &  & 0.026 & 0.024 & 0.024 & 0.025 \\
\bottomrule
\end{tabular}}
\end{table}

\newpage

\begin{table}[th!]
\centering
\caption{Different estimation approaches applied to simulated count time series from INGARCH$(1,1)$-CMEM with Bin-counting series: Mean, simulated standard error (SSE), and mean approximated standard error (ASE) for different sample sizes~$n$.}
\label{tabDGP3Bin}

\smallskip
\resizebox{\linewidth}{!}{
\begin{tabular}{lrcrrrrcrrrr}
\toprule
$n$ &  & Param. & PQ & NQ & EQ & 2W & Param. & PQ & NQ & EQ & 2W \\
\midrule
300 & Mean & $a_0=1$ & 1.477 & 1.460 & 1.470 & 1.496 & $a_0=1$ & 1.517 & 1.581 & 1.600 & 1.531 \\
 & SSE &  & 0.771 & 0.712 & 0.712 & 0.884 &  & 0.642 & 0.637 & 0.644 & 0.674 \\
 & ASE &  & 0.619 & 0.540 & 0.541 & 0.554 &  & 0.546 & 0.524 & 0.528 & 0.512 \\[1ex]
600 & Mean &  & 1.215 & 1.214 & 1.226 & 1.230 &  & 1.247 & 1.300 & 1.316 & 1.274 \\
 & SSE &  & 0.396 & 0.361 & 0.367 & 0.447 &  & 0.362 & 0.365 & 0.373 & 0.369 \\
 & ASE &  & 0.383 & 0.328 & 0.330 & 0.331 &  & 0.340 & 0.323 & 0.326 & 0.319 \\[1ex]
1000 & Mean &  & 1.133 & 1.133 & 1.142 & 1.135 &  & 1.142 & 1.180 & 1.193 & 1.169 \\
 & SSE &  & 0.283 & 0.255 & 0.260 & 0.274 &  & 0.255 & 0.253 & 0.259 & 0.253 \\
 & ASE &  & 0.285 & 0.241 & 0.242 & 0.242 &  & 0.250 & 0.235 & 0.236 & 0.233 \\
\midrule
300 & Mean & $a_1=0.25$ & 0.246 & 0.257 & 0.258 & 0.255 & $a_1=0.25$ & 0.257 & 0.264 & 0.265 & 0.262 \\
 & SSE &  & 0.060 & 0.055 & 0.056 & 0.059 &  & 0.049 & 0.047 & 0.047 & 0.048 \\
 & ASE &  & 0.066 & 0.058 & 0.058 & 0.058 &  & 0.051 & 0.048 & 0.048 & 0.048 \\[1ex]
600 & Mean &  & 0.248 & 0.255 & 0.256 & 0.255 &  & 0.255 & 0.260 & 0.261 & 0.259 \\
 & SSE &  & 0.044 & 0.039 & 0.039 & 0.040 &  & 0.034 & 0.031 & 0.031 & 0.032 \\
 & ASE &  & 0.047 & 0.040 & 0.040 & 0.040 &  & 0.036 & 0.033 & 0.034 & 0.033 \\[1ex]
1000 & Mean &  & 0.248 & 0.253 & 0.254 & 0.254 &  & 0.254 & 0.258 & 0.258 & 0.257 \\
 & SSE &  & 0.035 & 0.030 & 0.030 & 0.031 &  & 0.027 & 0.025 & 0.025 & 0.025 \\
 & ASE &  & 0.037 & 0.031 & 0.031 & 0.031 &  & 0.028 & 0.026 & 0.026 & 0.026 \\
\midrule
300 & Mean & $b_1=0.65$ & 0.598 & 0.591 & 0.589 & 0.588 & $b_1=0.65$ & 0.589 & 0.576 & 0.574 & 0.583 \\
 & SSE &  & 0.114 & 0.102 & 0.102 & 0.123 &  & 0.087 & 0.083 & 0.083 & 0.090 \\
 & ASE &  & 0.102 & 0.090 & 0.090 & 0.092 &  & 0.084 & 0.081 & 0.081 & 0.079 \\[1ex]
600 & Mean &  & 0.628 & 0.622 & 0.620 & 0.619 &  & 0.619 & 0.609 & 0.607 & 0.613 \\
 & SSE &  & 0.066 & 0.058 & 0.059 & 0.069 &  & 0.054 & 0.052 & 0.052 & 0.054 \\
 & ASE &  & 0.066 & 0.058 & 0.058 & 0.058 &  & 0.055 & 0.053 & 0.053 & 0.052 \\[1ex]
1000 & Mean &  & 0.636 & 0.632 & 0.631 & 0.632 &  & 0.632 & 0.624 & 0.623 & 0.626 \\
 & SSE &  & 0.049 & 0.043 & 0.043 & 0.046 &  & 0.040 & 0.039 & 0.039 & 0.039 \\
 & ASE &  & 0.050 & 0.043 & 0.043 & 0.043 &  & 0.041 & 0.039 & 0.040 & 0.039 \\
\midrule
300 & Mean & $\sigma^2=1$ & 1.040 & 1.022 & 1.020 & 1.067 & $\sigma^2=0.4$ & 0.413 & 0.409 & 0.408 & 0.415 \\
 & SSE &  & 0.144 & 0.115 & 0.113 & 0.525 &  & 0.047 & 0.044 & 0.043 & 0.059 \\
 & ASE &  & 0.121 & 0.110 & 0.108 & 0.137 &  & 0.035 & 0.032 & 0.032 & 0.036 \\[1ex]
600 & Mean &  & 1.024 & 1.014 & 1.012 & 1.025 &  & 0.409 & 0.406 & 0.406 & 0.409 \\
 & SSE &  & 0.096 & 0.080 & 0.078 & 0.136 &  & 0.034 & 0.031 & 0.030 & 0.037 \\
 & ASE &  & 0.083 & 0.077 & 0.076 & 0.083 &  & 0.024 & 0.023 & 0.022 & 0.024 \\[1ex]
1000 & Mean &  & 1.016 & 1.010 & 1.009 & 1.014 &  & 0.406 & 0.404 & 0.404 & 0.405 \\
 & SSE &  & 0.072 & 0.061 & 0.060 & 0.083 &  & 0.026 & 0.024 & 0.023 & 0.026 \\
 & ASE &  & 0.063 & 0.060 & 0.059 & 0.062 &  & 0.018 & 0.017 & 0.017 & 0.018 \\
\bottomrule
\end{tabular}}
\end{table}

\newpage

\begin{table}[th!]
\centering
\caption{Different estimation approaches applied to simulated count time series from INGARCH$(1,1)$-CMEM with Poi-counting series, like in Table~\ref{tabDGP1Poi}, but by assuming a Bin-counting series (model misspecification): Mean, simulated standard error (SSE), and mean approximated standard error (ASE) for different sample sizes~$n$.}
\label{tabDGP1Poi_misp}

\smallskip
\resizebox{\linewidth}{!}{
\begin{tabular}{lrcrrrrcrrrr}
\toprule
$n$ &  & Param. & PQ & NQ & EQ & 2W & Param. & PQ & NQ & EQ & 2W \\
\midrule
300 & Mean & $a_0=2.8$ & 2.997 & 2.922 & 2.915 & 2.936 & $a_0=2.8$ & 2.965 & 2.936 & 2.932 & 2.937 \\
 & SSE &  & 0.722 & 0.686 & 0.688 & 0.696 &  & 0.681 & 0.668 & 0.671 & 0.669 \\
 & ASE &  & 0.790 & 0.688 & 0.685 & 0.689 &  & 0.727 & 0.673 & 0.671 & 0.672 \\[1ex]
600 & Mean &  & 2.904 & 2.858 & 2.854 & 2.865 &  & 2.889 & 2.875 & 2.874 & 2.875 \\
 & SSE &  & 0.529 & 0.492 & 0.493 & 0.498 &  & 0.488 & 0.474 & 0.476 & 0.475 \\
 & ASE &  & 0.544 & 0.477 & 0.475 & 0.476 &  & 0.503 & 0.468 & 0.467 & 0.467 \\[1ex]
1000 & Mean &  & 2.871 & 2.839 & 2.837 & 2.841 &  & 2.855 & 2.844 & 2.843 & 2.843 \\
 & SSE &  & 0.406 & 0.375 & 0.376 & 0.379 &  & 0.380 & 0.366 & 0.367 & 0.366 \\
 & ASE &  & 0.419 & 0.367 & 0.366 & 0.367 &  & 0.387 & 0.360 & 0.359 & 0.359 \\
\midrule
300 & Mean & $a_1=0.4$ & 0.384 & 0.393 & 0.394 & 0.393 & $a_1=0.4$ & 0.392 & 0.395 & 0.396 & 0.396 \\
 & SSE &  & 0.081 & 0.077 & 0.077 & 0.077 &  & 0.062 & 0.060 & 0.060 & 0.060 \\
 & ASE &  & 0.085 & 0.075 & 0.075 & 0.075 &  & 0.064 & 0.059 & 0.059 & 0.059 \\[1ex]
600 & Mean &  & 0.392 & 0.397 & 0.397 & 0.397 &  & 0.397 & 0.398 & 0.399 & 0.398 \\
 & SSE &  & 0.057 & 0.053 & 0.053 & 0.053 &  & 0.044 & 0.042 & 0.043 & 0.042 \\
 & ASE &  & 0.061 & 0.053 & 0.053 & 0.053 &  & 0.045 & 0.042 & 0.042 & 0.042 \\[1ex]
1000 & Mean &  & 0.395 & 0.399 & 0.399 & 0.399 &  & 0.398 & 0.399 & 0.399 & 0.399 \\
 & SSE &  & 0.045 & 0.041 & 0.041 & 0.041 &  & 0.034 & 0.033 & 0.033 & 0.033 \\
 & ASE &  & 0.048 & 0.042 & 0.041 & 0.041 &  & 0.035 & 0.033 & 0.033 & 0.033 \\
\midrule
300 & Mean & $b_1=0.2$ & 0.180 & 0.184 & 0.185 & 0.182 & $b_1=0.2$ & 0.182 & 0.183 & 0.184 & 0.183 \\
 & SSE &  & 0.121 & 0.116 & 0.117 & 0.118 &  & 0.112 & 0.110 & 0.111 & 0.110 \\
 & ASE &  & 0.130 & 0.116 & 0.115 & 0.116 &  & 0.120 & 0.112 & 0.112 & 0.112 \\[1ex]
600 & Mean &  & 0.189 & 0.192 & 0.193 & 0.191 &  & 0.190 & 0.190 & 0.190 & 0.190 \\
 & SSE &  & 0.089 & 0.083 & 0.084 & 0.084 &  & 0.080 & 0.079 & 0.079 & 0.079 \\
 & ASE &  & 0.090 & 0.081 & 0.080 & 0.080 &  & 0.083 & 0.078 & 0.078 & 0.078 \\[1ex]
1000 & Mean &  & 0.192 & 0.194 & 0.195 & 0.194 &  & 0.194 & 0.194 & 0.194 & 0.194 \\
 & SSE &  & 0.068 & 0.064 & 0.064 & 0.064 &  & 0.063 & 0.061 & 0.062 & 0.061 \\
 & ASE &  & 0.069 & 0.062 & 0.062 & 0.062 &  & 0.064 & 0.060 & 0.060 & 0.060 \\
\midrule
300 & Mean & $\sigma^2=1$ & 1.174 & 1.171 & 1.171 & 1.173 & $\sigma^2=0.4$ & 0.556 & 0.556 & 0.556 & 0.556 \\
 & SSE &  & 0.124 & 0.124 & 0.124 & 0.125 &  & 0.058 & 0.058 & 0.058 & 0.058 \\
 & ASE &  & 0.128 & 0.127 & 0.127 & 0.128 &  & 0.046 & 0.046 & 0.046 & 0.046 \\[1ex]
600 & Mean &  & 1.173 & 1.172 & 1.172 & 1.172 &  & 0.557 & 0.557 & 0.557 & 0.557 \\
 & SSE &  & 0.087 & 0.087 & 0.087 & 0.087 &  & 0.040 & 0.040 & 0.040 & 0.040 \\
 & ASE &  & 0.092 & 0.091 & 0.091 & 0.091 &  & 0.033 & 0.033 & 0.033 & 0.033 \\[1ex]
1000 & Mean &  & 1.172 & 1.171 & 1.171 & 1.171 &  & 0.557 & 0.557 & 0.557 & 0.557 \\
 & SSE &  & 0.067 & 0.067 & 0.067 & 0.067 &  & 0.031 & 0.031 & 0.031 & 0.031 \\
 & ASE &  & 0.071 & 0.071 & 0.071 & 0.071 &  & 0.025 & 0.025 & 0.025 & 0.025 \\
\bottomrule
\end{tabular}}
\end{table}

\newpage

\begin{table}[th!]
\centering
\caption{Different estimation approaches applied to simulated count time series from INGARCH$(1,1)$-CMEM with Bin-counting series, like in Table~\ref{tabDGP1Bin}, but by assuming a Poi-counting series (model misspecification): Mean, simulated standard error (SSE), and mean approximated standard error (ASE) for different sample sizes~$n$.}
\label{tabDGP1Bin_misp}

\smallskip
\resizebox{\linewidth}{!}{
\begin{tabular}{lrcrrrrcrrrr}
\toprule
$n$ &  & Param. & PQ & NQ & EQ & 2W & Param. & PQ & NQ & EQ & 2W \\
\midrule
300 & Mean & $a_0=2.8$ & 2.995 & 2.926 & 2.920 & 2.938 & $a_0=2.8$ & 2.967 & 2.942 & 2.939 & 2.948 \\
 & SSE &  & 0.713 & 0.667 & 0.667 & 0.684 &  & 0.655 & 0.629 & 0.629 & 0.633 \\
 & ASE &  & 0.757 & 0.679 & 0.678 & 0.683 &  & 0.689 & 0.669 & 0.671 & 0.668 \\[1ex]
600 & Mean &  & 2.913 & 2.869 & 2.865 & 2.876 &  & 2.892 & 2.879 & 2.878 & 2.881 \\
 & SSE &  & 0.520 & 0.477 & 0.477 & 0.485 &  & 0.479 & 0.454 & 0.453 & 0.457 \\
 & ASE &  & 0.522 & 0.471 & 0.471 & 0.472 &  & 0.476 & 0.464 & 0.466 & 0.463 \\[1ex]
1000 & Mean &  & 2.875 & 2.844 & 2.842 & 2.850 &  & 2.865 & 2.856 & 2.855 & 2.858 \\
 & SSE &  & 0.406 & 0.369 & 0.368 & 0.372 &  & 0.366 & 0.347 & 0.347 & 0.349 \\
 & ASE &  & 0.400 & 0.361 & 0.362 & 0.362 &  & 0.366 & 0.357 & 0.358 & 0.356 \\
\midrule
300 & Mean & $a_1=0.4$ & 0.384 & 0.393 & 0.394 & 0.392 & $a_1=0.4$ & 0.392 & 0.394 & 0.395 & 0.394 \\
 & SSE &  & 0.077 & 0.072 & 0.072 & 0.074 &  & 0.056 & 0.054 & 0.054 & 0.054 \\
 & ASE &  & 0.076 & 0.069 & 0.069 & 0.069 &  & 0.054 & 0.053 & 0.053 & 0.053 \\[1ex]
600 & Mean &  & 0.391 & 0.396 & 0.396 & 0.395 &  & 0.396 & 0.398 & 0.398 & 0.397 \\
 & SSE &  & 0.056 & 0.050 & 0.050 & 0.051 &  & 0.040 & 0.038 & 0.038 & 0.038 \\
 & ASE &  & 0.054 & 0.049 & 0.049 & 0.049 &  & 0.038 & 0.037 & 0.038 & 0.037 \\[1ex]
1000 & Mean &  & 0.395 & 0.398 & 0.398 & 0.398 &  & 0.398 & 0.399 & 0.399 & 0.399 \\
 & SSE &  & 0.043 & 0.038 & 0.038 & 0.039 &  & 0.030 & 0.029 & 0.029 & 0.029 \\
 & ASE &  & 0.042 & 0.038 & 0.038 & 0.038 &  & 0.030 & 0.029 & 0.029 & 0.029 \\
\midrule
300 & Mean & $b_1=0.2$ & 0.181 & 0.185 & 0.185 & 0.184 & $b_1=0.2$ & 0.181 & 0.183 & 0.183 & 0.182 \\
 & SSE &  & 0.118 & 0.112 & 0.112 & 0.115 &  & 0.106 & 0.103 & 0.103 & 0.103 \\
 & ASE &  & 0.122 & 0.112 & 0.112 & 0.112 &  & 0.112 & 0.109 & 0.109 & 0.109 \\[1ex]
600 & Mean &  & 0.189 & 0.192 & 0.192 & 0.191 &  & 0.189 & 0.190 & 0.190 & 0.190 \\
 & SSE &  & 0.086 & 0.080 & 0.080 & 0.082 &  & 0.078 & 0.075 & 0.075 & 0.075 \\
 & ASE &  & 0.084 & 0.078 & 0.078 & 0.078 &  & 0.078 & 0.076 & 0.076 & 0.076 \\[1ex]
1000 & Mean &  & 0.192 & 0.194 & 0.194 & 0.193 &  & 0.192 & 0.193 & 0.193 & 0.193 \\
 & SSE &  & 0.067 & 0.062 & 0.062 & 0.063 &  & 0.059 & 0.057 & 0.057 & 0.057 \\
 & ASE &  & 0.065 & 0.060 & 0.060 & 0.060 &  & 0.060 & 0.058 & 0.059 & 0.058 \\
\midrule
300 & Mean & $\sigma^2=1$ & 0.835 & 0.831 & 0.831 & 0.834 & $\sigma^2=0.4$ & 0.247 & 0.246 & 0.246 & 0.246 \\
 & SSE &  & 0.091 & 0.091 & 0.091 & 0.105 &  & 0.034 & 0.034 & 0.034 & 0.034 \\
 & ASE &  & 0.100 & 0.099 & 0.099 & 0.100 &  & 0.029 & 0.029 & 0.029 & 0.029 \\[1ex]
600 & Mean &  & 0.834 & 0.832 & 0.832 & 0.832 &  & 0.246 & 0.246 & 0.246 & 0.246 \\
 & SSE &  & 0.065 & 0.065 & 0.065 & 0.065 &  & 0.023 & 0.023 & 0.023 & 0.023 \\
 & ASE &  & 0.071 & 0.071 & 0.071 & 0.071 &  & 0.021 & 0.021 & 0.021 & 0.021 \\[1ex]
1000 & Mean &  & 0.834 & 0.832 & 0.832 & 0.832 &  & 0.246 & 0.246 & 0.246 & 0.246 \\
 & SSE &  & 0.049 & 0.049 & 0.049 & 0.049 &  & 0.018 & 0.018 & 0.018 & 0.018 \\
 & ASE &  & 0.055 & 0.055 & 0.055 & 0.055 &  & 0.016 & 0.016 & 0.016 & 0.016 \\
\bottomrule
\end{tabular}}
\end{table}

\newpage

\begin{table}[th!]
\centering
\caption{Mean MAR values for different estimation approaches applied to simulated count time series from INGARCH$(1,1)$-CMEM, where correct (corr) vs.\ misspecified (misp) model is assumed.
Refers to simulations from Table~\ref{tabDGP1Poi} vs.\ Table~\ref{tabDGP1Poi_misp} (Poi-counting series, misspecified as Bin-counting series), and from Table~\ref{tabDGP1Bin} vs.\ Table~\ref{tabDGP1Bin_misp} (Bin-counting series, misspecified as Poi-counting series).}
\label{tabDGP1_misp_MAR}

\smallskip
\resizebox{\linewidth}{!}{
\begin{tabular}{lrcrrrrcrrrr}
\toprule
$n$ &  & Param. & PQ & NQ & EQ & 2W & Param. & PQ & NQ & EQ & 2W \\
\midrule
\multicolumn{12}{l}{Simulations from Table~\ref{tabDGP1Poi} vs.\ Table~\ref{tabDGP1Poi_misp}} \\
\midrule
300 & corr & $\sigma^2=1$ & 5.885 & 5.892 & 5.893 & 5.896 & $\sigma^2=0.4$ & 4.040 & 4.041 & 4.042 & 4.041 \\
 & misp &  & 5.885 & 5.892 & 5.893 & 5.893 &  & 4.040 & 4.041 & 4.042 & 4.042 \\[1ex]
600 & corr &  & 5.899 & 5.902 & 5.903 & 5.904 &  & 4.052 & 4.053 & 4.054 & 4.053 \\
 & misp &  & 5.899 & 5.902 & 5.903 & 5.902 &  & 4.052 & 4.053 & 4.054 & 4.053 \\[1ex]
1000 & corr &  & 5.904 & 5.906 & 5.906 & 5.906 &  & 4.055 & 4.055 & 4.055 & 4.055 \\
 & misp &  & 5.904 & 5.906 & 5.906 & 5.906 &  & 4.055 & 4.055 & 4.055 & 4.055 \\
\midrule
\multicolumn{12}{l}{Simulations from Table~\ref{tabDGP1Bin} vs.\ Table~\ref{tabDGP1Bin_misp}} \\
\midrule
300 & corr & $\sigma^2=1$ & 5.336 & 5.337 & 5.338 & 5.337 & $\sigma^2=0.4$ & 3.050 & 3.047 & 3.048 & 3.047 \\
 & misp &  & 5.336 & 5.337 & 5.338 & 5.340 &  & 3.050 & 3.047 & 3.048 & 3.047 \\[1ex]
600 & corr &  & 5.322 & 5.320 & 5.320 & 5.320 &  & 3.029 & 3.027 & 3.027 & 3.027 \\
 & misp &  & 5.322 & 5.320 & 5.320 & 5.322 &  & 3.029 & 3.027 & 3.027 & 3.027 \\[1ex]
1000 & corr &  & 5.309 & 5.306 & 5.307 & 5.307 &  & 3.018 & 3.017 & 3.017 & 3.017 \\
 & misp &  & 5.309 & 5.306 & 5.307 & 5.307 &  & 3.018 & 3.017 & 3.017 & 3.017 \\
\bottomrule
\end{tabular}}
\end{table}

\newpage

\begin{table}[h!]
\centering
\caption{Mean MSPR values for different estimation approaches applied to simulated count time series from INGARCH$(1,1)$-CMEM, where correct (corr) vs.\ misspecified (misp) model is assumed.
Refers to simulations from Table~\ref{tabDGP1Poi} vs.\ Table~\ref{tabDGP1Poi_misp} (Poi-counting series, misspecified as Bin-counting series), and from Table~\ref{tabDGP1Bin} vs.\ Table~\ref{tabDGP1Bin_misp} (Bin-counting series, misspecified as Poi-counting series).}
\label{tabDGP1_misp_MSPR}

\smallskip
\resizebox{\linewidth}{!}{
\begin{tabular}{lrcrrrrcrrrr}
\toprule
$n$ &  & Param. & PQ & NQ & EQ & 2W & Param. & PQ & NQ & EQ & 2W \\
\midrule
\multicolumn{12}{l}{Simulations from Table~\ref{tabDGP1Poi} vs.\ Table~\ref{tabDGP1Poi_misp}} \\
\midrule
300 & corr & $\sigma^2=1$ & 1.000 & 1.000 & 1.000 & 1.000 & $\sigma^2=0.4$ & 1.000 & 1.000 & 1.000 & 1.000 \\
 & misp &  & 1.000 & 1.000 & 1.000 & 1.000 &  & 0.999 & 0.999 & 0.999 & 0.999 \\[1ex]
600 & corr &  & 1.000 & 1.000 & 1.000 & 1.000 &  & 1.000 & 1.000 & 1.000 & 1.000 \\
 & misp &  & 1.000 & 1.000 & 1.000 & 1.000 &  & 0.999 & 0.999 & 0.999 & 0.999 \\[1ex]
1000 & corr &  & 1.000 & 1.000 & 1.000 & 1.000 &  & 1.000 & 1.000 & 1.000 & 1.000 \\
 & misp &  & 1.000 & 1.000 & 1.000 & 1.000 &  & 0.999 & 0.999 & 0.999 & 0.999 \\
\midrule
\multicolumn{12}{l}{Simulations from Table~\ref{tabDGP1Bin} vs.\ Table~\ref{tabDGP1Bin_misp}} \\
\midrule
300 & corr & $\sigma^2=1$ & 1.000 & 1.000 & 1.000 & 1.000 & $\sigma^2=0.4$ & 1.000 & 1.000 & 1.000 & 1.000 \\
 & misp &  & 1.004 & 1.004 & 1.004 & 1.004 &  & 1.013 & 1.014 & 1.014 & 1.014 \\[1ex]
600 & corr &  & 1.000 & 1.000 & 1.000 & 1.000 &  & 1.000 & 1.000 & 1.000 & 1.000 \\
 & misp &  & 1.004 & 1.004 & 1.004 & 1.004 &  & 1.014 & 1.014 & 1.014 & 1.014 \\[1ex]
1000 & corr &  & 1.000 & 1.000 & 1.000 & 1.000 &  & 1.000 & 1.000 & 1.000 & 1.000 \\
 & misp &  & 1.004 & 1.004 & 1.004 & 1.004 &  & 1.014 & 1.014 & 1.014 & 1.014 \\
\bottomrule
\end{tabular}}
\end{table}

\newpage

\begin{table}[th!]
\centering
\caption{Different estimation approaches applied to simulated count time series from INGARCH$(1,1)$-CMEM with Poi-counting series, like in Table~\ref{tabDGP1Poi}, but by using softplus response function $s_2(x)$ for data generation (model misspecification): Mean, simulated standard error (SSE), and mean approximated standard error (ASE) for different sample sizes~$n$.}
\label{tabDGP1Poi_sp2}

\smallskip
\resizebox{\linewidth}{!}{
\begin{tabular}{lrcrrrrcrrrr}
\toprule
$n$ &  & Param. & PQ & NQ & EQ & 2W & Param. & PQ & NQ & EQ & 2W \\
\midrule
300 & Mean & $a_0=2.8$ & 3.252 & 3.201 & 3.197 & 3.215 & $a_0=2.8$ & 3.193 & 3.182 & 3.181 & 3.184 \\
 & SSE &  & 0.790 & 0.756 & 0.758 & 0.766 &  & 0.730 & 0.713 & 0.715 & 0.714 \\
 & ASE &  & 0.871 & 0.771 & 0.769 & 0.772 &  & 0.800 & 0.750 & 0.749 & 0.749 \\[1ex]
600 & Mean &  & 3.178 & 3.158 & 3.158 & 3.165 &  & 3.131 & 3.141 & 3.144 & 3.143 \\
 & SSE &  & 0.575 & 0.540 & 0.541 & 0.546 &  & 0.530 & 0.519 & 0.521 & 0.520 \\
 & ASE &  & 0.599 & 0.533 & 0.532 & 0.533 &  & 0.554 & 0.523 & 0.523 & 0.522 \\[1ex]
1000 & Mean &  & 3.124 & 3.125 & 3.126 & 3.129 &  & 3.099 & 3.111 & 3.113 & 3.113 \\
 & SSE &  & 0.449 & 0.421 & 0.422 & 0.424 &  & 0.422 & 0.409 & 0.410 & 0.410 \\
 & ASE &  & 0.460 & 0.410 & 0.410 & 0.410 &  & 0.426 & 0.402 & 0.402 & 0.402 \\
\midrule
300 & Mean & $a_1=0.4$ & 0.372 & 0.377 & 0.377 & 0.377 & $a_1=0.4$ & 0.377 & 0.378 & 0.378 & 0.378 \\
 & SSE &  & 0.081 & 0.076 & 0.076 & 0.076 &  & 0.062 & 0.060 & 0.060 & 0.060 \\
 & ASE &  & 0.084 & 0.074 & 0.074 & 0.074 &  & 0.063 & 0.059 & 0.059 & 0.059 \\[1ex]
600 & Mean &  & 0.381 & 0.382 & 0.382 & 0.383 &  & 0.382 & 0.382 & 0.382 & 0.382 \\
 & SSE &  & 0.057 & 0.052 & 0.052 & 0.052 &  & 0.044 & 0.042 & 0.042 & 0.042 \\
 & ASE &  & 0.060 & 0.053 & 0.053 & 0.053 &  & 0.045 & 0.042 & 0.042 & 0.042 \\[1ex]
1000 & Mean &  & 0.384 & 0.384 & 0.383 & 0.383 &  & 0.384 & 0.383 & 0.383 & 0.383 \\
 & SSE &  & 0.045 & 0.040 & 0.040 & 0.040 &  & 0.034 & 0.033 & 0.033 & 0.033 \\
 & ASE &  & 0.047 & 0.041 & 0.041 & 0.041 &  & 0.035 & 0.033 & 0.032 & 0.032 \\
\midrule
300 & Mean & $b_1=0.2$ & 0.176 & 0.179 & 0.179 & 0.177 & $b_1=0.2$ & 0.179 & 0.180 & 0.180 & 0.179 \\
 & SSE &  & 0.124 & 0.120 & 0.121 & 0.122 &  & 0.116 & 0.114 & 0.114 & 0.114 \\
 & ASE &  & 0.135 & 0.122 & 0.122 & 0.122 &  & 0.126 & 0.119 & 0.119 & 0.119 \\[1ex]
600 & Mean &  & 0.180 & 0.182 & 0.182 & 0.180 &  & 0.184 & 0.183 & 0.183 & 0.183 \\
 & SSE &  & 0.091 & 0.087 & 0.087 & 0.088 &  & 0.084 & 0.083 & 0.084 & 0.083 \\
 & ASE &  & 0.093 & 0.084 & 0.084 & 0.084 &  & 0.087 & 0.083 & 0.083 & 0.083 \\[1ex]
1000 & Mean &  & 0.187 & 0.187 & 0.186 & 0.186 &  & 0.188 & 0.187 & 0.187 & 0.187 \\
 & SSE &  & 0.070 & 0.067 & 0.067 & 0.068 &  & 0.067 & 0.065 & 0.065 & 0.065 \\
 & ASE &  & 0.071 & 0.065 & 0.065 & 0.065 &  & 0.067 & 0.064 & 0.064 & 0.064 \\
\midrule
300 & Mean & $\sigma^2=1$ & 1.165 & 1.162 & 1.162 & 1.163 & $\sigma^2=0.4$ & 0.550 & 0.549 & 0.549 & 0.549 \\
 & SSE &  & 0.123 & 0.122 & 0.122 & 0.123 &  & 0.057 & 0.057 & 0.057 & 0.057 \\
 & ASE &  & 0.127 & 0.126 & 0.126 & 0.126 &  & 0.045 & 0.045 & 0.045 & 0.045 \\[1ex]
600 & Mean &  & 1.166 & 1.164 & 1.164 & 1.164 &  & 0.550 & 0.550 & 0.550 & 0.550 \\
 & SSE &  & 0.086 & 0.086 & 0.086 & 0.086 &  & 0.040 & 0.040 & 0.040 & 0.040 \\
 & ASE &  & 0.091 & 0.090 & 0.090 & 0.090 &  & 0.032 & 0.032 & 0.032 & 0.032 \\[1ex]
1000 & Mean &  & 1.163 & 1.162 & 1.162 & 1.162 &  & 0.550 & 0.550 & 0.550 & 0.550 \\
 & SSE &  & 0.067 & 0.066 & 0.066 & 0.066 &  & 0.031 & 0.031 & 0.031 & 0.031 \\
 & ASE &  & 0.070 & 0.070 & 0.070 & 0.070 &  & 0.025 & 0.025 & 0.025 & 0.025 \\
\bottomrule
\end{tabular}}
\end{table}

\newpage

\begin{table}[th!]
\centering
\caption{Different estimation approaches applied to simulated count time series from INGARCH$(1,1)$-CMEM with Bin-counting series, like in Table~\ref{tabDGP1Bin}, but by using softplus response function $s_2(x)$ for data generation (model misspecification): Mean, simulated standard error (SSE), and mean approximated standard error (ASE) for different sample sizes~$n$.}
\label{tabDGP1Bin_sp2}

\smallskip
\resizebox{\linewidth}{!}{
\begin{tabular}{lrcrrrrcrrrr}
\toprule
$n$ &  & Param. & PQ & NQ & EQ & 2W & Param. & PQ & NQ & EQ & 2W \\
\midrule
300 & Mean & $a_0=2.8$ & 3.248 & 3.204 & 3.202 & 3.214 & $a_0=2.8$ & 3.198 & 3.194 & 3.195 & 3.194 \\
 & SSE &  & 0.772 & 0.729 & 0.730 & 0.746 &  & 0.713 & 0.688 & 0.688 & 0.692 \\
 & ASE &  & 0.836 & 0.760 & 0.760 & 0.763 &  & 0.763 & 0.746 & 0.749 & 0.744 \\[1ex]
600 & Mean &  & 3.165 & 3.149 & 3.150 & 3.154 &  & 3.120 & 3.131 & 3.134 & 3.126 \\
 & SSE &  & 0.570 & 0.528 & 0.528 & 0.536 &  & 0.529 & 0.505 & 0.505 & 0.508 \\
 & ASE &  & 0.575 & 0.526 & 0.527 & 0.527 &  & 0.526 & 0.518 & 0.520 & 0.516 \\[1ex]
1000 & Mean &  & 3.125 & 3.124 & 3.126 & 3.126 &  & 3.091 & 3.106 & 3.109 & 3.101 \\
 & SSE &  & 0.447 & 0.410 & 0.410 & 0.413 &  & 0.404 & 0.386 & 0.386 & 0.387 \\
 & ASE &  & 0.440 & 0.404 & 0.405 & 0.404 &  & 0.404 & 0.398 & 0.400 & 0.396 \\
\midrule
300 & Mean & $a_1=0.4$ & 0.372 & 0.377 & 0.377 & 0.376 & $a_1=0.4$ & 0.376 & 0.377 & 0.377 & 0.377 \\
 & SSE &  & 0.078 & 0.072 & 0.072 & 0.073 &  & 0.057 & 0.055 & 0.055 & 0.055 \\
 & ASE &  & 0.075 & 0.068 & 0.068 & 0.068 &  & 0.054 & 0.053 & 0.053 & 0.053 \\[1ex]
600 & Mean &  & 0.379 & 0.380 & 0.380 & 0.380 &  & 0.381 & 0.380 & 0.380 & 0.380 \\
 & SSE &  & 0.056 & 0.050 & 0.050 & 0.051 &  & 0.040 & 0.038 & 0.038 & 0.038 \\
 & ASE &  & 0.054 & 0.048 & 0.048 & 0.048 &  & 0.039 & 0.038 & 0.038 & 0.038 \\[1ex]
1000 & Mean &  & 0.384 & 0.383 & 0.382 & 0.383 &  & 0.383 & 0.381 & 0.381 & 0.382 \\
 & SSE &  & 0.043 & 0.038 & 0.038 & 0.039 &  & 0.031 & 0.029 & 0.029 & 0.029 \\
 & ASE &  & 0.042 & 0.038 & 0.038 & 0.038 &  & 0.030 & 0.029 & 0.029 & 0.029 \\
\midrule
300 & Mean & $b_1=0.2$ & 0.176 & 0.178 & 0.178 & 0.177 & $b_1=0.2$ & 0.178 & 0.178 & 0.178 & 0.178 \\
 & SSE &  & 0.121 & 0.116 & 0.117 & 0.119 &  & 0.111 & 0.109 & 0.109 & 0.109 \\
 & ASE &  & 0.128 & 0.119 & 0.118 & 0.119 &  & 0.119 & 0.117 & 0.117 & 0.117 \\[1ex]
600 & Mean &  & 0.183 & 0.184 & 0.184 & 0.184 &  & 0.186 & 0.185 & 0.185 & 0.186 \\
 & SSE &  & 0.089 & 0.084 & 0.084 & 0.085 &  & 0.082 & 0.080 & 0.080 & 0.080 \\
 & ASE &  & 0.088 & 0.082 & 0.082 & 0.082 &  & 0.082 & 0.081 & 0.082 & 0.081 \\[1ex]
1000 & Mean &  & 0.186 & 0.187 & 0.187 & 0.186 &  & 0.189 & 0.188 & 0.188 & 0.188 \\
 & SSE &  & 0.070 & 0.065 & 0.066 & 0.066 &  & 0.063 & 0.061 & 0.061 & 0.061 \\
 & ASE &  & 0.067 & 0.063 & 0.063 & 0.063 &  & 0.063 & 0.063 & 0.063 & 0.062 \\
\midrule
300 & Mean & $\sigma^2=1$ & 0.842 & 0.839 & 0.839 & 0.841 & $\sigma^2=0.4$ & 0.253 & 0.253 & 0.253 & 0.253 \\
 & SSE &  & 0.092 & 0.091 & 0.091 & 0.096 &  & 0.035 & 0.035 & 0.035 & 0.035 \\
 & ASE &  & 0.100 & 0.099 & 0.099 & 0.100 &  & 0.029 & 0.029 & 0.029 & 0.029 \\[1ex]
600 & Mean &  & 0.842 & 0.840 & 0.840 & 0.841 &  & 0.252 & 0.252 & 0.252 & 0.252 \\
 & SSE &  & 0.065 & 0.065 & 0.065 & 0.066 &  & 0.024 & 0.024 & 0.024 & 0.024 \\
 & ASE &  & 0.071 & 0.071 & 0.071 & 0.071 &  & 0.021 & 0.020 & 0.020 & 0.020 \\[1ex]
1000 & Mean &  & 0.842 & 0.841 & 0.841 & 0.841 &  & 0.252 & 0.252 & 0.252 & 0.252 \\
 & SSE &  & 0.049 & 0.049 & 0.049 & 0.049 &  & 0.019 & 0.019 & 0.019 & 0.019 \\
 & ASE &  & 0.055 & 0.055 & 0.055 & 0.055 &  & 0.016 & 0.016 & 0.016 & 0.016 \\
\bottomrule
\end{tabular}}
\end{table}

\newpage

\begin{table}[th!]
\centering
\caption{Mean MAR values for different estimation approaches applied to simulated count time series from INGARCH$(1,1)$-CMEM, where linear (lin) or softplus (soft) response function $s_2(x)$ is used for data generation.
Refers to simulations from Table~\ref{tabDGP1Poi} vs.\ Table~\ref{tabDGP1Poi_sp2} (Poi-counting series), and from Table~\ref{tabDGP1Bin} vs.\ Table~\ref{tabDGP1Bin_sp2} (Bin-counting series).}
\label{tabDGP1_sp2_MAR}

\smallskip
\resizebox{\linewidth}{!}{
\begin{tabular}{lrcrrrrcrrrr}
\toprule
$n$ &  & Param. & PQ & NQ & EQ & 2W & Param. & PQ & NQ & EQ & 2W \\
\midrule
\multicolumn{12}{l}{Simulations from Table~\ref{tabDGP1Poi} vs.\ Table~\ref{tabDGP1Poi_sp2}} \\
\midrule
300 & lin & $\sigma^2=1$ & 5.885 & 5.892 & 5.893 & 5.896 & $\sigma^2=0.4$ & 4.040 & 4.041 & 4.042 & 4.041 \\
 & soft &  & 6.127 & 6.132 & 6.134 & 6.135 &  & 4.155 & 4.156 & 4.157 & 4.156 \\
600 & lin &  & 5.899 & 5.902 & 5.903 & 5.904 &  & 4.052 & 4.053 & 4.054 & 4.053 \\
 & soft &  & 6.151 & 6.153 & 6.153 & 6.153 &  & 4.170 & 4.171 & 4.171 & 4.171 \\
1000 & lin &  & 5.904 & 5.906 & 5.906 & 5.906 &  & 4.055 & 4.055 & 4.055 & 4.055 \\
 & soft &  & 6.157 & 6.157 & 6.157 & 6.157 &  & 4.176 & 4.176 & 4.176 & 4.176 \\
\midrule
\multicolumn{12}{l}{Simulations from Table~\ref{tabDGP1Bin} vs.\ Table~\ref{tabDGP1Bin_sp2}} \\
\midrule
300 & lin & $\sigma^2=1$ & 5.336 & 5.337 & 5.338 & 5.337 & $\sigma^2=0.4$ & 3.050 & 3.047 & 3.048 & 3.047 \\
 & soft &  & 5.561 & 5.559 & 5.560 & 5.563 &  & 3.146 & 3.144 & 3.144 & 3.144 \\
600 & lin &  & 5.322 & 5.320 & 5.320 & 5.320 &  & 3.029 & 3.027 & 3.027 & 3.027 \\
 & soft &  & 5.546 & 5.540 & 5.541 & 5.542 &  & 3.127 & 3.125 & 3.125 & 3.125 \\
1000 & lin &  & 5.309 & 5.306 & 5.307 & 5.307 &  & 3.018 & 3.017 & 3.017 & 3.017 \\
 & soft &  & 5.532 & 5.527 & 5.527 & 5.527 &  & 3.115 & 3.114 & 3.114 & 3.114 \\
\bottomrule
\end{tabular}}
\end{table}

\newpage

\begin{table}[h!]
\centering
\caption{Mean MSPR values for different estimation approaches applied to simulated count time series from INGARCH$(1,1)$-CMEM, where linear (lin) or softplus (soft) response function $s_2(x)$ is used for data generation.
Refers to simulations from Table~\ref{tabDGP1Poi} vs.\ Table~\ref{tabDGP1Poi_sp2} (Poi-counting series), and from Table~\ref{tabDGP1Bin} vs.\ Table~\ref{tabDGP1Bin_sp2} (Bin-counting series).}
\label{tabDGP1_sp2_MSPR}

\smallskip
\resizebox{\linewidth}{!}{
\begin{tabular}{lrcrrrrcrrrr}
\toprule
$n$ &  & Param. & PQ & NQ & EQ & 2W & Param. & PQ & NQ & EQ & 2W \\
\midrule
\multicolumn{12}{l}{Simulations from Table~\ref{tabDGP1Poi} vs.\ Table~\ref{tabDGP1Poi_sp2}} \\
\midrule
300 & lin & $\sigma^2=1$ & 1.000 & 1.000 & 1.000 & 1.000 & $\sigma^2=0.4$ & 1.000 & 1.000 & 1.000 & 1.000 \\
 & soft &  & 0.999 & 1.000 & 1.000 & 1.000 &  & 0.999 & 1.000 & 1.000 & 1.000 \\
600 & lin &  & 1.000 & 1.000 & 1.000 & 1.000 &  & 1.000 & 1.000 & 1.000 & 1.000 \\
 & soft &  & 1.000 & 1.000 & 1.000 & 1.000 &  & 0.999 & 1.000 & 1.000 & 1.000 \\
1000 & lin &  & 1.000 & 1.000 & 1.000 & 1.000 &  & 1.000 & 1.000 & 1.000 & 1.000 \\
 & soft &  & 0.999 & 1.000 & 1.000 & 1.000 &  & 1.000 & 1.000 & 1.000 & 1.000 \\
\midrule
\multicolumn{12}{l}{Simulations from Table~\ref{tabDGP1Bin} vs.\ Table~\ref{tabDGP1Bin_sp2}} \\
\midrule
300 & lin & $\sigma^2=1$ & 1.000 & 1.000 & 1.000 & 1.000 & $\sigma^2=0.4$ & 1.000 & 1.000 & 1.000 & 1.000 \\
 & soft &  & 1.003 & 1.003 & 1.003 & 1.003 &  & 1.011 & 1.011 & 1.011 & 1.011 \\
600 & lin &  & 1.000 & 1.000 & 1.000 & 1.000 &  & 1.000 & 1.000 & 1.000 & 1.000 \\
 & soft &  & 1.003 & 1.003 & 1.003 & 1.003 &  & 1.011 & 1.012 & 1.012 & 1.012 \\
1000 & lin &  & 1.000 & 1.000 & 1.000 & 1.000 &  & 1.000 & 1.000 & 1.000 & 1.000 \\
 & soft &  & 1.003 & 1.003 & 1.003 & 1.003 &  & 1.011 & 1.012 & 1.012 & 1.012 \\
\bottomrule
\end{tabular}}
\end{table}

\end{document}